\documentclass[11pt]{article}
\usepackage{fullpage}
\usepackage{times}
\usepackage{subfigure}
\usepackage{color}
\usepackage{url}
\usepackage{graphicx}
\usepackage{amsmath}
\usepackage{amsthm}
\usepackage{amssymb}
\usepackage{algorithm}
\usepackage{algpseudocode}
\usepackage{mathtools}
\usepackage{tikz}
\usetikzlibrary{positioning, calc, chains}

{\makeatletter
 \gdef\xxxmark{%
   \expandafter\ifx\csname @mpargs\endcsname\relax 
     \expandafter\ifx\csname @captype\endcsname\relax 
       \marginpar{xxx}
     \else
       xxx 
     \fi
   \else
     xxx 
   \fi}
 \gdef\xxx{\@ifnextchar[\xxx@lab\xxx@nolab}
 \long\gdef\xxx@lab[#1]#2{{\bf [\xxxmark #2 ---{\sc #1}]}}
 \long\gdef\xxx@nolab#1{{\bf [\xxxmark #1]}}
}

\newcommand{\abs}[1]{|#1|}

\newcommand{\norm}[2]{\lVert#2\rVert_{#1}}

\newcommand{\wh}{\widehat}

\newcommand{\twid}[1]{\tilde{#1}}

\DeclareMathOperator{\supp}{supp}

\def\R{\mathbb{R}}
\def\Z{\mathbb{Z}}
\def\C{\mathbb{C}}
\def\F{\mathcal{F}}
\def\B{\mathbb{B}}
\def\eps{\epsilon}

\newcommand{\expect}{{\bf \mbox{\bf E}}}
\newcommand{\prob}{{\bf \mbox{\bf Pr}}}

\definecolor{gray}{rgb}{0.5,0.5,0.5}
\newcommand{\e}{{\epsilon}}
\newcommand{\A}{{\mathcal{A}}}
\newcommand{\E}{{\mathcal{E}}}

\newtheorem{theorem}{Theorem}[section]
\newtheorem{lemma}[theorem]{Lemma}
\newtheorem{claim}[theorem]{Claim}
\newtheorem{corollary}[theorem]{Corollary}
\newtheorem{fact}[theorem]{Fact}
\newtheorem{definition}[theorem]{Definition}
\newtheorem{remark}[theorem]{Remark}

\newenvironment{proofof}[1]{\noindent{\bf Proof of #1:}}{$\qed$\par}
\usepackage{hyperref}
\hypersetup{
bookmarksnumbered
}

\DeclareMathOperator{\err}{Err}
\newcommand{\poly}{\text{poly}}

\newcommand{\one}{\mathbf{1}}

\newcommand{\fc}{F}
\newcommand{\gl}{\mathcal{M}_{d\times d}}

\renewcommand{\H}{\mathcal{W}}
\newcommand{\h}{\mathbf{w}}

\newcommand{\zero}{\mathbf{0}}

\newcommand{\defsk}{Let $S\subseteq \nsq, |S|\leq 2k$, be such that $||x_{\nsq\setminus S}||_\infty\leq \mu$. Suppose that $||x||_\infty/\mu\leq N^{O(1)}$.~Let $B\geq \GT\cdot k/\alpha^d$.~}

\newcommand{\quant}{\text{quant}}
\newcommand{\GS}{(2\pi)^{d\cdot \fc}}
\newcommand{\GSS}{(2\pi)^{2d\cdot \fc}}
\newcommand{\GT}{(2\pi)^{4d\cdot \fc}}
\newcommand{\GI}{(2\pi)^{-d\cdot \fc}}
\begin{document}

\begin{titlepage}
\title{Sparse Fourier Transform in Any Constant Dimension with Nearly-Optimal Sample Complexity in Sublinear Time}
\author{Michael Kapralov\\EPFL}

\maketitle

\begin{abstract}
We consider the problem of computing a $k$-sparse approximation to the Fourier transform of a length $N$ signal. Our main result is a randomized algorithm for computing such an approximation (i.e. achieving the $\ell_2/\ell_2$ sparse recovery guarantees using Fourier measurements) using $O_d(k\log N\log\log N)$ samples of the signal in time domain that runs in time $O_d(k\log^{d+3} N)$, where $d\geq 1$ is the dimensionality of the Fourier transform. The sample complexity matches the lower bound of $\Omega(k\log (N/k))$ for non-adaptive algorithms due to~\cite{DIPW} for any $k\leq N^{1-\delta}$ for a constant $\delta>0$ up to an $O(\log\log N)$ factor. Prior to our work a result with comparable sample complexity $k\log N \log^{O(1)}\log N$ and sublinear runtime was known for the Fourier transform on the line~\cite{IKP}, but for any dimension $d\geq 2$ previously known techniques either suffered from a $\poly(\log N)$ factor loss in sample complexity or required $\Omega(N)$ runtime.
\end{abstract}
\thispagestyle{empty}
\end{titlepage}

\tableofcontents

\newpage

\newcommand{\nsq}{{[n]^d}}


\section{Introduction}

The Discrete Fourier Transform (DFT) is a fundamental mathematical concept that allows to represent a discrete signal of length $N$ as a linear combination of $N$ pure harmonics, or frequencies. The development of a fast algorithm for Discrete Fourier Transform, known as FFT (Fast Fourier Transform) in 1965 revolutionized digital signal processing, earning FFT a place in the top 10 most important algorithms of the twentieth century~\cite{citeulike:6838680}. Fast Fourier Transform (FFT) computes the DFT of a length $N$ signal in time $O(N\log N)$, and finding a faster algorithm for DFT is a major open problem in theoretical computer science.  While FFT applies to general signals, many of the applications of FFT (e.g. image and video compression schemes such as JPEG and MPEG) rely on the fact that  the Fourier spectrum of signals that arise in practice can often be approximated very well by only a few of the top Fourier coefficients, i.e. practical signals are often (approximately) {\em sparse} in the Fourier basis. 


Besides applications in signal processing, the Fourier sparsity property of real world signal plays and important role in medical imaging, where the cost of {\em measuring a signal}, i.e. {\em sample complexity}, is often a major bottleneck.  For example, an MRI machine effectively measures the Fourier transform of a signal $x$ representing the object being scanned, and the reconstruction problem is exactly the problem of inverting the Fourier transform $\wh{x}$ of $x$ approximately given a set of measurements. Minimizing the sample complexity of acquiring a signal using Fourier measurements thus translates directly to reduction in the time the patient spends in the MRI machine~\cite{MRICS} while a scan is being taken. In applications to Computed Tomography (CT) reduction in measurement cost leads to reduction in the radiation dose that a patient receives~\cite{CT_CS}.  Because of this strong practical motivation, 
the problem of computing a good approximation to the FFT of a Fourier sparse signal fast and using few measurements in time domain has been the subject of much attention several communities. 
In the area of {\em compressive sensing}~\cite{Don,CTao}, where one studies the task of recovering (approximately) sparse signals from linear measurements,  Fourier measurements have been one of the key settings of interest. In particular, the seminal work of~\cite{CTao,RV} has shown that length $N$ signals with at most $k$ nonzero Fourier coefficients can be recovered using only $k \log^{O(1)} N$ samples in time domain. The recovery algorithms are based on linear programming and run in time polynomial in $N$. 
A different line of research on the {\em Sparse Fourier Transform} (Sparse FFT), initiated in the fields of computational complexity and learning theory, has been focused on  developing algorithms whose sample complexity {\em and} running time scale with the sparsity as opposed to the length of the input signal. Many such algorithms have been proposed in the literature, including \cite{GL,KM,Man,GGIMS,AGS,GMS,Iw,Ak,HIKP,HIKP2,BCGLS,HAKI,pawar2013computing,heidersparse, IKP}. These works show that, for a wide range of signals, both the time complexity and the number of signal samples taken can be significantly sub-linear in $N$, often of the form $k \log^{O(1)} N$.  

In this paper we consider the problem of computing a sparse approximation to a signal $x\in \mathbb{C}^N$ given access to its Fourier transform $\wh{x}\in \mathbb{C}^N$.\footnote{Note that the problem of reconstructing a signal from Fourier measurements is equivalent to the problem of computing the Fourier transform of a signal $x$ whose spectrum is approximately sparse, as the DFT and its inverse are only different by a conjugation.} The best known results obtained in both compressive sensing literature and sparse FFT literature on this problem are summarized in Fig.~\ref{fig:1}. 
We focus on algorithms that work for worst-case signals and recover $k$-sparse approximations satisfying the so-called $\ell_2/\ell_2$ approximation
guarantee. In this case, the goal of an algorithm is as follows: given $m$ samples of the Fourier transform $\wh{x}$ of a signal $x$,  and the sparsity parameter $k$, output $x'$ satisfying
\begin{equation}
\label{e:l2l2}
\| x-x'\|_2 \le C \min_{k \text{-sparse } y }  \|x-y\|_2,
\end{equation}
The algorithms are randomized\footnote{Some of the algorithms~\cite{CTao,RV,CGV} can in
  fact be made deterministic, but at the cost of satisfying a somewhat
  weaker $\ell_2/\ell_1$ guarantee.}
 and succeed with at least constant probability.
 
 \begin{figure*}[h]
\begin{center}
\begin{tabular}{|c|c|c|c|c|}
\hline
Reference & Time & Samples & $C$  & Dimension\\
&  & & &  $d>1$?\\
\hline\hline
 \cite{CTao,RV, CGV}&  &  & &\\
 \cite{Bourgain2014, HR16}  & $N \times m$ linear program & $O(k \log^2(k) \log(N))$ & $O(1)$ &yes\\
\cite{CP} & $N \times m$ linear program & $O(k \log N)$ & $(\log N)^{O(1)}$ &yes\\
\cite{HIKP2} & $O(k \log(N) \log(N/k))$ & $O(k \log(N) \log(N/k))$   &   any &no \\
\cite{IKP} &  $k \log^2(N)  \log^{O(1)} \log N$ & $k \log(N) \log^{O(1)} \log N$ & any &no\\
\cite{IK14a} & $N \log^{O(1)} N$ & $O(k \log N)$ & any & yes\\
\hline
\cite{DIPW} & & $\Omega(k \log(N/k))$ & O(1) & lower bound\\
\hline
\end{tabular}
\end{center}
\caption{Bounds for the algorithms that recover $k$-sparse Fourier approximations. All algorithms produce an output satisfying Equation~\ref{e:l2l2} with probability of success that is at least constant. The forth column specifies constraints on approximation factor $C$. For example, $C=O(1)$ means that the algorithm can only handle constant $C$ as opposed to any $C>1$. The last column specifies whether the sample complexity bounds are unchanged, up to factors that depend on dimension $d$ only, for higher dimensional DFT.}\label{fig:1}
\end{figure*}

{\bf Higher dimensional Fourier transform.}  While significant attention in the sublinear Sparse FFT literature has been devoted to the basic case of Fourier transform on the line (i.e. one-dimensional signals),  the  sparsest signals often occur in applications involving higher-dimensional DFTs.  Although a reduction from DFT on a two-dimensional grid {\em with relatively prime side lengths} $p \times q$ to a one-dimensional DFT of length $pq$ is possible~\cite{GMS,Iw-arxiv}),  the reduction does not apply to the most common case when the side lengths of the grid are equal to the same powers of two. It turns out that most sublinear Sparse FFT techniques developed for the one-dimensional DFT do not extend well to the higher dimensional setting, suffering from {\em at least a polylogaritmic loss in sample complexity}. Specifically, the only prior sublinear time algorithm that applies to general $m\times m$ grids is due to  to~\cite{GMS}, has $O(k \log^c N)$ sample and time complexity for a
 rather large value of $c$.  If $N$ is a power of $2$, a
 two-dimensional adaptation of the~\cite{HIKP2} algorithm (outlined in~\cite{GHIKPS}) 
has roughly $O(k \log^3 N)$ time and sample complexity, and an adaptation of~\cite{IKP} has $O(k\log^2 N(\log\log N)^{O(1)})$ sample complexity. In general dimension $d\geq 1$ both of these algorithms have sample complexity $\Omega(k\log^d N)$.

Thus, none of the results obtained so far was able to guarantee sparse recovery from high dimensional (any $d\geq 2$) Fourier measurements without suffering at least a polylogarithmic loss in sample complexity, while at the same time achieving sublinear runtime.

{\bf Our results.} In this paper we give an algorithm that achieves the $\ell_2/\ell_2$ sparse recovery guarantees~\eqref{e:l2l2} with $d$-dimensional Fourier measurements that uses $O_d(k \log N$ $\log \log N)$ samples of the signal and has the running time of $O_d(k \log^{d+3} N)$. This is the first sublinear time algorithm that comes within a $\poly(\log\log N)$ factor of the sample complexity lower bound of $\Omega(k\log (N/k))$ due to ~\cite{DIPW} for any dimension higher than one.

{\bf Sparse Fourier Transform overview.}  The overall outline of our algorithm follows the framework of~\cite{GMS, HIKP2,IKP,IK14a}, which adapt the methods of~\cite{CCF,GLPS} from arbitrary linear
measurements to Fourier measurements.  The idea is to take, multiple times,  a set of $B=O(k)$
linear measurements of the form
\[
\twid{u}_j = \sum_{i : h(i) = j} s_i x_i 
\]
for random hash functions $h: [N] \to [B]$ and random sign changes
$s_i$ with $\abs{s_i} = 1$.  This corresponds to \emph{hashing} to $B$
\emph{buckets}.  With such ideal linear measurements, $O(\log(N/k))$
hashes suffice for sparse recovery, giving an $O(k \log(N/k))$ sample
complexity.

The sparse Fourier transform algorithms approximate $\twid{u}$ using linear combinations of Fourier samples. Specifically, the coefficients of $x$ are first permuted via a random affine permutation of the input space.  Then the coefficients are partitioned into buckets. This step uses  the``filtering'' process that approximately partitions the range of $x$ into intervals (or, in higher dimension, squares, or $\ell_\infty$ balls) with $N/B$ coefficients each, where each interval corresponds to one bucket. Overall, this  ensures that
 most of the large coefficients are ``isolated'', i.e.,  are hashed to unique buckets, as well as that the contributions from the  ``tail'' of the signal $x$ to those buckets is not much greater than the average (the tail of the signal defined as $\err_k(x)=\min_{k-\text{sparse}~y} ||x-y||_2$).  This allows one to mimic the iterative recovery algorithm described for linear measurements above. However, there are several difficulties in making this work using an optimal number of samples.

 This enables the algorithm to identify the locations of the dominant coefficients and estimate their values, producing a sparse estimate $\chi$ of $x$.  To improve this
estimate, we repeat the process on $x - \chi$ by subtracting the
influence of $\chi$ during hashing, thereby {\em refining } the approximation of $x$ constructed.  After a few iterations of this refinement process the algorithm obtains a good 
sparse approximation $\chi$ of $x$.

A major hurdle in implementing this strategy is that any filter that has been constructed in the literature so far is imprecise in that  coefficients  contribute (``leak''') to buckets other than the one they are technically mapped into. This contribution, however, is limited and can be controlled by the quality of the filter.  The details of filter choice have played a crucial role in recent developments in Sparse FFT algorithms. For example, the best known runtime for one-dimensional Sparse FFT, due to~\cite{HIKP}, was obtained by constructing filters that (almost) precisely mimic the ideal hash process, allowing for a very fast implementation of the process in dimension one. The price to pay for the precision of the filter, however, is that each hashing becomes a $\log^d N$ factor more costly in terms of sample complexity and runtime than in the idealized case. At the other extreme, the algorithm of~\cite{GMS} uses much less precise filters, which only lead to a $C^d$ loss of sample complexity in higher dimensions $d$, for a constant $C>0$. Unfortunately, because of the imprecision of the filters the iterative improvement process becomes quite noisy, requiring $\Omega(\log N)$ iterations of the refinement process above. As~\cite{GMS} use fresh randomness for each such iteration, this results in an $\Omega(\log N)$ factor loss in sample complexity. The result of~\cite{IKP} uses a hybrid strategy, effectively interpolating between~\cite{HIKP} and~\cite{GMS}. This gives the near optimal $O(k\log N \log^{O(1)}\log N)$ sample complexity in dimension one (i.e. Fourier transform on the line), but still suffers from a $\log^{d-1} N$ loss in dimension $d$.



{\bf Techniques of~\cite{IK14a}.} The first algorithm to achieve optimal sample complexity was recently introduced in~\cite{IK14a}.  The algorithms uses an approach inspired by~\cite{GMS} (and hence uses `crude' filters that do not lose much in sample complexity), but introduces a key innovation enabling optimal sample complexity: the algorithm does {\em not} use fresh hash functions in every repetition of the refinement process. Instead, $O(\log N)$ hash functions are chosen at the beginning of the process, such that each large coefficient is isolated by most of those functions with high probability.  The same hash functions are then used throughout the execution of the algorithm. As every hash function required a separate set of samples to construct the buckets, reusing the hash functions makes sample complexity {\em independent of the number of iterations}, leading to the optimal bound. 
 
 While a natural idea, reusing hash functions creates a major difficulty: if the algorithm identified a non-existent large coefficient (i.e. a false positive) by mistake and added it to $\chi$, this coefficient would be present in the difference vector $x - \chi$ (i.e. residual signal) and would need to be corrected later. As the spurious coefficient depends on the measurements, the `isolation' properties required for recovery need not hold for it as its position is determined by the hash functions themselves, and the algorithm might not be able to correct the mistake.  This hurdle was overcome in~\cite{IK14a} by ensuring that no large coefficients are created spuriously throughout the execution process. This is a nontrivial property to achieve, as the hashing process is quite noisy due to use of the `crude' filters to reduce the number of samples (because the filters are quite simple, the bucketing process suffers from substantial leakage). The solution was to recover the large coefficients in decreasing order of their magnitude. Specifically, in each step, the algorithm recovered coefficients with magnitude that exceeded a specific threshold (that decreases at an exponential rate). With this approach the $\ell_\infty$ norm of the residual signal decreases by a constant factor in every round, resulting in the even stronger $\ell_\infty/\ell_2$ sparse recovery guarantees in the end. The price to pay for this strong guarantee was the need for a very strong primitive for locating dominant elements in the residual signal: a primitive was needed that would make mistakes with at most inverse polynomial probability. This was achieved by essentially brute-force decoding over all potential elements in $[N]$: the algorithm loops over all elements $i\in [N]$ and for each $i$ tests, using the $O(\log N)$ measurements taken, whether $i$ is a dominant element in the residual signal. This resulted in $\Omega(N)$ runtime.

{\bf Our techniques.} In this paper we show how to make the aforementioned algorithm run in sub-linear time, at the price of a slightly increased sampling complexity of $O_d(k\log N$ $\log\log N)$. To achieve a sub-linear runtime, we need to replace the loop over all $N$ coefficients by a location  primitive (similar to that in prior works) that identifies the position of any large coefficient that is isolated in a bucket in $\log^{O(1)} N$ time per bucket, i.e. without resorting to brute force enumeration over the domain of size $N$. Unfortunately, the identification step alone increases the sampling complexity by $O(\log N)$ per hash function, so unlike~\cite{IK14a}, here we cannot repeat this process using $O(\log N)$ hash functions to ensure that each large coefficient is isolated by one of those functions. Instead, we can only afford $O(\log \log N)$  hash functions overall, which means that $1/\log^{O(1)} N$ fraction of large coefficients will not be isolated in most hashings. This immediately precludes the possibility of using the initial samples to achieve $\ell_\infty$ norm reduction as in~\cite{IK14a}. Another problem, however, is that the weaker location primitive that we use may  generate {\em spurious coefficients} at every step of the recovery process. These spurious coefficients, together with the $1/\log^{O(1)} N$ fraction of non-isolated elements, contaminate the recovery process and essentially render the original samples useless after a small number of refinement steps. To overcome these hurdles, instead of the $\ell_\infty$ reduction process of~\cite{IK14a} we use a weaker invariant on the reduction of mass in the `heavy' elements of the signal throughout our iterative process. Specifically, instead of reduction of $\ell_\infty$ norm of the residual as in~\cite{IK14a} we give a procedure for reducing the $\ell_1$ norm of the `head' of the signal. To overcome the contamination coming from non-isolated as well as spuriously created coefficients, we achieve $\ell_1$ norm reduction by alternating two procedures. The first procedure uses the $O(\log\log N)$ hash functions to reduce the $\ell_1$ norm of `well-hashed' elements in the signal, and the second uses a simple sparse recovery primitive to reduce the $\ell_\infty$ norm of offending coefficients when the first procedure gets stuck.  This can be viewed as a signal-to-noise ratio (SNR) reduction step similar in spirit the one achieved in~\cite{IKP}. The SNR reduction phase is insufficient for achieving the $\ell_2/\ell_2$ sparse recovery guarantee, and hence we need to run a cleanup phase at the end, when the signal to noise ratio is constant. It has been observed before (in~\cite{IKP}) that if the signal to noise ratio is constant, then recovery can be done using standard techniques with optimal sample complexity. The crucial difference between~\cite{IKP} and our setting is, however, that we only have bounds on $\ell_1$-SNR as opposed to $\ell_2$-SNR In~\cite{IKP}. It turns out, however, that this is not a problem -- we give a stronger analysis of the corresponding primitive from~\cite{IKP}, showing that $\ell_1$-SNR bound is sufficient.

{\bf Related work on continuous Sparse FFT.} Recently \cite{BCGLS} and \cite{PZ15} gave algorithms for the related problem of computing Sparse FFT in the continuous setting. These results are not directly comparable to ours, and suffer from a polylogarithmic  inefficiency in sample complexity bounds.

\section{Preliminaries}\label{sec:prelim}

For a positive even integer $a$ we will use the notation $[a]=\{-\frac{a}{2}, -\frac{a}{2}+1, \ldots, -1, 0, 1,\ldots, \frac{a}{2}-1\}$. We will consider signals of length $N=n^d$, where $n$ is a power of $2$ and $d\geq 1$ is the dimension.  We use the notation $\omega=e^{2\pi i/n}$ for the root of unity of order $n$. The $d$-dimensional forward and inverse Fourier transforms are given by
\begin{equation}\label{eq:dft-forward}
\hat x_{j}=\frac1{\sqrt{N}}\sum_{i\in \nsq}  \omega^{-i^Tj}x_i \text{~~and~~}x_{j}=\frac1{\sqrt{N}}\sum_{i\in \nsq}  \omega^{i^Tj}\hat x_i
\end{equation}
respectively, where $j\in \nsq$. We will denote the forward Fourier transform by $\F$ and 
Note that we use the orthonormal version of the Fourier transform. We assume that the input signal has entries of polynomial precision and range. Thus, we have $||\hat x||_2=||x||_2$ for all $x\in \C^N$ (Parseval's identity).
Given access to samples of $\wh{x}$, we recover a signal $z$ such that 
\begin{equation*}
||x-z||_2\leq (1+\e)\min_{k-\text{~sparse~} y} ||x-y||_2
\end{equation*}

We will use pseudorandom spectrum permutations, which we now define.  We write $\gl$ for the set of $d\times d$ matrices  over $\Z_n$ with odd determinant.
For $\Sigma\in \gl, q\in \nsq$ and $i\in \nsq$ let 
$\pi_{\Sigma, q}(i)=\Sigma(i-q) \mod n$.
Since $\Sigma\in \gl$, this is a permutation. Our algorithm will use $\pi$ to hash heavy hitters into $B$ buckets, where we will choose $B\approx k$. We will often omit the subscript $\Sigma, q$ and simply write $\pi(i)$ when $\Sigma, q$ is fixed or clear from context. For $i, j\in \nsq$ we let $o_i(j)= \pi(j) - (n/b) h(i)$  be the ``offset'' of $j\in \nsq$ relative to $i\in \nsq$ (note that this definition is different from the one in~\cite{IK14a}). We will always have $B=b^d$, where $b$ is a power of $2$. 
\begin{definition}
  Suppose that $\Sigma^{-1}$ exists $\bmod~n$. For $a, q\in \nsq$ we define the
  permutation $P_{\Sigma, a, q}$ by $(P_{\Sigma, a, q}\hat x)_i=\hat x_{\Sigma^T(i-a)} \omega^{i^T\Sigma q}$.  
\end{definition}
\begin{lemma}\label{lm:perm}
$\F^{-1}({P_{\Sigma, a, q} \hat x})_{\pi_{\Sigma, q}(i)}=x_i \omega^{a^T \Sigma i}$
\end{lemma}
The proof is given in~\cite{IK14a} and we do not repeat it here. 
Define 
\begin{equation}\label{eq:mu-def}
\begin{split}
\err_k(x)=\min_{k-\text{sparse}~y} ||x-y||_2\text{~~and~~}\mu^2=\err_k^2(x)/k.
\end{split}
\end{equation}
In this paper, we assume knowledge of $\mu$ (a constant factor upper bound on $\mu$ suffices). We also assume that the signal to noise ration is bounded by a polynomial, namely 
that $R^*:=||x||_\infty/\mu\leq N^{O(1)}$. 
We use the notation $\B^\infty_{r}(x)$ to denote the $\ell_\infty$ ball of radius $r$ around $x$:
$\B^\infty_{r}(x)=\{y\in \nsq: ||x-y||_\infty\leq r\}$, where $||x-y||_\infty=\max_{s\in d} ||x_s-y_s||_{\circ}$, and $||x_s-y_s||_{\circ}$ is the circular distance on $\mathbb{Z}_n$. 
We will also use the notation $f\lesssim g$ to denote $f=O(g)$. For a real number $a$ we write $|a|_+$ to denote the positive part of $a$, i.e. $|a|_+=a$ if $a\geq 0$ and $|a|_+=0$ otherwise.

We will use the filter $G, \hat G$ constructed in~\cite{IK14a}. The filter is defined by a parameter $F\geq 1$ that governs its decay properties. The filter satisfies  $\supp \hat G\subseteq [-F\cdot b, F\cdot b]^d$ and 
\begin{lemma}[Lemma~3.1 in~\cite{IK14a}]\label{lm:filter-prop}
One has {\bf (1)} $G_j\in [\frac1{(2\pi)^{\fc \cdot d}}, 1]$ for all $j\in \nsq$ such that $||j||_\infty\leq \frac{n}{2b}$ and {\bf (2)} $|G_j|\leq \left(\frac2{1+(b/n)||j||_\infty}\right)^{\fc}$ for all $j\in \nsq$
as long as $b\geq 3$ and {\bf (3)} $G_j\in [0, 1]$ for all $j$ as long as $F$ is even.
\end{lemma}
\begin{remark}
Property {\bf (3)} was not stated explicitly in Lemma~3.1 of~\cite{IK14a}, but follows directly from their construction.
\end{remark}
The properties above imply that most of the mass of the filter is concentrated in a square of side $O(n/b)$, approximating the ``ideal'' filter (whose value would be equal to $1$ for entries within the square and equal to $0$ outside of it). 
Note that for each $i\in \nsq$ one has $|G_{o_i(i)}|\geq \frac1{\GS}$. We refer to the parameter $F$ as the {\em sharpness} of the filter. Our hash functions are not pairwise independent, but possess a property that still makes hashing using our filters efficient:
\begin{lemma}[Lemma~3.2 in~\cite{IK14a}]\label{lemma:limitedindependence}
Let $i, j\in \nsq$. Let $\Sigma$ be uniformly random with odd determinant. Then for all $t\geq 0$ one has $
\Pr[||\Sigma(i-j)||_\infty \leq t] \leq 2(2t/n)^d$.
\end{lemma}

Pseudorandom spectrum permutations combined with a filter $G$ give us the ability to `hash' the elements of the input signal into a number of buckets (denoted by $B$). We formalize this using the notion of a {\em hashing}. A hashing is a tuple consisting of a pseudorandom spectrum permutation $\pi$, target number of buckets $B$ and a sharpness parameter $F$ of our filter, denoted by $H=(\pi, B, F)$. Formally, $H$ is a function that maps a signal $x$ to $B$ signals, each corresponding to a hash bucket, allowing us to solve the $k$-sparse recovery problem on input $x$ by reducing it to $1$-sparse recovery problems on the bucketed signals. We give the formal definition below.

\begin{definition}[Hashing $H=(\pi, B, F)$]
For a permutation $\pi=(\Sigma, q)$, parameters $b>1$, $B=b^d$ and $F$,  a {\em hashing} $H:=(\pi, B, F)$ is a function mapping a signal $x\in \C^{\nsq}$ to $B$ signals $H(x)=(u_s)_{s\in [b]^d}$, where $u_s\in \C^{\nsq}$ for each $s\in [b]^d$, such that for each $i\in \nsq$
$$
  u_{s, i}  = \sum_{j\in \nsq} G_{\pi(j)-(n/b)\cdot s}x_j \omega^{i^T \Sigma j}\in \C,
$$
where $G$ is a filter with $B$ buckets and sharpness $F$ constructed in Lemma~\ref{lm:filter-prop}.
\end{definition}

For a hashing $H=(\pi, B, F), \pi=(\Sigma, q)$ we sometimes write $P_{H, a}, a\in \nsq$ to denote $P_{\Sigma, a, q}$. We will consider hashings of the input signal $x$, as well as the residual signal $x-\chi$, where 

\begin{definition}[Measurement $m=m(x, H, a)$]
For a signal $x\in \C^{\nsq}$, a hashing $H=(\pi, B, F)$ and a parameter $a\in \nsq$, a {\em measurement} $m=m(x, H, a)\in \C^{[b]^d}$ is the $B$-dimensional complex valued vector of evaluations of a hashing $H(x)$ at $a\in \C^{\nsq}$, i.e. 
length $B$, indexed by $[b]^d$ and given by evaluating the hashing $H$ at $a\in \nsq$, i.e. for $s\in [b]^d$
$$
m_s  = \sum_{j\in \nsq} G_{\pi(j)-(n/b)\cdot s}x_j \omega^{a^T \Sigma j},
$$
where $G$ is a filter with $B$ buckets and sharpness $F$ constructed in Lemma~\ref{lm:filter-prop}.
\end{definition}

\begin{definition}
For any $x\in \C^\nsq$ and any hashing $H=(\pi, B, G)$ define the vector $\mu^2_{H, \cdot}(x)\in \R^\nsq$ by letting for every $i\in \nsq$ 
$$\mu^2_{H, i}(x):=\abs{G_{o_i(i)}^{-1}}\sum_{j \in \nsq\setminus \{i\}} \abs{x_j}^2 \abs{G_{o_i(j)}}^2.$$
\end{definition}

We access the signal $x$ in Fourier domain via the function $\Call{HashToBins}{\hat x, \chi, (H, a)}$, which evaluates the hashing $H$ of residual signal $x-\chi$ at point $a\in \nsq$, i.e. computes the measurement $m(x, H, a)$ (the computation is done with polynomial precision). One can view this function as ``hashing''  $x$ into $B$ bins by convolving it with the filter $G$ constructed above and subsampling appropriately. The pseudocode for this function is given in section~\ref{sec:hash2bins}. In what follows we will use the following properties of \textsc{HashToBins}:
\begin{lemma}\label{lm:hashing}
There exists a constant $C>0$ such that for any dimension $d\geq 1$, any integer $B\geq 1$, any $x, \chi\in \C^\nsq, x':=x-\chi$, if  $\Sigma\in \gl, a, q\in \nsq$  are selected uniformly at random, the following conditions hold. 

Let $\pi=(\Sigma, q)$, $H=(\pi, B, G)$, where $G$ is the filter with $B$ buckets and sharpness $F$ constructed in Lemma~\ref{lm:filter-prop}, and let
$u =\Call{HashToBins}{\hat x, \chi, (H, a)}$.  Then if $\fc\geq 2d, F=\Theta(d)$,  for any $i\in \nsq$
  \begin{description}
\item[(1)] For any $H$ one has $\max_{a\in \nsq} \abs{G_{o_i(i)}^{-1}\omega^{-a^T\Sigma i}u_{h(i)} - x'_i}\leq G_{o_i(i)}^{-1}\cdot \sum_{j\in S\setminus \{i\}} G_{o_i(j)} |x'_j|$. Furthermore,
$\expect_H[G_{o_i(i)}^{-1}\cdot \sum_{j\in S\setminus \{i\}} G_{o_i(j)} |x'_j|]\leq \GS\cdot C^d ||x'||_1/B+N^{-\Omega(c)}$;
  \item[(2)] $\expect_H[\mu^2_{H, i}(x')] \leq  \GSS\cdot C^d \norm{2}{x'}^2/B$,
\end{description}
Furthermore, 
\begin{description}
\item[(3)] for any hashing $H$, if $a$ is chosen uniformly at random from $\nsq$, one has 
  $$
  \expect_{a}[\abs{G_{o_i(i)}^{-1}\omega^{-a^T\Sigma i}u_{h(i)} - x'_i}^2]\leq \mu^2_{H, i}(x')+N^{-\Omega(c)}.
  $$
\end{description}
Here $c>0$ is an absolute constant that can be chosen arbitrarily large at the expense of a  factor of $c^{O(d)}$ in runtime.
\end{lemma}
The proof of Lemma~\ref{lm:hashing} is given in Appendix~\ref{app:A}.  We will need several definitions and lemmas from~\cite{IK14a}, which we state here. We sometimes need slight modifications of the corresponding statements from~\cite{IK14a}, in which case we provide proofs in Appendix~\ref{app:A}.  
Throughout this paper the main object of our analysis is a properly defined set $S\subseteq \nsq$ that contains the 'large' coefficients of the input vector $x$. Below we state our definitions and auxiliary lemmas without specifying the identity of this set, and then use specific instantiations of $S$ to analyze outer primitives such as \textsc{ReduceL1Norm}, \textsc{ReduceInfNorm} and \textsc{RecoverAtConstSNR}.  This is convenient because the analysis of all of these primitives can then use the same basic claims about estimation and location primitives. The definition of $S$ given in~\eqref{eq:s-def-l1} above is the one we use for analyzing \textsc{ReduceL1Norm} and the SNR reduction loop. Analysis of \textsc{ReduceInfNorm} (section~\ref{sec:linf}) and \textsc{RecoverAtConstantSNR} (section~\ref{sec:const-snr}) use different instantiations of $S$, but these are local to the corresponding sections, and hence the definition in~\eqref{eq:s-def-l1} is the best one to have in mind for the rest of this section. 

First, we need the definition of an element $i\in \nsq$ being isolated under a hashing $H=(\pi, B, F)$. Intuitively, an element $i\in S$ is isolated under hashing $H$ with respect to set $S$ if not too many other elements $S$ are hashed too close to $i$. Formally, we have
\begin{definition}[Isolated element]\label{def:isolated}
Let $H=(\pi, B, F)$, where $\pi=(\Sigma, q)$, $\Sigma\in \gl, q\in \nsq$. We say that an element $i\in \nsq$ is {\em isolated} under hashing $H$ {\em at scale $t$} if
$$
|\pi(S\setminus \{i\})\cap \B^\infty_{(n/b)\cdot h(i)}((n/b)\cdot 2^t)|\leq \GI\cdot \alpha^{d/2} 2^{(t+1)d}\cdot 2^{t}.
$$
We say that $i$ is simply {\em isolated} under hashing $H$ if it is isolated under $H$ at all scales $t\geq 0$.
\end{definition}
The following lemma shows that any element $i\in S$ is likely to be isolated under a random permutation $\pi$:
\begin{lemma}\label{lm:isolated-pi}
For any integer $k\geq 1$ and any $S\subseteq \nsq, |S|\leq 2k$, if  $B\geq \GT\cdot k/\alpha^d$ for $\alpha\in (0, 1)$ smaller than an absolute constant, $F\geq 2d$, and a hashing $H=(\pi, B, F)$ is chosen randomly (i.e. $\Sigma\in \gl, q\in \nsq$ are chosen uniformly at random, and $\pi=(\Sigma, q)$), then each $i\in \nsq$ is {\em isolated} under permutation $\pi$ with probability at least $1-\frac1{2}\sqrt{\alpha}$.
\end{lemma}
The proof of the lemma is very similar to Lemma~5.4 in~\cite{IK14a} (the only difference is that the $\ell_\infty$ ball is centered at the point that $i$ hashes to in Lemma~\ref{lm:isolated-pi}, whereas it was centered at $\pi(i)$ in Lemma~5.4 of~\cite{IK14a}) and is given in Appendix~\ref{app:A} for completeness.

As every element $i\in S$ is likely to be isolated under one random hashing, it is very likely to be isolated under a large fraction of hashings $H_1,\ldots, H_{r_{max}}$:
\begin{lemma}\label{lm:good-prob}
For any integer $k\geq 1$, and any $S\subseteq \nsq, |S|\leq 2k$, if $B\geq \GT\cdot k/\alpha^d$ for $\alpha\in (0, 1)$ smaller than an absolute constant, $F\geq 2d$, $H_r=(\pi_r, B, F)$, $r=1,\ldots, r_{max}$ a sequence of random hashings, then every $i\in \nsq$ is isolated with respect to $S$ under at least $(1-\sqrt{\alpha})r_{max}$ hashings $H_r, r=1,\ldots, r_{max}$ with probability at least $1-2^{-\Omega(\sqrt{\alpha}r_{max})}$.
\end{lemma}
\begin{proof}
Follows by an application of Chernoff bounds and Lemma~\ref{lm:isolated-pi}.
\end{proof}

It is convenient for our location primitive (\textsc{LocateSignal}, see Algorithm~\ref{alg:location}) to sample the signal at pairs of locations chosen randomly (but in a correlated fashion). The two points are then combined into one in a linear fashion. We now define notation for this common operation on pairs of numbers in $\nsq$. Note that we are viewing pairs in $\nsq\times \nsq$ as vectors in dimension $2$, and the $\star$ operation below is just the dot product over this two dimensional space. However, since our input space is already endowed with a dot product (for $i, j\in \nsq$ we denote their dot product by $i^Tj$), having special notation here will help avoid confusion.
\paragraph{Operations on vectors in $\nsq$.} For a pair of vectors $(\alpha_1, \beta_1), (\alpha_2, \beta_2)\in \nsq\times \nsq$ we let $(\alpha_1, \beta_1)\star (\alpha_2, \beta_2)$ denote the vector $\gamma\in \nsq$ such that 
$$
\gamma_i=(\alpha_1)_i\cdot (\alpha_2)_i+(\beta_1)_i\cdot (\beta_2)_i\text{~~~for all ~}i\in [d].
$$
Note that for any $a, b, c\in \nsq\times \nsq$ one has $a\star b+a\star c=a\star (b+c)$, where addition for elements of $\nsq\times \nsq$ is componentwise.
We write $\one\in \nsq$ for the all ones vector in dimension $d$, and $\zero\in \nsq$ for the zero vector. For a set $\A\subseteq \nsq\times \nsq$ and a vector $(\alpha, \beta)\in \nsq\times \nsq$ we denote
$$
\A\star (\alpha, \beta):=\{a\star(\alpha, \beta): a\in \A\}.
$$

\begin{definition}[Balanced set of points]\label{def:balance}
For an integer $\Delta\geq 2$  we say that a (multi)set $\mathcal{Z}\subseteq \nsq$ is {\em $\Delta$-balanced} in coordinate $s\in [1:d]$ if for every $r=1,\ldots, \Delta-1$ at least $49/100$ fraction of elements in the set $\{\omega_{\Delta}^{r\cdot z_s}\}_{z\in \mathcal{Z}}$ belong to the left halfplane $\{u\in \C: \text{Re}(u)\leq 0\}$ in the complex plane, where $\omega_\Delta=e^{2\pi i/\Delta}$ is the $\Delta$-th root of unity.
\end{definition}

Note that if $\Delta$ divides $n$, then for any fixed value of $r$ the point $\omega_\Delta^{r\cdot z_s}$ is uniformly distributed over the $\Delta'$-th roots of unity for some $\Delta'$ between $2$ and $\Delta$ for every $r=1,\ldots, \Delta-1$ when $z_s$ is uniformly random in $[n]$. Thus for $r\neq 0$ we expect at least half the points to lie in the halfplane $\{u\in \C: \text{Re}(u)\leq 0\}$. A set $\mathcal{Z}$ is balanced if it does not deviate from expected behavior too much.
The following claim is immediate via standard concentration bounds:

\begin{claim}\label{cl:balanced}
There exists a constant $C>0$ such that for any $\Delta$ a power of two, $\Delta=\log^{O(1)} n$,  and $n$ a power of $2$ the following holds if $\Delta<n$. If elements of a (multi)set $\A\subseteq \nsq\times \nsq$ of size $C\log\log N$ are chosen uniformly at random with replacement from $\nsq\times \nsq$, then with probability at least $1-1/\log^4 N$ one has that for every $s\in [1:d]$ the set $\A\star (\zero, \mathbf{e}_s)$ is $\Delta$-balanced in coordinate $s$.
\end{claim}
Since we only use one value of $\Delta$ in the paper (see line~8 in Algorithm~\ref{alg:location}), we will usually say that a set is simply `balanced'  to denote the $\Delta$-balanced property for this value of $\Delta$.

\section{The algorithm and proof overview}\label{sec:sublinear}

In this section we state our algorithm and give an outline of the analysis. The formal proofs are then presented in the rest of the paper (the organization of the rest of the paper is presented in section~\ref{sec:org}).   Our algorithm (Algorithm~\ref{alg:main-sublinear}), at a high level, proceeds as follows. 

{\bf Measuring $\wh{x}$.} The algorithms starts by taking measurements of the signal in lines~5-16. Note that the algorithm selects $O(\log\log N)$ hashings $H_r=(\pi_r, B, F), r=1,\ldots, O(\log\log N)$, where $\pi_r$ are selected uniformly at random, and for each $r$ selects a set $\A_r\subseteq \nsq\times \nsq$ of size $O(\log\log N)$ that determines locations to access in frequency domain. The signal $\wh{x}$ is accessed via the function \textsc{HashToBins} (see Lemma~\ref{lm:hashing} above for its properties. The function \textsc{HashToBins} accesses filtered versions of $\wh{x}$ shifted by elements of a randomly selected set (the number of shifts is $O(\log N/\log \log N)$). These shifts are useful for locating `heavy' elements from the output of \textsc{HashToBins}.  Note that since each hashing takes $O(B)=O(k)$ samples, the total sample complexity of the measurement step is $O(k\log N \log\log N)$. This is the dominant contribution to sample complexity, but it is not the only one. The other contribution of $O(k\log N\log\log N)$ comes from invocations of \textsc{EstimateValues} from our $\ell_1$-SNR reduction loop (see below). The loop goes over $O(\log R^*)=O(\log N)$ iterations, and in each iteration \textsc{EstimateValues} uses $O(\log\log N)$ fresh hash functions to keep the number of false positives and estimation error small.

The location algorithm is Algorithm~\ref{alg:location}.  Our main tool for bounding performance of \textsc{LocateSignal} is Theorem~\ref{thm:l1-res-loc}, stated below. Theorem~\ref{thm:l1-res-loc} applies to the following setting. Fix a set $S\subseteq \nsq$ and a set of hashings $H_1,\ldots, H_{r_{max}}$ that encode signal measurement patterns, and let $S^*\subseteq S$ denote the set of elements of $S$ that are not isolated with respect to most of these hashings. Theorem~\ref{thm:l1-res-loc} shows that for any signal $x$ and partially recovered signal $\chi$, if $L$ denotes the output list of an invocation of \textsc{LocateSignal} on the pair $(x, \chi)$ with measurements given by $H_1,\ldots, H_{r_{max}}$ and a set of random shifts, then the $\ell_1$ norm of elements of the residual $(x-\chi)_S$ that are not discovered by \textsc{LocateSignal} can be bounded by a function of the amount of $\ell_1$ mass of the residual that fell outside of the `good' set $S\setminus S^*$, plus the `noise level' $\mu\geq ||x_{\nsq\setminus S}||_\infty$ times $k$.

If we think of applying Theorem~\ref{thm:l1-res-loc} iteratively, we intuitively get that the fixed set of measurements given by hashings $H_1,\ldots, H_r$ allows us to always reduce the $\ell_1$ norm of the residual $x'=x-\chi$ on the `good' set $S\setminus S^*$ to about the amount of mass that is located outside of this good set(this is exactly how we use \textsc{LocateSignal} in our signal to noise ratio reduction loop below). In section~\ref{sec:l1} we prove

\begin{theorem}\label{thm:l1-res-loc}
For any constant $C'>0$ there exist absolute constants $C_1, C_2, C_3>0$ such that for any $x, \chi\in \C^N$, $x'=x-\chi$, any integer $k\geq 1$ and any $S\subseteq \nsq$ such that $||x_{\nsq\setminus S}||_\infty\leq C'\mu$, where $\mu=||x_{\nsq\setminus [k]}||_2/\sqrt{k}$, the following conditions hold if $||x'||_\infty/\mu=N^{O(1)}$.

Let $\pi_r=(\Sigma_r, q_r), r=1,\ldots, r_{max}$ denote permutations, and let $H_r=(\pi_r, B, F)$, $F\geq 2d, F=\Theta(d)$, where $B\geq \GT k/\alpha^d$ for $\alpha\in (0, 1)$ smaller than a constant. Let $S^*\subseteq S$ denote the set of elements that are not isolated with respect to at least a $\sqrt{\alpha}$ fraction of hashings $\{H_r\}$. Then if additionally for every $s\in [1:d]$ the sets $\A_r\star (\one, \mathbf{e}_s)$ are balanced in coordinate $s$ (as per Definition~\ref{def:balance}) for all $r=1,\ldots, r_{max}$, and $r_{max}, c_{max}\geq (C_1/\sqrt{\alpha})\log\log N$, then
$$
L:=\bigcup_{r=1}^{r_{max}}\textsc{LocateSignal}\left(\chi, k, \{m(\wh{x}, H_r, a\star (\one, \h))\}_{r=1, a\in \A_r, \h\in \H}^{r_{max}}\right)
$$
satisfies
$$
||x'_{S\setminus S^*\setminus L}||_1\leq (C_2\alpha)^{d/2} ||x'_S||_1+C_3^{d^2}(||\chi_{\nsq\setminus S}||_1+||x'_{S^*}||_1)+4\mu |S|.
$$
\end{theorem}

\paragraph{\bf Reducing signal to noise ratio.} Once the samples have been taken, the algorithm proceeds to the signal to noise (SNR) reduction loop (lines~17-23). The objective of this loop is to reduce the mass of the top (about $k$) elements in the residual signal to roughly the noise level $\mu\cdot k$ (once this is done, we run a `cleanup' primitive, referred to as \textsc{RecoverAtConstantSNR}, to complete the recovery process -- see below). Specifically, we define the set $S$ of `head elements' in the original signal $x$ as 
\begin{equation}\label{eq:s-def-l1}
S=\{i\in \nsq: |x_i|>\mu\},
\end{equation}
where $\mu^2=\err_k^2(x)/k$ is the average tail noise level. Note that we have $|S|\leq 2k$. Indeed, if $|S|>2k$, more than $k$ elements of $S$ belong to the tail, amounting to more than $\mu^2\cdot k=\err_k^2(x)$ tail mass. Ideally, we would like this loop to construct and approximation  $\chi^{(T)}$ to $x$ {\em supported only on $S$} such that $||(x-\chi^{(T)})_S||_1=O(\mu k)$, i.e. the $\ell_1$-SNR of the residual signal on the set $S$ of heavy elements is reduced to a constant. 
As some false positives will unfortunately occur throughout the execution of our algorithm due to the weaker sublinear time location and estimation primitives that we use,  our SNR reduction loop is to construct an approximation $\chi^{(T)}$ to $x$ with the somewhat weaker properties that
\begin{equation}\label{eq:294ht943htgrr}
||(x-\chi^{(T)})_S||_1+||\chi^{(T)}||_{\nsq\setminus S}=O(\mu k)\text{~~~~~and~~~~}||\chi^{(T)}||_0\ll k.
\end{equation}
Thus, we reduce the $\ell_1$-SNR on the set $S$ of `head' elements to a constant, and at the same time not introduce too many spurious coefficients  (i.e. false positives) outside $S$, and these coefficients do not contribute much $\ell_1$ mass. The SNR reduction loop itself consists of repeated alternating invocations of two primitives, namely \textsc{ReduceL1Norm} and \textsc{ReduceInfNorm}. Of these two the former can be viewed as performing most of the reduction, and \textsc{ReduceInfNorm} is naturally viewed as performing a `cleanup' phase to fix inefficiencies of \textsc{ReduceL1Norm} that are due to the small number of hash functions (only $O(\log\log N)$ as opposed to $O(\log N)$ in~\cite{IK14a}) that we are allowed to use, as well as some mistakes that our sublinear runtime location and estimation primitives used in \textsc{ReduceL1Norm} might make.

\begin{algorithm}[H]
\caption{Location primitive: given a set of measurements corresponding to a single hash function, returns a list of elements in $\nsq$, one per each hash bucket}\label{alg:location}
\begin{algorithmic}[1]
\Procedure{LocateSignal}{$\chi, H, \{m(\wh{x}, H, a\star (\one, \h)\}_{a\in \A, \h\in \H}$}\Comment{$H=(\pi, B, F), B=b^d$}
\State Let $x':=x-\chi$. Compute $\{m(\wh{x'}, H, a\star (\one, \h)\}_{a\in \A, \h\in \H}$ using Corollary~\ref{c:semiequi1} and \textsc{HashToBins}.
\State $L\gets \emptyset$
\For{$j \in [b]^d$}\Comment{Loop over all hash buckets, indexed by $j\in [b]^d$}
\State ${\bf f}\gets {\bf 0}^d$
\For {$s=1$ to $d$}\Comment{Recovering each of $d$ coordinates separately}
\State $\Delta\gets 2^{\lfloor \frac1{2}\log_2 \log_2 n\rfloor}$
\For {$g=1$ to $\log_\Delta n-1$} 
\State $\h\gets n\Delta^{-g} \cdot \mathbf{e}_s$ \Comment{Note that $\h\in \H$}
\State {\bf If}~there exists a unique $r\in [0:\Delta-1]$ such that  
\State~~~~~$\left|\omega_{\Delta}^{-r\cdot \beta_s}\cdot \omega^{-(n\cdot \Delta^{-g} {\bf f}_s)\cdot \beta_s}\cdot \frac{m_j(\wh{x'}, H, a\star (\one, \h))}{m_j(\wh{x'}, H, a\star (\one, \zero))}-1\right|<1/3$ for at least $3/5$ fraction of $a=(\alpha, \beta)\in \A$
\State {\bf then}  ${\bf f}\gets {\bf f}+\Delta^{g-1}\cdot r\cdot {\bf e}_s$ ~{\bf else} return {\bf FAIL}
\EndFor
\EndFor
\State $L \gets L \cup \{\Sigma^{-1}{\bf f}\}$ \Comment{Add recovered element to output list}
\EndFor
\State \textbf{return} $L$
\EndProcedure 
\end{algorithmic}
\end{algorithm}

\textsc{ReduceL1Norm} is presented as Algorithm~\ref{alg:l1-norm-reduction} below. The algorithm performs $O(\log\log N)$ rounds of the following process: first, run \textsc{LocateSignal} on the current residual signal, then estimate values of the elements that belong to the list $L$ output by \textsc{LocateSignal}, and {\bf only keep those that are above a certain threshold} (see threshold $\frac1{10000} 2^{-t} \nu+4\mu$ in the call the \textsc{EstimateValues} in line~9 of Algorithm~\ref{alg:l1-norm-reduction}). This thresholding operation is crucial, and allows us to control the number of false positives. In fact, this is very similar to the approach of~\cite{IK14a} of recovering elements starting from the largest. The only difference is that {\bf (a)} our `reliability threshold' is dictated by the $\ell_1$ norm of the residual rather than the $\ell_\infty$ norm, as in~\cite{IK14a}, and {\bf (b)} some false positives can still occur due to our weaker estimation primitives. Our main tool for formally stating the effect of \textsc{ReduceL1Norm} is Lemma~\ref{lm:reduce-l1-norm}  below. Intuitively, the lemma shows that \textsc{ReduceL1Norm} reduces the $\ell_1$ norm of the head elements of the input signal $x-\chi$ by a polylogarthmic factor, and does not introduce too many new spurious elements (false positives) in the process. The  introduced spurious elements, if any, do not contribute much $\ell_1$ mass to the head of the signal. Formally, we show in section~\ref{sec:rl1n}

\begin{lemma}\label{lm:reduce-l1-norm}
For any $x\in \C^N$, any integer $k\geq 1$, $B\geq \GT\cdot k/\alpha^d$ for $\alpha\in (0, 1]$ smaller than an absolute constant and $F\geq 2d, F=\Theta(d)$ the following conditions hold for the set $S:=\{i\in \nsq: |x_i|>\mu\}$, where $\mu^2:=||x_{\nsq\setminus [k]}||_2^2/k$. Suppose that $||x||_\infty/\mu=N^{O(1)}$.

For any sequence of hashings $H_r=(\pi_r, B, F)$, $r=1,\ldots, r_{max}$, if $S^*\subseteq S$ denotes the set of elements of $S$ that are not isolated with respect to at least a $\sqrt{\alpha}$ fraction of the hashings  $H_r, r=1,\ldots, r_{max}$, then for any $\chi\in \C^\nsq$, $x':=x-\chi$, if $\nu\geq (\log^4 N)\mu$ is a parameter such that
\begin{description}
\item[A] $||(x-\chi)_S||_1\leq  (\nu +20\mu)k$;
\item[B] $||\chi_{\nsq\setminus S}||_0\leq \frac{1}{\log^{19} N}k$;
\item[C] $||(x-\chi)_{S^*}||_1+||\chi_{\nsq\setminus S}||_1\leq \frac{\nu}{\log^4 N}k$,
\end{description}
the following conditions hold.

If parameters $r_{max}, c_{max}$ are chosen to be at least $(C_1/\sqrt{\alpha})\log\log N$, where $C_1$ is the constant from Theorem~\ref{thm:l1-res-loc} and measurements are taken as in Algorithm~\ref{alg:main-sublinear}, then the output $\chi'$ of the call 
$$
\Call{ReduceL1Norm}{\chi, k, \{m(\wh{x}, H_r, a\star (\one, \h))\}_{r=1, a\in \A_r, \h\in \H}^{r_{max}}, 4\mu (\log^4 n)^{T-t}, \mu}
$$ 
satisfies
\begin{enumerate}
\item $||(x'-\chi')_S||_1\leq \frac1{\log^{4} N} \nu k+20\mu k$ ~~~~~~~~($\ell_1$ norm of head elements is reduced by $\approx \log^4 N$ factor)
\item $||(\chi+\chi')_{\nsq\setminus S}||_0\leq ||\chi_{\nsq\setminus S}||_0+ \frac1{\log^{20} N}k$~~~~(few spurious coefficients are introduced)
\item $||(x'-\chi')_{S^*}||_1+||(\chi+\chi')_{\nsq\setminus S}||_1\leq ||x'_{S^*}||_1+||\chi_{\nsq\setminus S}||_1+\frac1{\log^{20} N}\nu k$ ~~~~($\ell_1$ norm of spurious coefficients does not grow fast)
\end{enumerate}
with probability at least $1-1/\log^{2} N$ over the randomness used to take measurements $m$ and by calls to \textsc{EstimateValues}.
The number of samples used is bounded by $2^{O(d^2)}k(\log\log N)^2$, and the runtime is bounded by $2^{O(d^2)} k\log^{d+2} N$.
\end{lemma}

Equipped with Lemma~\ref{lm:reduce-l1-norm} as well as its counterpart Lemma~\ref{lm:linf} that bounds the performance of \textsc{ReduceInfNorm} (see section~\ref{sec:linf}) we are able to prove that the SNR reduction loop indeed achieves its cause, namely~\eqref{eq:294ht943htgrr}. Formally, we prove in section~\ref{sec:snr-loop}
\begin{theorem}\label{thm:l1snr}
For any $x\in \C^N$, any integer $k\geq 1$, if $\mu^2=\err_k^2(x)/k$ and $R^*\geq ||x||_\infty/\mu=N^{O(1)}$, the following conditions hold for the set $S:=\{i\in \nsq: |x_i|>\mu\}\subseteq \nsq$.

Then the SNR reduction loop of Algorithm~\ref{alg:main-sublinear} (lines~19-25) returns $\chi^{(T)}$ such that 
\begin{equation*}
\begin{split}
&||(x-\chi^{(T)})_S||_1\lesssim \mu\text{~~~~~~~~~~~~~~~~~($\ell_1$-SNR on head elements is constant)}\\
&||\chi^{(T)}_{\nsq\setminus S}||_1\lesssim \mu \text{~~~~~~~~~~~~~~~~~~~~~~~~~~(spurious elements contribute little in $\ell_1$ norm)}\\
&||\chi^{(T)}_{\nsq\setminus S}||_0\lesssim \frac1{\log^{19} N} k\text{~~~~~~~~~~~~~(small number of spurious elements have been introduced)}
\end{split}
\end{equation*}

with probability at least $1-1/\log N$ over the internal randomness used by Algorithm~\ref{alg:main-sublinear}. The sample complexity is $2^{O(d^2)}k\log N(\log\log N)$. The runtime is bounded by $2^{O(d^2)} k \log^{d+3} N$.
\end{theorem}

\paragraph{ Recovery at constant $\ell_1$-SNR.} Once ~\eqref{eq:294ht943htgrr} has been achieved, we run the \textsc{RecoverAtConstantSNR} primitive (Algorithm~\ref{alg:const-snr}) on the residual signal. Adding the correction $\chi'$ that it outputs to the output $\chi^{(T)}$ of the SNR reduction loop gives the final output of the algorithm. We prove in section~\ref{sec:const-snr}
\begin{lemma}\label{lm:const-snr}
For any $\e>0$, $\hat x, \chi\in \C^N$, $x'=x-\chi$ and any integer $k\geq 1$ if $||x'_{[2k]}||_1\leq O(||x_{\nsq\setminus [k]}||_2\sqrt{k})$ and $||x'_{\nsq\setminus [2k]}||_2^2\leq ||x_{\nsq\setminus [k]}||_2^2$, the following conditions hold. If $||x||_\infty/\mu=N^{O(1)}$, then the output $\chi'$ of 
\Call{RecoverAtConstantSNR}{$\hat x, \chi, 2k, \e$} satisfies
$$
||x'-\chi'||^2_2\leq (1+O(\e))||x_{\nsq\setminus [k]}||_2^2
$$
with at least $99/100$ probability over its internal randomness. The sample complexity is $2^{O(d^2)}\frac1{\e}  k\log N$, and the runtime complexity is at most $2^{O(d^2)}\frac1{\e}  k \log^{d+1} N.$
\end{lemma}
We give the intuition behind the proof here, as the argument is somewhat more delicate than the analysis of \textsc{RecoverAtConstSNR} in~\cite{IKP}, due to the $\ell_1$-SNR, rather than $\ell_2$-SNR assumption.
Specifically, if instead of $||(x-\chi)_{[2k]}||_1\leq O(\mu k)$ we had $||(x-\chi)_{[2k]}||_2^2\leq O(\mu^2 k)$, then it would be essentially sufficient to note that after a single hashing into about $k/(\e\alpha)$ buckets for a constant $\alpha\in (0, 1)$, every element $i\in [2k]$ is recovered with probability at least $1-O(\e\alpha)$, say, as it is enough to (on average) recover all but about an $\e$ fraction of coefficients. This would not be sufficient here since we only have a bound on the $\ell_1$ norm of the residual, and hence some elements can contribute much more $\ell_2$ norm than others. However, we are able to show that the probability that an element of the residual signal $x'_i$ is not recovered is bounded by
$O(\frac{\alpha \e \mu^2}{|x'_i|^2}+\frac{\alpha \e \mu}{|x'_i|})$, where the first term corresponds to contribution of tail noise and the second corresponds to the head elements. This bound implies that the total expected $\ell_2^2$ mass in the elements that are not recovered  is upper bounded by 
$\sum_{i\in [2k]} |x'_i|^2\cdot O(\frac{\alpha \e \mu^2}{|x'_i|^2}+\frac{\alpha \e \mu}{|x'_i|})\leq O(\e \mu^2 k+\e\mu \sum_{i\in [2k]} |x'_i|)=O(\e \mu^2 k)$, giving the result.

Finally, putting the results above together, we prove in section~\ref{sec:sfft}

\begin{theorem}\label{thm:main}
For any $\e>0$, $x\in \C^\nsq$ and any integer $k\geq 1$, if $R^*\geq ||x||_\infty/\mu=N^{O(1)}$, $\mu^2=O(||x_{\nsq\setminus [k]}||_2^2/k)$  and $\alpha>0$ is smaller than an absolute constant, \textsc{SparseFFT}$(\hat x, k, \e, R^*, \mu)$ solves the $\ell_2/\ell_2$ sparse recovery problem using $2^{O(d^2)} (k\log N \log\log N+\frac1{\e}k\log N)$ samples and 
$2^{O(d^2)} \frac1{\e}k \log^{d+3} N$ time with at least $98/100$ success probability.
\end{theorem}

\begin{algorithm}
\caption{\textsc{SparseFFT}($\hat x, k, \e, R^*, \mu$)}\label{alg:main-sublinear} 
\begin{algorithmic}[1] 
\Procedure{SparseFFT}{$\hat x, k, \e, R^*, \mu$}
\State $\chi^{(0)} \gets 0$ \Comment{in $\C^n$}.
\State $T \gets \log_{(\log^4 N)} R^*$
\State $F\gets 2d$
\State $B\leftarrow \GT\cdot k/\alpha^d$, $\alpha>0$ sufficiently small constant
\State $r_{max}\gets (C/\sqrt{\alpha})\log\log N, c_{max}\gets (C/\sqrt{\alpha})\log\log N$ for a sufficiently large constant $C>0$
\State $\H\gets \{\mathbf{0}_d\}$, $\Delta\gets 2^{\lfloor \frac1{2}\log_2 \log_2 n\rfloor}$ \Comment{$\mathbf{0}_d$ is the zero vector in dimension $d$}
\For {$g=1$ to $\lceil \log_\Delta n \rceil$}
\State $\H\gets \H\cup \bigcup_{s=1}^d n \Delta^{-g} \cdot \mathbf{e}_s$ \Comment{$\mathbf{e}_s$ is the unit vector in direction $s$}
\EndFor
\State $G\gets$ filter with $B$ buckets and sharpness $F$, as per Lemma~\ref{lm:filter-prop}
\For {$r=1$ to $r_{max}$}\Comment{Samples that will be used for location}
\State Choose $\Sigma_r\in \gl$, $q_r\in \nsq$ uniformly at random, let $\pi_r:=(\Sigma_r, q_r)$ and let $H_r:=(\pi_r, B, F)$
\State Let $\A_r\gets $ $C\log\log N$ elements of $\nsq\times \nsq$ sampled uniformly at random with replacement
\For {$\h\in \H$}
\State $m(\wh{x}, H_r, a\star(\one, \h))\gets \Call{HashToBins}{\hat x, 0, (H_r, a\star (\one, \h))}$ for all $a\in \A_r, \h\in \H$
\EndFor
\EndFor
\For{$t = 0, 1, \dotsc,T-1$}
\State $\chi' \gets \textsc{ReduceL1Norm}\left(\chi^{(t)}, k, \{m(\wh{x}, H_r, a\star (\one, \h))\}_{r=1, a\in \A_r, \h\in \H}^{r_{max}}, 4\mu (\log^4 n)^{T-t}, \mu\right)$ 
\State\Comment{Reduce $\ell_1$ norm of dominant elements in the residual signal}
\State $\nu'\gets (\log^4 N)(4\mu (\log^4 N)^{T-(t+1)}+20\mu)$ \Comment{Threshold}
\State $\chi'' \gets \Call{ReduceInfNorm}{\hat x,  \chi^{(t)}+\chi', 4k/(\log^4 N),  \nu', \nu'}$
\State\Comment{Reduce $\ell_\infty$ norm of spurious elements introduced by \textsc{ReduceL1Nom}}
\State $\chi^{(t+1)} \gets \chi^{(t)} + \chi'+\chi''$
\EndFor
\State $\chi' \gets \textsc{RecoverAtConstantSNR}(\hat x, \chi^{(T)}, 2k, \eps)$
\State \textbf{return} $\chi^{(T)} + \chi'$
\EndProcedure 
\end{algorithmic}
\end{algorithm}


\begin{algorithm}

\caption{\textsc{ReduceL1Norm}$\left(\wh{x}, \chi, k, \chi^{(t)}, k, \{m(\wh{x}, H_r, a\star (\one, \h))\}_{r=1, a\in \A_r, \h\in \H}^{r_{max}}, \nu, \mu\right)$}\label{alg:l1-norm-reduction} 
\begin{algorithmic}[1] 
\Procedure{ReduceL1Norm}{\mbox{$\wh{x}, \chi, k, \chi^{(t)}, k, \{m(\wh{x}, H_r, a\star (\one, \h))\}_{r=1, a\in \A_r, \h\in \H}^{r_{max}}, \nu, \mu$}}
\State $\chi^{(0)} \gets 0$ \Comment{in $\C^n$}
\State $B\gets  \GT\cdot k/\alpha^d$
\For{$t=0$ to $\log_2 (\log^4 N)$} 
\For {$r=1$ to $r_{max}$}
\State $L_r \gets \textsc{LocateSignal}\left(\chi+\chi^{(t)}, k, \{m(\wh{x}, H_r, a\star (\one, \h))\}_{r=1, a\in \A_r, \h\in \H}^{r_{max}}\right)$ 
\EndFor
\State $L\gets \bigcup_{r=1}^{r_{max}} L_r$
\State $\chi' \gets \Call{EstimateValues}{\wh{x},  \chi+\chi^{(t)}, L, 4k, 1, \frac1{1000} \nu 2^{-t}+4\mu, C(\log\log N+d^2+\log (B/k))}$
\State $\chi^{(t+1)} \gets \chi^{(t)} + \chi'$
\EndFor
\State \textbf{return} $\chi+\chi^{(T)}$
\EndProcedure 
\end{algorithmic}
\end{algorithm}

\section{Organization}\label{sec:org}
The rest of the paper is organized as follows. In section~\ref{sec:l11} we set up notation necessary for the analysis of \textsc{LocateSignal}, and specifically for a proof of Theorem~\ref{thm:l1-res-loc}, as well as prove some basic claims. In section~\ref{sec:l1} we prove Theorem~\ref{thm:l1-res-loc}. In section~\ref{sec:reduce-l1} we prove performance guarantees for \textsc{ReduceL1Norm} (Lemma~\ref{lm:reduce-l1-norm}), then combine them with Lemma~\ref{lm:linf} to prove that the main loop in Algorithm~\ref{alg:main-sublinear} reduces $\ell_1$ norm of the head elements. We then conclude with a proof of correctness for Algorithm~\ref{alg:main-sublinear}. Section~\ref{sec:linf} is devoted to analyzing the \textsc{ReduceInfNorm} procedure, and section~\ref{sec:const-snr} is devoted to analyzing the \textsc{RecoverAtConstantSNR} procedure. Some useful lemmas are gathered in section~\ref{sec:utils}, and section~\ref{sec:semiequi} describes the algorithm for semiequispaced Fourier transform that we use to update our samples with the residual signal. Appendix~\ref{app:A} contains proofs omitted from the main body of the paper.

\section{Analysis of \textsc{LocateSignal}: main definitions and basic claims}\label{sec:l11}

In this section we state our main signal location primitive, \textsc{LocateSignal} (Algorithm~\ref{alg:location}). Given a sequence of measurements $m(\wh{x}, H_r, a\star (\one, \h))\}_{a\in \A_r, \h\in \H}, r=1,\ldots, r_{max}$ a  signal $\wh{x}\in \mathbb{C}^\nsq$ and a partially recovered signal $\chi\in \mathbb{C}^\nsq$, \textsc{LocateSignal} outputs a list of locations $L\subseteq \nsq$ that, as we show below in Theorem~\ref{thm:l1-res-loc} (see section~\ref{sec:l1}), contains the elements of $x$ that contribute most of its $\ell_1$ mass. An important feature of \textsc{LocateSignal} is that it is an entirely deterministic procedure, giving recovery guarantees for any signal $x$ and any partially recovered signal $\chi$.  As Theorem~\ref{thm:l1-res-loc} shows, however, these guarantees are strongest when most of the mass of the residual $x-\chi$ resides on elements in $\nsq$ that are {\em isolated with respect to most hashings $H_1,\ldots, H_{r_{max}}$} used for measurements. This flexibility is crucial for our analysis, and is exactly what allows us to reuse measurements and thereby achieve near-optimal sample complexity.

In the rest of this section we first state Algorithm~\ref{alg:location}, and then derive useful characterization of elements $i$ of the input signal $(x-\chi)_i$ that are successfully located by \textsc{LocateSignal}. The main result of this section is Corollary~\ref{cor:loc}. This comes down to bounding, for a given input signal $x$ and partially recovered signal $\chi$, the expected $\ell_1$ norm of the noise contributed to the process of locating heavy hitters in a call to \Call{LocateSignal}{$\wh{x}, \chi, H, \{m(\wh{x}, H, a\star (\one, \h))\}_{a\in \A, \h\in \H}$} by {\bf (a)} the tail of the original signal $x$ (tail noise $e^{tail}$) and {\bf (b)} the heavy hitters and false positives (heavy hitter noise $e^{head}$). It is useful to note that unlike in ~\cite{IK14a}, we cannot expect the tail of the signal to not change, but rather need to control this change. 

In what follows we derive useful conditions under which an element $i\in \nsq$ is identified by \textsc{LocateSignal}. Let $S\subseteq \nsq$ be any set of size at most $2k$, and let $\mu$ be such that $x_{\nsq\setminus S}\leq \mu$. Note that this fits the definition of $S$ given in~\eqref{eq:s-def-l1} (but other instantiations are possible, and will be used later in section~\ref{sec:const-snr}).

Consider a call to 
$$
\Call{LocateSignal}{\chi, H, \{m(\wh{x}, H, a\star (\one, \h)\}_{a\in \A, \h\in \H}}.
$$ 
For each $a\in \A$ and fixed $\h\in \H$ we let $z:=a\star (\one, \h)\in \nsq$ to simplify notation. The measurement vectors $m:=m(\wh{x'}, H, z)$  computed in \textsc{LocateSignal} satisfy, for every $i\in S$,  (by Lemma~\ref{l:hashtobins})
  \[
  m_{h(i)}  = \sum_{j\in \nsq} G_{o_i(j)}x'_j \omega^{z^T \Sigma j}+\Delta_{h(i), z},
  \]
where $\Delta$ corresponds to polynomially small estimation noise due to approximate computation of the Fourier transform, and the filter $G_{o_i(j)}$ is the filter corresponding to hashing $H$. In particular, for each hashing $H$ and  parameter $a\in \nsq$ one has:
\begin{equation*}
\begin{split}
G_{o_i(i)}^{-1}m_{h(i)} \omega^{-z^T \Sigma i} = x'_i + G_{o_i(i)}^{-1}\sum_{j\in \nsq\setminus \{i\}} G_{o_i(j)}x'_j \omega^{z^T \Sigma (j-i)}+G_{o_i(i)}^{-1}\Delta_{h(i), z}\omega^{-z^T \Sigma i}
\end{split}
\end{equation*}
It is useful to represent the residual signal $x$ as a sum of three terms: $x'=(x-\chi)_S-\chi_{\nsq\setminus S}+x_{\nsq\setminus S}$, where the first term is the residual signal coming from the `heavy' elements in $S$, the second corresponds to false positives, or spurious elements discovered and erroneously subtracted by the algorithm, and the third corresponds to the tail of the signal. Similarly, we bound the noise contributed by the first two (head elements and false positives) and the third (tail noise) parts of the residual signal to the location process separately. For each $i\in S$ we write
\begin{equation}\label{eq:uexp1}
\begin{split}
&G_{o_i(i)}^{-1}m_{h(i)}\omega^{-z^T \Sigma i}=x'_i \\
&+G_{o_i(i)}^{-1}\cdot \left[\sum_{j\in S\setminus \{i\}} G_{o_i(j)}x'_j \omega^{z^T \Sigma (j-i)}-\sum_{j\in \nsq\setminus S} G_{o_i(j)}\chi_j \omega^{z^T \Sigma (j-i)}\right]\text{~~~~(head elements and false positives)}\\
 &+G_{o_i(i)}^{-1}\cdot \sum_{j\in \nsq\setminus S} G_{o_i(j)}x_j \omega^{z^T \Sigma (j-i)}\text{~~~~~~~~~~~~~~~~~~~~~~~~~~~~~~~~~~~~~~~~~~~~~~~~~~~~~~~~~(tail noise)}\\
&+G_{o_i(i)}^{-1}\cdot \Delta_{h(i)}\omega^{-z^T \Sigma i}.
\end{split}
\end{equation}

\paragraph{Noise from heavy hitters.} The first term in \eqref{eq:uexp1} corresponds to noise from $(x-\chi)_{S\setminus \{i\}}-\chi_{\nsq\setminus (S\setminus \{i\})}$, i.e. noise from heavy hitters and false positives. For every $i\in S$, hashing $H$  we let
\begin{equation}\label{eq:eh-pi}
e^{head}_{i}(H, x, \chi):=G_{o_i(i)}^{-1}\cdot \sum_{j\in S\setminus \{i\}} G_{o_i(j)} |y_j|,\text{~~~~where~}y=(x-\chi)_{S}-\chi_{\nsq\setminus S}.\\
\end{equation}
We thus get that $e^{head}_i(H, x, \chi)$ upper bounds the absolute value of the first error term in~\eqref{eq:uexp1}. Note that $G\geq 0$ by Lemma~\ref{lm:filter-prop} as long as $F$ is even, which is the setting that we are in. If $e^{head}_{i}(H, x, \chi)$ is large, \textsc{LocateSignal} may not be able to locate $i$ using measurements of the residual signal $x-\chi$ taken with hashing $H$. However, the noise in other hashings may be smaller, allowing recovery. In order to reflect this fact we define, for a sequence of hashings $H_1,\ldots, H_r$ and a signal $y\in \C^\nsq$
\begin{equation}\label{eq:eh}
e^{head}_{i}(\{H_r\}, x, \chi):=\quant^{1/5}_r e^{head}_i(H_r, x, \chi),
\end{equation}
where for a list of reals $u_1,\ldots, u_s$ and a number $f\in (0, 1)$ we let $\quant^{f}(u_1,\ldots, u_s)$ denote the $\lceil f \cdot s\rceil$-th largest element of $u_1,\ldots, u_s$.  

\paragraph{Tail noise.} To capture the second term in \eqref{eq:uexp1} (corresponding to tail noise), we define, for any $i\in S, z\in \nsq, \h\in \H$, permutation $\pi=(\Sigma, q)$ and hashing $H=(\pi, B, F)$
\begin{equation}\label{eq:et-pi}
e^{tail}_i(H, z, x):=\left|G_{o_i(i)}^{-1}\cdot \sum_{j\in \nsq\setminus S} G_{o_i(j)}x_j \omega^{z^T \Sigma (j-i)}\right|.
\end{equation}

With this definition in place $e^{tail}_i(H, z, x)$ upper bounds the second term in~\eqref{eq:uexp1}. As our algorithm uses several values of $a\in \A_r\subseteq \nsq\times \nsq$ to perform location, a more robust version of $e^{tail}_i(H, z)$ will be useful. To that effect we let for any $\mathcal{Z}\subseteq \nsq$ (we will later use $\mathcal{Z}=\A_r\star (\one, \h)$ for various $\h\in \H$)
\begin{equation}\label{eq:et-pi-a-h}
e^{tail}_i(H, \mathcal{Z}, x):=\quant^{1/5}_{z\in \mathcal{Z}} \left|G_{o_i(i)}^{-1}\cdot \sum_{j\in \nsq\setminus S} G_{o_i(j)}x_j \omega^{z^T \Sigma (j-i)}\right|.
\end{equation}
Note that the algorithm first selects sets $\A_r\subseteq \nsq\times \nsq$, and then access the signal at locations $\A_r\star (\one, \h), \h\in \H$.

The definition of $e^{tail}_i(H, \A\star (\one, \h), x)$ for a fixed $\h\in \H$ allows us to capture the amount of noise that our measurements that use $H$ suffer from for locating a specific set of bits of $\Sigma i$. Since the algorithm requires all $\h\in \H$ to be not too noisy in order to succeed (see precondition 2 of Lemma~\ref{lm:loc}), it is convenient to introduce notation that captures this. We define
\begin{equation}\label{eq:et-pi-a}
e^{tail}_i(H, \A, x):=40\mu_{H, i}(x)+\sum_{\h\in \H} \left|e^{tail}_i(H, \A\star (\one, \h), x)-40\mu_{H, i}(x)\right|_+
\end{equation}
where for any $\eta\in \mathbb{R}$ one has $|\eta|_+=\eta$ if $\eta>0$ and $|\eta|_+=0$ otherwise. 

The following definition is useful for bounding the norm of elements $i\in S$ that are not discovered by several calls to \textsc{LocateSignal} on a sequence of hashings $\{H_r\}$. For a sequence of measurement patterns $\{H_r, \A_r\}$ we let
\begin{equation}\label{eq:et}
e^{tail}_i(\{H_r, \A_r\}, x):=\quant^{1/5}_r e^{tail}_i(H_r, \A_r, x).
\end{equation}

Finally, for any $S\subseteq \nsq$ we let 
$$
e^{head}_S(\cdot):=\sum_{i\in S} e^{head}_i(\cdot)\text{~~~and~~~}e^{tail}_S(\cdot):=\sum_{i\in S} e^{tail}_i(\cdot),
$$
where $\cdot$ stands for any set of parameters as above.

Equipped with the definitions above,  we now prove the following lemma, which yields sufficient conditions for recovery of elements $i\in S$ in \textsc{LocateSignal} in terms of $e^{head}$ and $e^{tail}$.
\begin{lemma}\label{lm:loc}
Let $H=(\pi, B, R)$ be a hashing, and let $\A\subseteq \nsq\times \nsq$. Then for every $S\subseteq \nsq$ and for every $x, \chi\in \C^{\nsq}$ and $x'=x-\chi$, the following conditions hold.
Let $L$ denote the output of 
$$
\Call{LocateSignal}{\chi, H, \{m(\wh{x}, H, a\star (\one, \h))\}_{a\in \A, \h\in \H}}.
$$

Then for any $i\in S$ such that $|x'_i|>N^{-\Omega(c)}$, if there exists $r\in [1:r_{max}]$ such that 
\begin{enumerate}
\item  $e^{head}_{i}(H, x')<|x'_i|/20$;
\item  $e^{tail}_{i}(H, \A\star (\one, \h), x')< |x'_i|/20$ for all $\h\in \H$;
\item for every $s\in [1:d]$ the set $\A\star(\zero, \mathbf{e}_s)$ is balanced in coordinate $s$ (as per Definition~\ref{def:balance}),
\end{enumerate}
 then $i\in L$. The time taken by the invocation of \textsc{LocateSignal} is $O(B\cdot \log^{d+1} N)$.
\end{lemma}
\begin{proof}
We show that each coordinate $s=1,\ldots, d$ of $\Sigma i$ is successfully recovered in \textsc{LocateSignal}. Let $q=\Sigma i$ for convenience.
Fix $s\in [1:d]$. We show by induction on $g=0,\ldots, \log_{\Delta} n-1$ that after the $g$-th iteration of lines~6-10 of Algorithm~\ref{alg:location} we have that ${\bf f}_s$ coincides with ${\bf q}_s$ on the bottom $g\cdot \log_2 \Delta$ bits, i.e. ${\bf f}_s-{\bf q}_s= 0 \mod \Delta^g$ (note that we trivially have ${\bf f}_s< \Delta^g$ after iteration $g$).

The {\bf base} of the induction is trivial and is provided by $g=0$.  
We now show the {\bf inductive step}. Assume by the inductive hypothesis that ${\bf f}_s-{\bf q}_s= 0 \mod \Delta^{g-1}$, so that
${\bf q}_s={\bf f}_s+\Delta^{g-1}(r_0+\Delta r_1+\Delta^2 r_2+\ldots)$ for some sequence $r_0,r_1,\ldots$, $0\leq r_j<\Delta$. Thus,  $(r_0, r_1,\ldots)$ is the expansion of $({\bf q}_s-{\bf f}_s)/\Delta^{g-1}$ base $\Delta$, and $r_0$ is the least significant digit. We now show that $r_0$ is the unique value of $r$ that satisfies the conditions of lines~8-10 of Algorithm~\ref{alg:location}. 

First, we have by~\eqref{eq:uexp1} together with \eqref{eq:eh-pi} and ~\eqref{eq:et-pi} one has for each $a\in \A$ and $\h\in \H$
\begin{equation*}
\begin{split}
\left|m_{h(i)}(\wh{x'}, H, a\star (\one, \h))- G_{o_i(i)} x'_i \omega^{((a\star(\one, \h))^T {\bf q}}\right|&\leq  e^{head}_i(H, x, \chi)+e^{tail}_i(H, a\star(\one, \h), x)+N^{-\Omega(c)}.
\end{split}
\end{equation*}
Since $\zero\in \H$, we also have for all $a\in \A$
\begin{equation*}
\begin{split}
\left|m_{h(i)}(\wh{x'}, H, a\star (\one, \zero))- G_{o_i(i)} x'_i \omega^{(a\star(\one, \zero))^T {\bf q}}\right|&\leq  e^{head}_i(H, x, \chi)+e^{tail}_i(H, a\star (\one, \zero), x)+N^{-\Omega(c)},
\end{split}
\end{equation*}
where the $N^{-\Omega(c)}$ terms correspond to polynomially small error from approximate computation of the Fourier transform via Lemma~\ref{c:semiequi1}.

Let $j:=h(i)$. We will show that $i$ is recovered from bucket $j$. The bounds above imply that 
\begin{equation}\label{eq:gergergre}
\begin{split}
\frac{m_j(\wh{x'}, H, a\star (\one, \h))}{m_j(\wh{x'}, H, a\star (\one, \zero))}=\frac{x'_i \omega^{(a\star (\one, \h))^T {\bf q}}+E'}{x'_i \omega^{(a\star (\one, \zero))^T {\bf q}}+E''}
\end{split}
\end{equation}
for some $E', E''$ satisfying $|E'|\leq e^{head}_i(H, x, \chi)+e^{tail}_i(H, a\star (\one, \h), x)+N^{-\Omega(c)}$ and $|E''|\leq e^{head}_i(H, x, \chi)+e^{tail}_i(H, a\star (\one, \zero))+N^{-\Omega(c)}$. For all but $1/5$ fraction of $a\in \A$ we have by definition of $e^{tail}$ (see~\eqref{eq:et-pi-a-h}) that {\bf both} 
\begin{equation}\label{eq:etail-bounds-eq-1}
e^{tail}_i(H, a\star (\one, \h), x)\leq e^{tail}_i(H, \A\star (\one, \h), x)\leq |x'_i|/20
\end{equation}
 and 
\begin{equation}\label{eq:etail-bounds-eq-2}
 e^{tail}_i(H, a\star (\one, \zero)\leq e^{tail}_i(H, \A\star (\one, \zero), x)\leq |x'_i|/20.
\end{equation}
In particular, we can rewrite ~\eqref{eq:gergergre} as
\begin{equation}\label{eq:gergergre-2}
\begin{split}
\frac{m_j(\wh{x'}, H, a\star (\one, \h))}{m_j(\wh{x'}, H, a\star (\one, \zero))}&=\frac{x'_i \omega^{(a\star (\one, \h))^T {\bf q}}+E'}{x'_i \omega^{(a\star (\one, \zero))^T {\bf q}}+E''}\\
&=\frac{\omega^{(a\star (\one, \h))^T {\bf q}}}{\omega^{(a\star (\one, \zero))^T {\bf q}}}\cdot\xi\text{~~~where~~}\xi=\frac{1+\omega^{-(a\star (\one, \h))^T {\bf q}}E'/x_i'}{1+\omega^{-(a\star (\one, \zero))^T{\bf q}} E''/x_i'}\\
&=\omega^{(a\star (\one, \h))^T {\bf q}-(a\star (\one, \zero))^T {\bf q}}\cdot\xi\\
&=\omega^{(a\star (\zero, \h))^T {\bf q}}\cdot\xi.\\
\end{split}
\end{equation}

Let $\A^*\subseteq \A$ denote the set of values of $a\in \A$ that satisfy the bounds~\eqref{eq:etail-bounds-eq-1} and~\eqref{eq:etail-bounds-eq-2} above.
We thus have for  $a\in \A^*$, combining ~\eqref{eq:gergergre-2} with assumptions {\bf 1-2} of the lemma, that
\begin{equation}\label{eq:bound-1-oigb344tg32t}
|E'|/x_i'\leq  (2/20)+1/N^{-\Omega(c)}\leq 1/8\text{~~~and~~~~}|E''|/x_i'\leq  (2/20)+1/N^{-\Omega(c)}\leq 1/8
\end{equation}
for sufficiently large $N$, where $O(c)$ is the word precision of our semi-equispaced Fourier transform computation. Note that we used the assumption that $|x'_i|\geq N^{-\Omega(c)}$.

Writing $a=(\alpha, \beta)\in\nsq\times \nsq$, we have by~\eqref{eq:gergergre-2} that $\frac{m_j(\wh{x'}, H, a\star (\one, \h))}{m_j(\wh{x'}, H, a\star (\one, \zero))}=\omega^{((\alpha, \beta)\star (\zero, \h))^T {\bf q}}\cdot\xi$, and since $\h^T{\bf q}=n\Delta^{-g}{\bf q}_s$ when $\h=n \Delta^{-g} {\bf e}_s$ (as in line~8 of Algorithm~\ref{alg:location}), we get 
$$
\frac{m_j(\wh{x'}, H, a\star (\one, \h))}{m_j(\wh{x'}, H, a\star (\one, \zero))}=\omega^{(a\star (\zero, \h))^T {\bf q}}\cdot\xi=\omega^{n\Delta^{-g} \beta_s {\bf q}_s}\cdot\xi=\omega^{n\Delta^{-g} \beta_s {\bf q}_s}+\omega^{n\Delta^{-g} \beta_s {\bf q}_s}(\xi-1). 
$$
We analyze the first term now, and will show later that the second term is small. Since ${\bf q}_s={\bf f}_s+\Delta^{g-1}(r_0+\Delta r_1+\Delta^2 r_2+\ldots)$ by the inductive hypothesis, we have, substituting the first term above into the expression in line~10 of Algorithm~\ref{alg:location},
\begin{equation*}
\begin{split}
\omega_\Delta^{-r\cdot \beta_s}\cdot \omega^{-n\Delta^{-g}{\bf f_s}\cdot \beta_s}\cdot \omega^{n\Delta^{-g} \beta_s {\bf q}_s}&=\omega_\Delta^{-r\cdot \beta_s}\cdot \omega^{n\Delta^{-g}({\bf q}_s-{\bf f_s})\cdot \beta_s}\\
&=\omega_\Delta^{-r\cdot \beta_s}\cdot \omega^{n\Delta^{-g}(\Delta^{g-1}(r_0+\Delta r_1+\Delta^2 r_2+\ldots))\cdot \beta_s}\\
&=\omega_\Delta^{-r\cdot \beta_s}\cdot \omega^{(n/\Delta)\cdot (r_0+\Delta r_1+\Delta^2 r_2+\ldots)\cdot \beta_s}\\
&=\omega_\Delta^{-r\cdot \beta_s}\cdot \omega_{\Delta}^{r_0\cdot \beta_s}\\
&=\omega_\Delta^{(-r+r_0)\cdot \beta_s}.
\end{split}
\end{equation*}
We used the fact that $\omega^{n/\Delta}=e^{2\pi i (n/\Delta)/n}=e^{2\pi i/\Delta}=\omega_\Delta$ and $(\omega_{\Delta})^\Delta=1$. Thus, we have
\begin{equation}\label{eq:92hg34grggggdds}
\omega_{\Delta}^{-r\cdot \beta_s}\omega^{-(n2^{-g}{\bf f_s})\cdot \beta_s}\frac{m_j(\wh{x'}, H, a\star (\one, \h))}{m_j(\wh{x'}, H, a\star (\one, \zero))}=\omega_\Delta^{(-r+r_0)\cdot \beta_s}+\omega_\Delta^{(-r+r_0)\cdot \beta_s}(\xi-1).
\end{equation}

We now consider two cases. First suppose that $r=r_0$. Then $\omega_\Delta^{(-r+r_0)\cdot \beta_s}=1$, and it remains to note that  by~\eqref{eq:bound-1-oigb344tg32t} we have $|\xi-1|\leq \frac{1+1/8}{1-1/8}-1\leq 2/7< 1/3$.
Thus every $a\in \A^*$ passes the test in line~9 of Algorithm~\ref{alg:location}. Since $|\A^*|\geq (4/5)|\A|>(3/5)|\A|$ by the argument above, we have that $r_0$ passes the test in line~9. It remains to show that $r_0$ is the unique element in $0,\ldots, \Delta-1$ that passes this test.

Now suppose that $r\neq r_0$. Then by the assumption that $\A\star (\zero, \mathbf{e}_s)$ is balanced (assumption {\bf 3} of the lemma) at least $49/100$ fraction of $\omega_\Delta^{(-r+r_0)\cdot \beta_s}$ have negative real part.  This means that for at least $49/100$ of $a\in \A$ we have using triangle inequality
\begin{equation*}
\begin{split}
\left|\left[\omega_\Delta^{(-r+r_0)\cdot \beta_s}+\omega_\Delta^{(-r+r_0)\cdot \beta_s}(\xi-1)\right]-1\right|&\geq \left|\omega_\Delta^{(-r+r_0)\cdot \beta_s}-1\right|-\left|\omega_\Delta^{(-r+r_0)\cdot \beta_s}(\xi-1)\right|\\
&\geq \left|\mathbf{i}-1\right|-1/3\\
&\geq \sqrt{2}-1/3> 1/3,
\end{split}
\end{equation*}
and hence the condition in line~9 of Algorithm~\ref{alg:location} is not satisfied for any $r\neq r_0$. This shows that location is successful and completes the proof of correctness.

Runtime bounds follow by noting that \textsc{LocateSignal} recovers $d$ coordinates with $\log n$ bits per coordinate. Coordinates are recovered in batches of $\log \Delta$ bits, and the time taken is bounded by $B\cdot d(\log_\Delta n)\Delta\leq B (\log N)^{3/2}$. Updating the measurements using semi-equispaced FFT takes $B\log^{d+1} N$ time.
\end{proof}

We also get an immediate corollary of Lemma~\ref{lm:loc}. The corollary is crucial to our proof of Theorem~\ref{thm:l1-res-loc} (the main result about efficiency of \textsc{LocateSignal}) in the next section.
\begin{corollary}\label{cor:loc}
For any integer $r_{max}\geq 1$,  for any sequence of $r_{max}$ hashings $H_r=(\pi_r, B, R), r\in [1:r_{max}]$ and evaluation points $\A_r\subseteq \nsq\times \nsq$,  for every $S\subseteq \nsq$ and for every $x, \chi\in \C^{\nsq}, x':=x-\chi$, the following conditions hold.
If for each $r\in [1:r_{max}]$ $L_r\subseteq \nsq$ denotes the output of \Call{LocateSignal}{$\wh{x}, \chi, H_r, \{m(\wh{x}, H_r, a\star (\one, \h))\}_{a\in \A_r, \h\in \H}$}, $L=\bigcup_{r=1}^{r_{max}} L_r$, and the sets $\A_r\star(\zero, \h)$ are balanced for all $\h\in \H$ and $r\in [1:r_{max}]$, then
\begin{equation}
||x'_{S\setminus L}||_1\leq 20 ||e^{head}_S(\{H_r\}, x, \chi)||_1+20 ||e^{tail}_S(\{H_r, \A_r\}, x)||_1+|S|\cdot N^{-\Omega(c)}.\tag{*}
\end{equation}
Furthermore, every element $i\in S$ such that 
\begin{equation}
|x'_i|>20 (e^{head}_i(\{H_r\}, x, \chi)+e^{tail}_i(\{H_r, \A_r\}, x))+N^{-\Omega(c)}\tag{**}
\end{equation}
belongs to $L$.
\end{corollary}
\begin{proof}
Suppose that $i\in S$ fails to be located in any of the $R$ calls, and $|x'_i|\geq N^{-\Omega(c)}$. By Lemma~\ref{lm:loc} and the assumption that $\A_r\star(\zero, \h)$ is balanced for all $\h\in \H$ and $r\in [1:r_{max}]$ this means that for at least one half of values $r\in [1:r_{max}]$  either {\bf (A)} $e^{head}_{i}(H_r, x, \chi)\geq |x_i|/20$ or {\bf (B)} $e^{tail}_{i}(H_r, \A_r\star (\one, \h), x)> |x_i|/20$ for at least one $\h\in \H$. We consider these two cases separately.

\paragraph{Case (A).} In this case we have $e^{head}_{i}(H_s, x, \chi)\geq |x_i|/20$ for at least one half of $r\in [1:r_{max}]$, so 
in particular $e^{head}_i(\{H_r\}, x, \chi)\geq \text{quant}^{1/5}_r e^{head}_{i}(H_r, x, \chi)\geq |x'_i|/20$.

\paragraph{Case (B).} Suppose that $e^{tail}_{i}(H_r, \A_r\star (\one, \h), x)> |x'_i|/20$ for some $\h=\h(r)\in \H$ for at least one half of $r\in [1:r_{max}]$ (denote this set by $Q\subseteq [1:r_{max}]$). We then have
\begin{equation*}
\begin{split}
e^{tail}_i(\{H_r, \A_r\}, x)&=\quant^{1/5}_{r\in [1:r_{max}]} e^{tail}_i(H_r, \A_r, x)\\
&=\quant^{1/5}_{r\in [1:r_{max}]} \left[40\mu_{H_r, i}(x)+\sum_{\h\in \H} \left|e^{tail}_i(H_r, \A_r\star (\one, \h), x)-40\mu_{H_r, i}(x)\right|_+\right]\\
&\geq \min_{r\in Q} \left[40\mu_{H_r, i}(x)+\left|e^{tail}_i(H_r, \A_r\star (\one, \h(r)), x)-40\mu_{H_r, i}(x)\right|_+\right]\\
&\geq \min_{r\in Q} e^{tail}_i(H_r, \A_r\star (\one, \h(r)), x)\\
&\geq |x'_i|/20
\end{split}
\end{equation*}
 as required. This completes the proof of {\bf (*)} as well as {\bf (**)}.
\end{proof}

\section{Analysis of \textsc{LocateSignal}: bounding $\ell_1$ norm of undiscovered elements}\label{sec:l1}

The main result of this section is Theorem~\ref{thm:l1-res-loc}, which is our main tool for showing efficiency of \textsc{LocateSignal}. Theorem~\ref{thm:l1-res-loc} applies to the following setting. Fix a set $S\subseteq \nsq$ and a set of hashings $H_1,\ldots, H_{r_{max}}$, and let $S^*\subseteq S$ denote the set of elements of $S$ that are not isolated with respect to most of these hashings $H_1,\ldots, H_{r_{max}}$. Theorem~\ref{thm:l1-res-loc} shows that for any signal $x$ and partially recovered signal $\chi$, if $L$ denotes the output list of an invocation of \textsc{LocateSignal} on the pair $(x, \chi)$ with hashings $H_1,\ldots, H_{r_{max}}$, then 
the $\ell_1$ norm of elements of the residual $(x-\chi)_S$ that are not discovered by \textsc{LocateSignal} can be bounded by a function of the amount of $\ell_1$ mass of the residual that fell outside of the `good' set $S\setminus S^*$, plus the `noise level' $\mu\geq ||x_{\nsq\setminus S}||_\infty$ times $k$.

If we think of applying Theorem~\ref{thm:l1-res-loc} iteratively, we intuitively get that the fixed set of measurements with hashings $\{H_r\}$ allows us to always reduce the $\ell_1$ norm of the residual $x'=x-\chi$ on the `good' set $S\setminus S^*$ to about the amount of mass that is located outside of this good set.

\noindent{\em {\bf Theorem~\ref{thm:l1-res-loc}}
There exist absolute constants $C_1, C_2, C_3>0$ such that for any $x, \chi\in \C^N$ and residual signal $x'=x-\chi$ the following conditions hold. \defsk   Let $S^*\subseteq S$ denote the set of elements that are not isolated with respect to at least a $\sqrt{\alpha}$ fraction of hashings $\{H_r\}_{r=1}^{r_{max}}$.   Suppose that for every $s\in [1:d]$ the sets $\A_r\star(\zero, \mathbf{e}_s)$ are balanced (as per Definition~\ref{def:balance}),  $r=1,\ldots, r_{max}$, and the exponent $F$ of the filter $G$ is even and satisfies $F\geq 2d$.  
Let 
$$
L=\bigcup_{r=1}^{r_{max}}\Call{LocateSignal}{\chi, H_r, \{m(\wh{x}, H_r, a\star (\one, \h)\}_{a\in \A_r, \h\in \H_r}}.
$$

Then if $r_{max}, c_{max}\geq (C_1/\sqrt{\alpha})\log\log N$, one has
$$
||x'_{S\setminus S^*\setminus L}||_1\leq (C_2\alpha)^{d/2} ||x'_S||_1+C_3^{d^2}(||\chi_{\nsq\setminus S}||_1+||x'_{S^*}||_1)+4\mu |S|.
$$
}

As we will show later, Theorem~\ref{thm:l1-res-loc} can be used to show that (assuming perfect estimation) invoking \textsc{LocateSignal} repeatedly allows one to reduce to $\ell_1$ norm of the head elements down to essentially
$$
||x'_{S^*}||_1+||\chi_{\nsq\setminus S}||_1,
$$
i.e. the $\ell_1$ norm of the elements that are not well isolated and the set of new elements created by the process due to false positives in location. 
In what follows we derive bounds on $||e^{head}||_1$ (in section~\ref{sec:h-noise}) and $||e^{tail}||_1$ (in section~\ref{sec:tail-noise}) that lead to a proof of Theorem~\ref{thm:l1-res-loc}.

\subsection{Bounding noise from heavy hitters}\label{sec:h-noise}

We first derive bounds on noise from heavy hitters that a single hashing $H$ results in, i.e. $e^{head}(H, x)$,  (see Lemma~\ref{lm:l1b-single-pi}), and then use these bounds to bound $e^{head}(\{H\}, x)$ (see Lemma~\ref{lm:eh-s-sstar}). These bounds, together with upper bounds on contribution of tail noise from the next section, then lead to a proof of Theorem~\ref{thm:l1-res-loc}.

\begin{lemma}\label{lm:l1b-single-pi}
Let $x, \chi\in \C^N, x'=x-\chi$. \defsk
Let $\pi=(\Sigma, q)$ be a permutation, let $H=(\pi, B, F), F\geq 2d$ be a hashing into $B$ buckets and filter $G$ with sharpness $F$. Let $S^*_H\subseteq S$ denote the set of elements $i\in S$ that are not isolated under $H$. Then one has, for $e^{head}$ defined with respect to $S$, 
$$
||e^{head}_{S\setminus S^*_H}(H, x, \chi)||_1\leq 2^{O(d)} \alpha^{d/2} ||x'_{S\setminus S^*_H}||_1+\GS \cdot 2^{O(d)} (||x'_{S^*}||_1+||\chi_{\nsq\setminus S}||_1).
$$
Furthermore, if $\chi_{\nsq\setminus S}=0$ and $S^*_H=\emptyset$, then one has $||e^{head}_{S}(H, x, \chi)||_\infty\leq 2^{O(d)}\alpha^{d/2} ||x'_S||_\infty$.
\end{lemma}
\begin{proof}
By \eqref{eq:eh-pi} for $i\in S\setminus S^*_H$
\begin{equation}\label{eq:asdf}
\begin{split}
e^{head}_{i}(H, x')&= |G_{o_i(i)}^{-1}|\cdot \sum_{j\in S\setminus S^*_H \setminus  \{i\}} |G_{o_i(j)}|x'_j|\text{~~~~~~~~~~~~~~~~~~~~~~~~~~~~~~~~~~(isolated head elements)}\\
&+|G_{o_i(i)}^{-1}|\cdot \left[\sum_{j\in S^*_H} |G_{o_i(j)}|x'_j|+\sum_{j\in \nsq\setminus S} |G_{o_i(j)}||\chi_j|\right]\text{~~~(non-isolated head elements and false positives)}\\
&=|G_{o_i(i)}^{-1}|\cdot (A_1(i)+A_2(i)).
\end{split}
\end{equation}
Let $A_1:=\sum_{i\in S\setminus S^*_H} A_1(i), A_2:=\sum_{i\in S\setminus S^*_H} A_2(i)$.

We bound $A_1$ and $A_2$ separately. 

\paragraph{Bounding $A_1$.} We start with a convenient upper bound on $A_1$:
\begin{equation}\label{eq:2pogh}
\begin{split}
A_1=&\sum_{i\in S\setminus S^*_H} \sum_{j\in S\setminus S^*_H\setminus  \{i\}} |G_{o_i(j)}||x'_j|~~~~~~~~~~~~~~~~~~~~~~~~~~~~~~~~~~~~~~~~\text{(recall that $o_i(j)=\pi(j)-(n/b)h(i)$)}\\
& =\sum_{t\geq 0} \sum_{i\in S\setminus S^*_H} \sum_{\substack{j\in S\setminus S^*_H\setminus \{i\}\text{~s.t.~}\\ ||\pi(j)-\pi(i)||_\infty\in (n/b)\cdot [2^t-1, 2^{t+1}-1)}} |G_{o_i(j)}||x'_j|, \text{~~~(consider all scales $t\geq 0$)}\\
& \leq  \sum_{t\geq 0} \sum_{i\in S\setminus S^*_H} \max_{||\pi(j)-\pi(i)||_\infty\geq (n/b)\cdot (2^t-1)} G_{o_i(j)}\cdot \sum_{\substack{j\in S\setminus S^*_H\setminus \{i\}\text{~s.t.~}\\ ||\pi(j)-\pi(i)||_\infty\leq (n/b)\cdot (2^{t+1}-1)}} |x'_j|\\
&= \sum_{j\in S\setminus S^*_H} |x'_j|\cdot \sum_{t\geq 0} \max_{\substack{||\pi(j)-\pi(i)||_\infty\geq \\(n/b)\cdot (2^t-1)}} G_{o_i(j)}\cdot \left|\left\{i\in S\setminus S^*_H\setminus \{j\} \text{~s.t.~}||\pi(j)-\pi(i)||_\infty\leq (n/b)\cdot (2^{t+1}-1)\right\}\right|.
\end{split}
\end{equation}
Note that in the first line we summed, over all $i\in S\setminus S^*_H$  (i.e. all isolated $i$), the contributions of all other $i\in S$ to the noise in their buckets. We need to bound the first line in terms of $||x'_{S\setminus S^*_H}||_1$. For that, we first classified all $j\in S\setminus S^*_H$ according to the $\ell_\infty$ distance from $i$ to $j$ (in the second line), then upper bounded the value of the filter $G_{o_i(j)}$ based on the distance $||\pi(i)-\pi(j)||_\infty$, and finally changed order of summation to ensure that the outer summation is a weighted sum of absolute values of $x'_j$ {\em over all $j\in S\setminus S^*_H$}\footnote{We note here that we started by summing over $i$ first and then over $j$, but switched the order of summation to the opposite in the last line. This is because the quantity $G_{o_i(j)}$, which determines contribution of $j\in S$ to the estimation error of $i\in S$ {\em is not symmetric} in $i$ and $j$. Indeed, even though $G$ itself is symmetric around the origin, we have $o_i(j)=\pi(j)-(n/b)h(i)\neq o_j(i)$. }.  In order to upper bound $A_1$ it now suffices to upper bound all factors multiplying $x'_j$ in the last line of the equation above. As we now show, a strong bound follows from isolation properties of $i$. 

We start by upper bounding $G$ using Lemma~\ref{lm:filter-prop}, {\bf (2)}. We first note that by triangle inequality 
$$
||\pi(j)-(n/b)h(i)||_\infty\geq ||\pi(j)-\pi(i)||_\infty-||\pi(i)-(n/b)h(i)||_\infty\geq (n/b)(2^t-1)-(n/b)= (n/b) (2^{t-1}-2).
$$
The rhs is positive for all $t\geq 3$ and for such $t$ satisfies $2^{t-1}-2\leq 2^{t-2}$. We hence get for all  $t\geq 3$
\begin{equation}\label{eq:g-ubound}
\begin{split}
\max_{||\pi(j)-\pi(i)||_\infty\geq (n/b)\cdot (2^{t-1}-1)} G_{o_i(j)}\leq \left(\frac2{1+||\pi(j)-(n/b)h(i)||_\infty}\right)^{\fc}\leq \left(\frac2{1+2^{t-2}}\right)^{\fc}\leq 2^{-(t-3)\fc}.
\end{split}
\end{equation}
We also have the  bound $||G||_\infty\leq 1$ from Lemma~\ref{lm:filter-prop}, {\bf (3)}. It remains to bound the last term on the rhs of the last line in~\eqref{eq:2pogh}. We need the fact that for a pair $i, j$ such that $||\pi(j)-\pi(i)||_\infty\leq 2^{t+1}-1$ we have by triangle inequality
$$
||\pi(j)-(n/b)h(i)||_\infty\leq ||\pi(j)-\pi(i)||_\infty+||\pi(i)-(n/b)h(i)||_\infty\leq (n/b)(2^{t+1}-1)+(n/b)= (n/b) 2^{t+1}.
$$
Equipped with this bound, we now conclude that
\begin{equation}\label{eq:3hrgrg}
\begin{split}
&\left|\left\{i\in S\setminus S^*_H\setminus \{j\} \text{~s.t.~}||\pi(j)-\pi(i)||_\infty\leq (n/b)\cdot (2^{t+1}-1)\right\}\right|\\
&=|\pi(S\setminus \{i\})\cap \B^\infty_{(n/b) h(i)}((n/b)\cdot 2^{t+1})|\leq \GI\cdot \alpha^{d/2} 2^{(t+2)d+1}\cdot 2^t,
\end{split}
\end{equation}
where we used the assumption that $i\in S\setminus S^*_H$ are isolated  (see Definition~\ref{def:isolated}). We thus get for any $j\in S\setminus S^*_H$
\begin{equation*}
\begin{split}
\eta_j:=&\sum_{t\geq 0} \max_{\substack{||\pi(j)-\pi(i)||_\infty\geq \\(n/b)\cdot (2^t-1)}} G_{o_i(j)}\cdot \left|\left\{i\in S\setminus S^*_H\setminus \{j\} \text{~s.t.~}||\pi(j)-\pi(i)||_\infty\leq (n/b)\cdot (2^{t+1}-1)\right\}\right|\\
&\leq \sum_{t\geq 0} (\GI\cdot  \alpha^{d/2} 2^{(t+2)d+1}\cdot 2^{t}) \min\{1,  2^{-(t-3)\fc}\} \\
&\leq \GI\cdot  \alpha^{d/2} 2^{2d+1} \sum_{t\geq 0} 2^{t(d+1)}\cdot \min\{1, 2^{-(t-3)\fc}\}\\
\end{split}
\end{equation*}
We now note that
\begin{equation*}
\begin{split}
\sum_{t\geq 0} 2^{t(d+1)}\cdot \min\{1, 2^{-(t-3)\fc}\}&=1+2^{2(d+1)}+2^{3(d+1)}\sum_{t\geq 3} 2^{(t-3)(d+1)}\cdot \min\{1, 2^{-(t-3)\fc}\}\\
&=1+2^{2(d+1)}+2^{3(d+1)}\sum_{t\geq 3} 2^{(t-3)(d+1-\fc)}\leq 1+2^{2(d+1)}+2^{3(d+1)+1}\leq 2^{4(d+1)+1},\\
\end{split}
\end{equation*}
since $\fc\geq 2d$ by assumption of the lemma, and hence for all $j\in S\setminus S^*_H$ one has $\eta_j\leq  \GI\cdot 2^{O(d)} \alpha^{d/2}$. Combining the estimates above, we now get
\begin{equation*}
\begin{split}
A_1\leq &\sum_{j\in S\setminus S^*_H} |x'_j|\cdot \eta_j\leq  ||x'_S||_1\GI\cdot 2^{O(d)}\alpha^{d/2}, \\
\end{split}
\end{equation*}
 as required.  The $\ell_\infty$ bound for the case when $\chi_{\nsq\setminus S}=0$ follows in a similar manner and is hence omitted.

We now turn to bounding $A_2$. The bound that we get here is weaker since $\chi_{\nsq\setminus S}$ is an adversarially placed signal and we do not have isolation properties with respect to it, resulting in a weaker bound on (the equivalent of) $\eta_j$ for $j\in S^*_H$ than we had for $j\in S\setminus S^*_H$.  We let $y:=x'_{S^*}-\chi_{\nsq\setminus S}$ to simplify notation. We have, as in~\eqref{eq:2pogh},
\begin{equation*}
\begin{split}
A_2&\leq \sum_{j\in S\setminus S^*_H} |x'_j|\cdot \kappa_j,\\
& \text{~where~}\\
&\kappa_j=\sum_{t\geq 0} \max_{\substack{||\pi(j)-\pi(i)||_\infty\geq \\(n/b)\cdot (2^t-1)}} G_{o_i(j)}\cdot \left|\left\{i\in S\setminus S^*_H\setminus \{j\} \text{~s.t.~}||\pi(j)-\pi(i)||_\infty\leq (n/b)\cdot (2^{t+1}-1)\right\}\right|.
\end{split}
\end{equation*}
The first term can be upper bounded as before. For the second term, we note that every pair of points $i_1, i_2\in S\setminus S^*_H$ by triangle inequality satisfy 
$$
(n/b)||\pi(i_1)-\pi(i_2)||_\infty\leq (n/b)||\pi(i_1)-\pi(j)||_\infty+||\pi(j)-\pi(i_2)||_\infty\leq (n/b)\cdot (2^{t+2}-2)\leq (n/b)\cdot 2^{t+2}
$$ 
Since both $i_1$ and $i_2$ are isolated under $\pi$, this means that 
$$
\left|\left\{i\in S\setminus S^*_H\setminus \{j\} \text{~s.t.~}||\pi(j)-\pi(i)||_\infty\leq (n/b)\cdot (2^{t+1}-1)\right\}\right|\leq \GI\cdot \alpha^{d/2} 2^{(t+3)d}\cdot 2^{t+2}+1,
$$
where we used the bound from Definition~\ref{def:isolated} for $i$, but counted the point $i$ itself (this is what makes the bound on $\kappa_j$ weaker than the bound on $\eta_j$). A similar calculation to the one above for $A_1$ now gives
\begin{equation*}
\begin{split}
\kappa_j:=&\sum_{t\geq 0} \max_{\substack{||\pi(j)-\pi(i)||_\infty\geq \\(n/b)\cdot (2^t-1)}} G_{o_i(j)}\cdot \left|\left\{i\in S\setminus S^*_H\setminus \{j\} \text{~s.t.~}||\pi(j)-\pi(i)||_\infty\leq (n/b)\cdot (2^{t+1}-1)\right\}\right|\\
&\leq \sum_{t\geq 0} (\GI\cdot  \alpha^{d/2} 2^{(t+3)d}\cdot 2^{t+2}+1) \min\{1,  2^{-(t-3)\fc}\} \\
&\leq 2^{O(d)}(\GI\cdot \alpha^{d/2}+1)= 2^{O(d)}.
\end{split}
\end{equation*}
We thus have
\begin{equation*}
\begin{split}
A_2\leq &\sum_{j\in \nsq} |y_j|  \kappa_j\leq 2^{O(d)} ||y||_1.
\end{split}
\end{equation*}
Plugging our bounds on $A_1$ and $A_2$ into \eqref{eq:asdf}, we get
\begin{equation*}
\begin{split}
e^{head}_{i}(H, x, \chi)&\leq |G_{o_i(i)}^{-1}|\cdot (A_1+A_2)\leq |G_{o_i(i)}^{-1}|(2^{O(d)}\GI\cdot \alpha^{d/2} ||x'_S||_1+2^{O(d)}||y||_1)\\
&\leq 2^{O(d)}\alpha^{d/2} ||x'_S||_1+\GS \cdot 2^{O(d)} ||y||_1\\
\end{split}
\end{equation*}
as required.

\end{proof}
\begin{remark}
The second bound of this lemma will be useful later in section~\ref{sec:linf} for analyzing \textsc{ReduceInfNorm}.
\end{remark}

We now bound the final error induced by head elements, i.e. $e^{head}(\{H_r\}, x, \chi)$:
\begin{lemma}\label{lm:eh-s-sstar}
Let $x, \chi\in \nsq$, $x'=x-\chi$. \defsk Let $\{\pi_r\}_{r=1}^{r_{max}}$ be a set of permutations, let $H_r=(\pi_r, B, F), F\geq 2d$ be a hashing into $B$ buckets and filter $G$ with sharpness $F$. Let $S^*$ denote the set of elements $i\in S$ that are not isolated under at least $\sqrt{\alpha}$  fraction of $H_r$. Then, one has for $e^{head}$ defined with respect to $S$, 
$$
||e^{head}_{S\setminus S^*}(\{H_r\}, x, \chi)||_1\leq 2^{O(d)}\alpha^{d/2} ||x'_S||_1+\GS \cdot 2^{O(d)} ||\chi_{\nsq\setminus S}||_1.
$$
Furthermore, if $\chi_{\nsq\setminus S}=0$, then $||e^{head}_{S\setminus S^*}(\{H_r\}, x, \chi)||_\infty\leq 2^{d/2}\alpha^{d/2} ||x'_S||_\infty$.
\end{lemma}
\begin{proof}
Recall that by \eqref{eq:eh} one has for each $i\in \nsq$ $e^{head}_i(\{H_r\}, x, \chi)=\quant^{1/5}_{r\in [1:r_{max}]} e^{head}_i(H_r, x, \chi)$. This means that for each $i\in S\setminus S^*$ there exist at least $(1/5-\sqrt{\alpha})r_{max}$ values of $r$ such that 
$e^{head}_i(H_r, x, \chi)>e^{head}_i(\{H_r\}, x, \chi)$, and hence
$$
||e^{head}_{S\setminus S^*}(\{H_r\}, x, \chi)||_1\leq \frac1{(1/5-\sqrt{\alpha})r_{max}}\sum_{r=1}^{r_{max}} ||e^{head}_{S\setminus S^*_r}(H_r, x, \chi)||_1.
$$

By Lemma~\ref{lm:l1b-single-pi} one has
$$
||e^{head}_{S\setminus S^*_{H_r}}(H_r, x, \chi)||_1\leq 2^{O(d)}\alpha^{d/2} ||x'_S||_1+\GS \cdot 2^{O(d)} ||\chi_{\nsq\setminus S}||_1
$$
for all $r$, implying that 
\begin{equation*}
\begin{split}
||e^{head}_{S\setminus S^*}(\{H_r\}, x, \chi)||_1&\leq \frac{1}{(1/5-\sqrt{\alpha})}(2^{O(d)}\alpha^{d/2} ||x'_S||_1+\GS \cdot 2^{O(d)} ||\chi_{\nsq\setminus S}||_1)\\
&\leq 2^{O(d)}\alpha^{d/2} ||x'_S||_1+\GS \cdot 2^{O(d)} ||\chi_{\nsq\setminus S}||_1
\end{split}
\end{equation*}
as required.

The proof of the second bound follows analogously using the $\ell_\infty$ bound from Lemma~\ref{lm:l1b-single-pi}.
\end{proof}
\begin{remark}
The second bound of this lemma will be useful later in section~\ref{sec:linf} for analyzing \textsc{ReduceInfNorm}.
\end{remark}

\subsection{Bounding effect of tail noise}\label{sec:tail-noise}

\begin{lemma}\label{lm:loc-tail-small-single-H}
For any constant $C'>0$ there exists an absolute constant $C>0$ such that for any $x\in \C^{\nsq}$, any integer $k\geq 1$ and $S\subseteq \nsq$ such that $||x_{\nsq\setminus S}||_\infty\leq C'||x_{\nsq\setminus [k]}||_2/\sqrt{k}$, for any integer $B\geq 1$ a power of $2^d$ the following conditions hold.
If $(H, \A)$ are random measurements as in Algorithm~\ref{alg:main-sublinear}, $H=(\pi, B, F)$ satisfies $F\geq 2d$ and $||x_{\nsq\setminus [k]}||_2\geq N^{-\Omega(c)}$, where $O(c)$ is the word precision of our semi-equispaced Fourier transform computation, then for any $i\in \nsq$ one has, for $e^{tail}$ defined with respect to $S$, 
$$
\expect_{H, \A}\left[e^{tail}_{i}(H, \A, x)\right]\leq  \GS \cdot C^d(40+|\H|2^{-\Omega(|\A|)})  ||x_{\nsq \setminus [k]}||_2/\sqrt{B}.
$$ 
\end{lemma}
\begin{proof}

Recall that for any $H=(\pi, B, G), a, \h$ one has $(e^{tail}(H, a\star (\one, \h), x_{\nsq\setminus [k]}))^2=|u_i|^2$, where 
$$
u =\Call{HashToBins}{\widehat{x_{\nsq\setminus S}}, 0, (H, a\star (\one, \h))}.
$$

Since the elements of $\A$ are selected uniformly at random, we have for any $H$ and $\h$ by Lemma~\ref{lm:hashing}, {\bf (3)}, since $a\star (\one, \h)$ is uniformly random in $\nsq$, that
\begin{equation}\label{eq:expect-a-0hg443g}
\expect_a[(e^{tail}_i(H, a\star (\one, \h), x))^2]=\expect_{a}[\abs{G_{o_i(i)}^{-1}\omega^{-(a\star (\one, \h))^T\Sigma i}u_{h(i)} - x_i}^2]\leq \mu^2_{H, i}(x)+N^{-\Omega(c)},
\end{equation}
where $c>0$ is the large constant that governs the precision of our Fourier transform computations. 
By Lemma~\ref{lm:hashing}, {\bf (2)} applied to the pair $(\widehat{x_{\nsq\setminus S}}, 0)$ there exists a constant $C>0$ such that
$$
\expect_H[\mu^2_{H, i}] \leq  \GSS\cdot C^d ||x_{\nsq \setminus S}||_2^2/B
$$ 
We would like to upper bound the rhs in terms of $||x_{\nsq\setminus [k]}||_2^2$ (the tail energy), but this requires an argument since $S$ is not exactly the set of top $k$ elements of $x$. However, since $S$ contains the large coefficients of $x$, a bound is easy to obtain. Indeed, denoting the set of top $k$ coefficients of $x$ by $[k]\subseteq \nsq$ as usual, we get
\begin{equation*}
||x_{\nsq \setminus S}||_2^2\leq ||x_{\nsq \setminus (S\cup [k])}||_2^2+||x_{[k]\setminus S}||_2^2\leq ||x_{\nsq \setminus [k]}||_2^2+k\cdot ||x_{[k]\setminus S}||_\infty^2\leq (C'+1)||x_{\nsq \setminus [k]}||_2.
\end{equation*}
Thus, we have  
\begin{equation*}
\expect_H[\mu^2_{H, i}(x)+N^{-\Omega(c)}]\leq \GSS\cdot (C'+2)C^d ||x_{\nsq \setminus [k]}||_2^2/B,
\end{equation*}
where we used the assumption that $||x_{\nsq \setminus k}||_2\geq N^{-\Omega(c)}$.
We now get by Jensen's inequality 
\begin{equation}\label{eq:mu-b}
\expect_H[\mu_{H, i}(x)]\leq\GS\cdot (C'')^d ||x_{\nsq \setminus k}||_2/\sqrt{B}
\end{equation}
for a constant $C''>0$. Note that

By ~\eqref{eq:expect-a-0hg443g} for each $i\in \nsq$, hashing $H$, evaluation point $a\in \nsq\times \nsq$ and direction $\h$ we have
$\expect_a[(e^{tail}_i(H, a\star (\one,  \h), x))^2]=(\mu_{H, i}(x))^2$. Applying Jensen's inequality, we hence get for any $H$ and $\h\in \H$
\begin{equation}
\expect_a[e^{tail}_{i}(H, a\star (\one, \h), x)]\leq \mu_{H, i}(x).
\end{equation}
 Applying Lemma~\ref{lm:quant-exp} with $Y=e^{tail}_{i}(H, a\star (\one, \h), x)$ and $\gamma=1/5$ (recall that the definition of $e^{tail}_i(H, z, x)$ involves a $1/5$-quantile over $\A$) and using the previous bound, we get, for any fixed $H$ and $\h\in \H$
\begin{equation}\label{eq:e-b-exp}
\expect_\A\left[\left|e^{tail}_{i}(H, \A\star (\one, \h), x)-40\cdot \mu_{H, i}(x)\right|_+\right]\leq \mu_{H, i}(x)\cdot 2^{-\Omega(|\A|)},
\end{equation}
and hence by a union bound over all $\h\in \H$ we have 
\begin{equation*}
\expect_\A\left[\sum_{\h\in \H}\left|e^{tail}_{i}(H, \A\star (\one, \h), x)-40\cdot \mu_{H, i}(x)\right|_+\right]\leq \mu_{H, i}(x)\cdot |\H|2^{-\Omega(|\A|)}.
\end{equation*}
Putting this together with~\eqref{eq:mu-b}, we get 
\begin{equation*}
\begin{split}
&\expect_{H, \A}\left[e^{tail}_i(H, \A, x)\right]\\
&=\expect_H\left[\expect_\A\left[40\mu_{H, i}(x)+\sum_{\h\in \H}\left|e^{tail}_{i}(H, \A\star (\one, \h), x)-40\cdot \mu_{H, i}(x)\right|_+\right]\right]\\
&\leq\expect_H\left[\mu_{H, i}(x)(40+|\H|2^{-\Omega(|\A|)})\right]\\
&\leq \GS (C'')^d (40+|\H|2^{-\Omega(|\A|)}) ||x_{\nsq \setminus k}||_2/\sqrt{B}
\end{split}
\end{equation*}
as required.
\end{proof}

\begin{lemma}\label{lm:loc-tail-small}
For any constant $C'>0$ there exists an absolute constant $C>0$ such that for any $x\in \C^{\nsq}$, any integer $k\geq 1$ and $S\subseteq \nsq$ such that $||x_{\nsq\setminus S}||_\infty\leq C'||x_{\nsq\setminus [k]}||/\sqrt{k}$, if $B\geq 1$, then the following conditions hold , for $e^{tail}$ defined with respect to $S$.

If hashings $H_r=(\pi_r, B, F), F\geq 2d$ and sets $\A_r, |\A_r|\geq c_{max}$ for $r=1,\ldots, r_{max}$ are chosen at random, then
\begin{description}
\item[(1)] for every $i\in \nsq$ one has 
$$
\expect_{\{(H_r, \A_r)\}}\left[e^{tail}_i(\{H_r, \A_r\}, x)\right]\leq  \GS C^d (40+|\H|2^{-\Omega(c_{max})})  ||x_{\nsq \setminus [k]}||_2/\sqrt{B}.
$$ 

\item[(2)] for every $i\in \nsq$ one has 
$$
\prob_{\{(H_r, \A_r)\}}\left[e^{tail}_i(\{H_r, \A_r\}, x)>  \GS C^d (40+|\H|2^{-\Omega(c_{max})})  ||x_{\nsq \setminus [k]}||_2/\sqrt{B}\right]=2^{-\Omega(r_{max})}
$$ 
and
\begin{equation*}
\begin{split}
&\expect_{\{(H_r, \A_r)\}}\left[\left|e^{tail}_i(\{H_r, \A_r\}, x)- \GS C^d (40+|\H|2^{-\Omega(c_{max})})  ||x_{\nsq \setminus [k]}||_2/\sqrt{B}\right|_+\right]\\
&=2^{-\Omega(r_{max})}\cdot \GS C^d (40+|\H|2^{-\Omega(c_{max})})  ||x_{\nsq \setminus [k]}||_2/\sqrt{B}.
\end{split}
\end{equation*}

\end{description}
\end{lemma}
\begin{proof}

Follows by applying  Lemma~\ref{lm:quant-exp} with $Y=e^{tail}_{i}(H_r, \A_r, x)$.
\end{proof}

\subsection{Putting it together}
The bounds from the previous two sections yield a proof of Theorem~\ref{thm:l1-res-loc}, which we restate here for convenience of the reader:

{\em {\bf Theorem~\ref{thm:l1-res-loc}}
For any constant $C'>0$ there exist absolute constants $C_1, C_2, C_3>0$ such that for any $x\in \C^\nsq$, any integer $k\geq 1$ and any $S\subseteq \nsq$ such that $||x_{\nsq\setminus S}||_\infty\leq C'\mu$, where $\mu=||x_{\nsq\setminus [k]}||_2/\sqrt{k}$, the following conditions hold.

Let $\pi_r=(\Sigma_r, q_r), r=1,\ldots, r_{max}$ denote permutations, and let $H_r=(\pi_r, B, F)$, $F\geq 2d$, where $B\geq \GT k/\alpha^d$ for $\alpha\in (0, 1)$ smaller than a constant. Let $S^*\subseteq S$ denote the set of elements that are not isolated with respect to at least a $\sqrt{\alpha}$ fraction of hashings $\{H_r\}$. Then if $r_{max}, c_{max}\geq (C_1/\sqrt{\alpha})\log\log N$, then with probability at least $1-1/\log^2 N$ over the randomness of the measurements for all $\chi\in \C^\nsq$ such that $x':=x-\chi$ satisifies $||x'||_\infty/\mu\leq N^{O(1)}$ one has
$$
L:=\bigcup_{r=1}^{r_{max}}\textsc{LocateSignal}\left(\chi, k, \{m(\wh{x}, H_r, a\star (\one, \h))\}_{r=1, a\in \A_r, \h\in \H}^{r_{max}}\right)
$$
satisfies
$$
||x'_{S\setminus S^*\setminus L}||_1\leq (C_2\alpha)^{d/2} ||x'_S||_1+C_3^{d^2}(||\chi_{\nsq\setminus S}||_1+||x'_{S^*}||_1)+4\mu |S|.
$$

}
\begin{proof}
First note that with probability at least  $1-1/(10\log^2 N)$  for every $s\in [1:d]$ the sets $\A_r\star (\zero, \mathbf{e}_s)$ are balanced (as per Definition~\ref{def:balance}) for all $r=1,\ldots, r_{max}$ and all $\h\in \H$ by Claim~\ref{cl:balanced}.

By Corollary~\ref{cor:loc} applied with $S'=S\setminus S^*$ one has 
$$
||(x-\chi)_{(S\setminus S^*)\setminus L}||_1\leq 20 \cdot (||e^{head}_{S\setminus S^*}(\{H_r\}, x')||_1+||e^{tail}(\{H_r, \A_r\}, x)||_1)+||x'||_\infty |S|\cdot N^{-\Omega(c)}.
$$

We also have  
$$
||e^{head}_{S\setminus S^*}(\{H_r\}, x')||_1\leq 2^{O(d)}\alpha^{d/2} ||x'_S||_1+\GS \cdot 2^{O(d)} ||\chi_{\nsq\setminus S}||_1
$$
by Lemma~\ref{lm:eh-s-sstar} and with probability at least $1-1/(10\log^2 N)$
\begin{equation*}
||e^{tail}_{S\setminus S^*}(\{H_r, \A_r\}, x)||_1\leq \GS C^d (40+|\H|2^{-\Omega(c_{max})})  ||x_{\nsq \setminus [k]}||_2|S|/\sqrt{B} 
\end{equation*}
by Lemma~\ref{lm:loc-tail-small}. The rhs of the previous equation is bounded by $|S|\mu$ by the choice of $B$ as long as $\alpha$ is smaller than a absolute constant, as required.
Putting these bounds together and using the fact that $|\H|\leq \log N$ (so that $|\H|\cdot (2^{-\Omega(r_{max})}+2^{-\Omega(c_{max})})\leq 1$), and taking a union bound over the failure events, we get the result.
\end{proof}

\section{Analysis of \textsc{ReduceL1Norm} and \textsc{SparseFFT}}\label{sec:reduce-l1}
In this section we first give a correctness proof and runtime analysis for \textsc{ReduceL1Norm} (section~\ref{sec:rl1n}), then analyze the SNR reduction loop in \textsc{SparseFFT}(section~\ref{sec:snr-loop}) and finally prove correctness of \textsc{SparseFFT} and provide runtime bounds in section~\ref{sec:sfft}.

\subsection{Analysis of \textsc{ReduceL1Norm}}\label{sec:rl1n}

The main result of this section is Lemma~\ref{lm:reduce-l1-norm} (restated below). Intuitively, the lemma shows that \textsc{ReduceL1Norm} reduces the $\ell_1$ norm of the head elements of the input signal $x-\chi$ by a polylogarthmic factor, and does not introduce too many new spurious elements (false positives) in the process. The  introduced spurious elements, if any, do not contribute much $\ell_1$ mass to the head of the signal. Formally, we show

\noindent{\em {\bf Lemma~\ref{lm:reduce-l1-norm}}(Restated)
For any $x\in \C^N$, any integer $k\geq 1$, $B\geq \GT\cdot k/\alpha^d$ for $\alpha\in (0, 1]$ smaller than an absolute constant and $F\geq 2d, F=\Theta(d)$ the following conditions hold for the set $S:=\{i\in \nsq: |x_i|>\mu\}$, where $\mu^2\geq ||x_{\nsq\setminus [k]}||_2^2/k$. Suppose that $||x||_\infty/\mu=N^{O(1)}$.

For any sequence of hashings $H_r=(\pi_r, B, F)$, $r=1,\ldots, r_{max}$, if $S^*\subseteq S$ denotes the set of elements of $S$ that are not isolated with respect to at least a $\sqrt{\alpha}$ fraction of the hashings  $H_r, r=1,\ldots, r_{max}$, then for any $\chi\in \C^\nsq$, $x':=x-\chi$, if $\nu\geq (\log^4 N)\mu$ is a parameter such that
\begin{description}
\item[A] $||(x-\chi)_S||_1\leq  (\nu +20\mu)k$;
\item[B] $||\chi_{\nsq\setminus S}||_0\leq \frac{1}{\log^{19} N}k$;
\item[C] $||(x-\chi)_{S^*}||_1+||\chi_{\nsq\setminus S}||_1\leq \frac{\nu}{\log^4 N}k$,
\end{description}
the following conditions hold.

If parameters $r_{max}, c_{max}$ are chosen to be at least $(C_1/\sqrt{\alpha})\log\log N$, where $C_1$ is the constant from Theorem~\ref{thm:l1-res-loc} and measurements are taken as in Algorithm~\ref{alg:main-sublinear}, then the output $\chi'$ of the call 
$$
\Call{ReduceL1Norm}{\chi, k, \{m(\wh{x}, H_r, a\star (\one, \h))\}_{r=1, a\in \A_r, \h\in \H}^{r_{max}}, 4\mu (\log^4 n)^{T-t}, \mu}
$$ 
satisfies
\begin{enumerate}
\item $||(x'-\chi')_S||_1\leq \frac1{\log^{4} N} \nu k+20\mu k$ ~~~~~~~~($\ell_1$ norm of head elements is reduced by $\approx \log^4 N$ factor)
\item $||(\chi+\chi')_{\nsq\setminus S}||_0\leq ||\chi_{\nsq\setminus S}||_0+ \frac1{\log^{20} N}k$~~~~(few spurious coefficients are introduced)
\item $||(x'-\chi')_{S^*}||_1+||(\chi+\chi')_{\nsq\setminus S}||_1\leq ||x'_{S^*}||_1+||\chi_{\nsq\setminus S}||_1+\frac1{\log^{20} N}\nu k$ ~~~~($\ell_1$ norm of spurious coefficients does not grow fast)
\end{enumerate}
with probability at least $1-1/\log^{2} N$ over the randomness used to take measurements $m$ and by calls to \textsc{EstimateValues}.
The number of samples used is bounded by $2^{O(d^2)}k(\log\log N)^2$, and the runtime is bounded by $2^{O(d^2)} k\log^{d+2} N$.
}

Before giving the proof of Lemma~\ref{lm:reduce-l1-norm}, we prove two simple supporting lemmas.

\begin{lemma}[Few spurious elements are introduced in \textsc{ReduceL1Norm}]\label{lm:small-l0-increment}
For any $x\in \C^N$, any integer $k\geq 1$, $B\geq \GT\cdot k/\alpha^d$ for $\alpha\in (0, 1]$ smaller than an absolute constant and $F\geq 2d, F=\Theta(d)$ the following conditions hold for the set $S:=\{i\in \nsq: |x_i|>\mu\}$, where $\mu^2\geq ||x_{\nsq\setminus [k]}||_2^2/k$.

For any sequence of hashings $H_r=(\pi_r, B, F)$, $r=1,\ldots, r_{max}$, if $S^*\subseteq S$ denotes the set of elements of $S$ that are not isolated with respect to at least a $\sqrt{\alpha}$ fraction of the hashings  $H_r, r=1,\ldots, r_{max}$, then for any $\chi\in \C^\nsq$, $x':=x-\chi$ the following conditions hold.

Consider the call 
$$
\Call{ReduceL1Norm}{\chi, k, \{m(\wh{x}, H_r, a\star (\one, \h))\}_{r=1, a\in \A_r, \h\in \H}^{r_{max}}, 4\mu (\log^4 n)^{T-t}, \mu},
$$ 
where we assume that measurements of $x$ are taken as in Algorithm~\ref{alg:main-sublinear}. Denote, for each $t=0,\ldots, \log_2(\log^4 N)$, the signal recovered by step $t$ in this call by $\chi^{(t)}$ (see Algorithm~\ref{alg:l1-norm-reduction}). There exists an absolute constant $C>0$ such that if for a parameter $\nu\geq 2^t\mu$ at step $t$
\begin{description}
\item[A] $||(x'-\chi^{(t)})_S||_1\leq (2^{-t}\nu+20\mu) k$;
\item[B] $||(\chi+\chi^{(t)})_{\nsq\setminus S}||_0\leq \frac2{\log^{19} N} k$,
\item[C] $||(x'-\chi^{(t)})_{S^*}||_1+||(\chi+\chi^{(t)})_{\nsq\setminus S}||_1\leq \frac{2\nu}{\log^4 N} k$,
\end{description}
then  with probability at least $1-(\log N)^{-3}$ over the randomness used in \textsc{EstimateValues} at step $t$ one has
$$
||(\chi+\chi^{(t+1)})_{\nsq\setminus S}||_0-||(\chi+\chi^{(t)})_{\nsq\setminus S}||_0\leq \frac{1}{\log^{21} N}  k.
$$
\end{lemma}
\begin{proof}
Recall that $L'\subseteq L$ is the list output by \textsc{EstimateValues}. We let 
$$
L''=\left\{i\in L: |\chi'_i-x'_i|>\alpha^{1/2}\left(2^{-t}\nu+20\mu\right)\right\}
$$ denote the set of elements in $L'$ that failed to be estimated to within an additive $\alpha^{1/2}\left(2^{-t}\nu+20\mu\right)$ error term. 
For any element $i\in L$ we consider two cases, depending on whether $i\in L'\setminus L''$ or $i\in L''$.

\paragraph{Case 1:} First suppose that $i\in L'\setminus L''$, i.e. $|x'_i-\chi'_i|<\alpha^{1/2}(2^{-t}\nu+20\mu)$. Then if $\alpha$ is smaller than an absolute constant, we have
$$
|x'_i|>\frac1{1000}\nu 2^{-t}+4\mu-(\alpha^{1/2}(2^{-t}\nu+20\mu))\geq 2\mu,
$$
because only elements $i$ with $|\chi'_i|>\frac1{1000}\nu 2^{-t}+4\mu$ are included in the set $L'$ in the call 
$$
\chi' \gets \Call{EstimateValues}{x,  \chi^{(t)}, L, k, \e, C(\log\log N+d^2+O(\log (B/k))), \frac1{1000}\nu 2^{-t}+4\mu}
$$ 
due to the pruning threshold of $\frac1{1000}\nu 2^{-t}+4\mu$ passed to \textsc{EstimateValues} in the last argument. 

Since $||x_{\nsq\setminus S}||_\infty\leq \mu$ by definition of $S$, this means that either $i\in S$, or $i\in \supp \chi^{(t)}$. In both cases $i$ contributes at most $0$ to  
$||(\chi+\chi^{(t+1)})_{\nsq\setminus S}||_0-||(\chi+\chi^{(t)})_{\nsq\setminus S}||_0$.

\paragraph{Case 2:}Now suppose that  $i\in L''$, i.e. $|(x'-\chi')_i|\geq \alpha^{1/2}(2^{-t}\nu+20\mu)$. In this case $i$ may contribute $1$ to 
$||(\chi+\chi^{(t+1)})_{\nsq\setminus S}||_0-||(\chi+\chi^{(t)})_{\nsq\setminus S}||_0$. However, 
the number of elements in $L''$ is small. To show this, we invoke Lemma~\ref{lm:estimate-l1l2} to obtain precision guarantees for the call to \textsc{EstimateValues} on the pair $x, \chi$ and set of `head elements' $S\cup \supp \chi$. Note that $|S|\leq 2k$, as otherwise we would have $||x_{\nsq\setminus [k]}||_2^2> \mu\cdot k$, a contradiction. Further, by assumption {\bf B} of the lemma we have $||(\chi+\chi^{(t)})_{\nsq\setminus S}||_0\leq k$, so $|S\cup \supp (\chi+\chi^{(t)})|\leq 4k$. The $\ell_1$ norm of $x'-\chi^{(t)}$ on $S\cup \supp (\chi+\chi^{(t)})$ can be bounded as
\begin{equation*}
\begin{split}
\frac{||(x'-\chi^{(t)})_S||_1+||(x'-\chi^{(t)})_{\text{supp}(\chi+\chi^{(t)})\setminus S}||_1}{4k}&\leq \frac{||(x'-\chi^{(t)})_S||_1+||\chi^{(t)}_{\nsq\setminus S}||_1+||x'_{\nsq\setminus S}||_\infty\cdot |\supp (\chi+\chi^{(t)})|}{4k}\\
&\leq \frac{(2^{-t}\nu+20\mu )k+\frac{2\nu}{\log^4 N} k+\mu\cdot (4k)}{2k}\leq 2^{-t}\nu+20\mu,
\end{split}
\end{equation*}
For the $\ell_2$ bound on the tail of the signal we have 
$$
\frac{||(x'-\chi^{(t)})_{\nsq\setminus (S\cup \supp (\chi+\chi^{(t)}))}||_2^2}{4k}\leq \frac{||x_{\nsq\setminus S}||_2^2}{4k}\leq \mu^2.
$$

We thus have by Lemma~\ref{lm:estimate-l1l2}, {\bf (1)} for every $i\in L'$ that the estimate $w_i$ returned by \textsc{EstimateValues} satisfies
$$
\prob[|w_i-x'_i|>\alpha^{1/2}(2^{-t}\nu+20\mu)]<2^{-\Omega(r_{max})}.
$$

Since $r_{max}$ is chosen as $r_{max}=C(\log\log N+d^2+\log (B/k))$ for a sufficiently large absolute constant $C>0$, we have 
$$
\prob[|w_i-x'_i|>\alpha^{1/2}(2^{-t}\nu+20\mu)]<2^{-\Omega(r_{max})}\leq (k/B)\cdot (\log N)^{-25}.
$$
This means that 
$$
\expect[|L''|]\leq |L|\cdot (k/B)\cdot (\log N)^{-25}\leq (B\cdot r_{max})(k/B)\cdot (\log N)^{-25}\leq (\log N)^{-23},
$$
where the expectation is over the randomness used in \textsc{EstimateValues}. We used the fact that $|L'|\leq |L|\leq B\cdot r_{max}$ and that $r_{max}$ to derive the upper bound above. An application of Markov's inequality completes the proof.
\end{proof}

\begin{lemma}[Spurious elements do not introduce significant $\ell_1$ error]\label{lm:small-l1-increment}
For any $x\in \C^N$, any integer $k\geq 1$, $B\geq \GT\cdot k/\alpha^d$ for $\alpha\in (0, 1]$ smaller than an absolute constant and $F\geq 2d, F=\Theta(d)$ the following conditions hold for the set $S:=\{i\in \nsq: |x_i|>\mu\}$, where $\mu^2\geq ||x_{\nsq\setminus [k]}||_2^2/k$.

For any sequence of hashings $H_r=(\pi_r, B, F)$, $r=1,\ldots, r_{max}$, if $S^*\subseteq S$ denotes the set of elements of $S$ that are not isolated with respect to at least a $\sqrt{\alpha}$ fraction of the hashings  $H_r, r=1,\ldots, r_{max}$, then for any $\chi\in \C^\nsq$, $x':=x-\chi$ the following conditions hold.

Consider the call 
$$
\Call{ReduceL1Norm}{\chi, k, \{m(\wh{x}, H_r, a\star (\one, \h))\}_{r=1, a\in \A_r, \h\in \H}^{r_{max}}, 4\mu (\log^4 n)^{T-t}, \mu},
$$ 
where we assume that measurements of $x$ are taken as in Algorithm~\ref{alg:main-sublinear}. Denote, for each $t=0,\ldots, \log_2(\log^4 N)$, the signal recovered by step $t$ in this call by $\chi^{(t)}$ (see Algorithm~\ref{alg:l1-norm-reduction}). There exists an absolute constant $C>0$ such that if for a parameter $\nu\geq 2^t \mu$ at step $t$
\begin{description}
\item[A] $||(x'-\chi^{(t)})_S||_1\leq  (2^{-t}\nu+20\mu) k$;
\item[B] $||(\chi+\chi^{(t)})_{\nsq\setminus S}||_0\leq \frac2{\log^{19} N} k$;
\item[C] $||(x'-\chi^{(t)})_{S^*}||_1+||(\chi+\chi^{(t)})_{\nsq\setminus S}||_1\leq \frac{2\nu}{\log^4 N} k$,
\end{description}
then  with probability at least $1-(\log N)^{-3}$ over the randomness used in \textsc{EstimateValues} at step $t$ one has
$$
||(x'-\chi^{(t+1)})_{(\nsq\setminus S)\cup S^*}||_1-||(x'-\chi^{(t)})_{(\nsq\setminus S)\cup S^*}||_1\leq \frac{1}{\log^{21} N}k (\nu+\mu)
$$
\end{lemma}
\begin{proof}

We let $Q:=(\nsq\setminus S)\cup S^*$ to simplify notation, and recall that $L'\subseteq L$ is the list output by \textsc{EstimateValues}. We let 
$$
L''=\left\{i\in L: |\chi'_i-x'_i|>\alpha^{1/2}\left(2^{-t}\nu+20\mu\right)\right\}
$$ 
denote the set of elements in $L'$ that failed to be estimated to within an additive $\alpha^{1/2}\left(2^{-t}\nu+20\mu\right)$ error term.
We write
\begin{equation}\label{eq:dfskgj}
\begin{split}
||(x-\chi^{(t+1)})_Q||_1&=||(x-\chi^{(t+1)})_{Q\setminus L'}||_1+||(x-\chi^{(t+1)})_{(Q\cap L')\setminus L''}||_1+||(x-\chi^{(t+1)})_{Q\cap L''}||_1\\
\end{split}
\end{equation}
We first note that $\chi^{(t+1)}_i=\chi^{(t)}_i$ for all $i\not \in L'$, and hence $||(x'-\chi^{(t+1)})_{Q\setminus L}||_1=||(x'-\chi^{(t)})_{Q\setminus L}||_1$.

Second, for $i\in (Q\cap L')\setminus L''$ (second term) one has $|x'_i-\chi_i^{(t+1)}|\leq \sqrt{\alpha}(\nu 2^{-t}+4\mu)$.  Since only elements $i\in L$ with $|\chi'_i|>\frac1{1000}\nu 2^{-t}+4\mu$ are reported by the threshold setting in \textsc{EstimateValues}, so 
$|x'_i-\chi'|\leq \sqrt{\alpha} (2^{-t}\nu+20\mu)\leq x'_i$ as long as $\alpha$ is smaller than a constant.  We thus get that $||(x-\chi^{(t+1)})_{(Q\cap L')\setminus L''}||_1\leq ||(x-\chi^{(t)})_{(Q\cap L')\setminus L''}||_1$.

For the third term, we note that for each $i\in L$ the estimate $w_i$ computed in the call to \textsc{EstimateValues} satisfies
\begin{equation}\label{eq:9FHFJdhdjd}
\expect\left[\left||w_i-x'_i|-\sqrt{\alpha}(2^{-t}\nu+20\mu)\right|_+\right]\leq \sqrt{\alpha}(2^{-t}\nu+\mu)k2^{-\Omega(r_{max})} 
\end{equation}
by Lemma~\ref{lm:estimate-l1l2}, {\bf (2)}. Verification of the preconditions of the lemma is identical to Lemma~\ref{lm:small-l0-increment} (note that the assumptions of this lemma and Lemma~\ref{lm:small-l0-increment} are identical) and is hence omitted. Since $r_{max}=C(\log\log N+\log (B/k))$, the rhs of~\eqref{eq:9FHFJdhdjd} is bounded by $(\log N)^{-25} \sqrt{\alpha}(2^{-t}\nu+\mu)k$ as long as $C>0$ is larger than an absolute constant. 
We thus have
\begin{equation*}
\begin{split}
||(x'-\chi^{(t+1)})_{S\cap L''}||_1\leq \sum_{i\in S\cap L''} \left(\sqrt{\alpha}(2^{-t}\nu+20\mu)+\left||w_i-x'_i|-\sqrt{\alpha}(2^{-t}\nu+20\mu)\right|_+\right).
\end{split}
\end{equation*}
Combining ~\eqref{eq:9FHFJdhdjd} with the fact that by by Lemma~\ref{lm:estimate-l1l2}, {\bf (1)}, we have for every $i\in L$ 
$$
\prob\left[|w_i-x'_i|>\sqrt{\alpha}(2^{-t}\nu+20\mu)\right]\leq 2^{-\Omega(r_{max})}\leq (k/B)\cdot (\log N)^{-25}
$$
by our choice of $r_{max}$, we get that
$$
||(x'-\chi^{(t+1)})_{S\cap L''}||_1\leq  2\sqrt{\alpha}(2^{-t}\nu+20\mu)\cdot|L|\cdot (k/B)\cdot (\log N)^{-25}.
$$
An application of Markov's inequality then implies, if $\alpha$ is smaller than an absolute constant, that 
$$
\prob[||(x'-\chi^{(t+1)})_{S\cap L''}||_1> \frac1{\log^{21} N}(\nu+\mu)k]<1/\log^3 N.
$$

Substituting the bounds we just derived into \eqref{eq:dfskgj}, we get
$$
||(x-\chi^{(t+1)})_Q||_1\leq ||(x-\chi^{(t)})_Q||_1+\frac1{\log^{21} N}(\nu+\mu)k
$$
as required.

\end{proof}

Equipped with the two lemmas above, we can now give a proof of Lemma~\ref{lm:reduce-l1-norm}:

\begin{proofof}{Lemma~\ref{lm:reduce-l1-norm}}
We prove the result by strong induction on $t=0,\ldots, \log_2 (\log^4 N)$. Specifically, we prove that there exist events $\E_t, t=0,\ldots, \log_2 (\log^4 N)$ such that {\bf (a)} $\E_t$ depends on the randomness used in the call to \textsc{EstimateValues} at step $t$, $\E_t$ satisfies $\prob[\E_t|\E_0\wedge \ldots \E_{t-1}]\geq 1-3/\log^2 N$ and {\bf (b)} for all $t$ conditional on $\E_0\wedge \E_1\wedge\ldots\wedge\E_t$ one has
\begin{description}
\item[(1)] $||(x'-\chi^{(t)})_{S\setminus S^*}||_1\leq (2^{-t}\nu+20\mu) k$;
\item[(2)] $||(\chi+\chi^{(t)})_{\nsq\setminus S}||_0\leq ||\chi_{\nsq\setminus S}||_0+\frac{t}{\log^{21} N}k$;
\item[(3)] $||(x'-\chi^{(t)})_{S^*}||_1+||(\chi+\chi^{(t)})_{\nsq\setminus S}||_1\leq ||x'_{S^*}||_1+||\chi_{\nsq\setminus S}||_1+\frac{t}{\log^{21} N}\nu k$
\end{description}

The {\bf base} is provided by $t=0$ and is trivial since $\chi^{(0)}=0$. We now give the {\bf inductive step}.

We start by proving the inductive step for {\bf (2)} and {\bf (3)}.  We will use Lemma~\ref{lm:small-l0-increment} and Lemma~\ref{lm:small-l1-increment}, and hence we start by verifying that their preconditions (which are identical for the two lemmas) are satisfied. Precondition {\bf A} is satisfied directly by inductive hypothesis {\bf (1)}. Precondition {\bf B} is satisfied since 
$$
||(\chi+\chi^{(t)})_{\nsq\setminus S}||_0\leq ||\chi_{\nsq\setminus S}||_0+\frac{t}{\log^{21} N}k\leq \frac1{\log^{19} N} k+\frac{\log2(\log^4 N)}{\log^{21} N}\leq \frac2{\log^{19} N} k,
$$
where we used assumption {\bf B} of this lemma and inductive hypothesis {\bf (2)}. Precondition {\bf C} is satisfied since 
$$
||(x'-\chi^{(t)})_{S^*}||_1+||(\chi+\chi^{(t)})_{\nsq\setminus S}||_1\leq  ||x'_{S^*}||_1+||\chi_{\nsq\setminus S}||_1+\frac{t}{\log^{21} N}\nu k \leq \frac{\nu}{\log^4 N} k
+\frac{t}{\log^{21} N}\nu k\leq \frac{2\nu}{\log^4 N} k,
$$
where we used assumption {\bf 3} of this lemma, inductive assumption {\bf (3)} and the fact that $t\leq \log_2 (\log^4 N)\leq \log N$ for sufficiently large $N$.

\paragraph{Proving {\bf (2)}.}

To prove the inductive step for {\bf (2)}, we use Lemma~\ref{lm:small-l0-increment}. Lemma~\ref{lm:small-l0-increment} shows that with probability at least $1-(\log N)^{-2}$ over the randomness used in \textsc{EstimateValues} (denote the success event by $\E^1_t$) we have
$$
||(\chi+\chi^{(t+1)})_{\nsq\setminus S}||_0-||(\chi+\chi^{(t)})_{\nsq\setminus S}||_0\leq \frac{1}{\log^{21} N}k,
$$
so $||(\chi+\chi^{(t+1)})_{\nsq\setminus S}||_0\leq ||(\chi+\chi^{(t)})_{\nsq\setminus S}||_0+\frac{1}{\log^{21} N}k\leq ||\chi_{\nsq\setminus S}||_0+\frac{t+1}{\log^{21} N}k$ as required.

\paragraph{Proving {\bf (3)}.}
At the same time we have by Lemma~\ref{lm:small-l1-increment} that with probability at least $1-(\log N)^{-2}$ (denote the success event by $\E^2_t$)
$$
||(x'-\chi^{(t+1)})_{(\nsq\setminus S)\cup S^*}||_1-||(x'-\chi^{(t)})_{(\nsq\setminus S)\cup S^*}||_1\leq \frac{1}{\log^{21} N}k \nu,
$$
so by combing this with assumption {\bf (3)} of the lemma we get
$$
||(x'-\chi^{(t+1)})_{(\nsq\setminus S)\cup S^*}||_1\leq \frac{1}{\log^{20} N}\nu k+\frac{t+1}{\log^{21} N}\nu k
$$
as required.

\paragraph{Proving {\bf (1)}.}
We let $L''\subseteq L$ denote the set of elements in $L$ that fail to be estimated to within a small additive error. Specifically, we let
$$
L''=\left\{i\in L: |\chi'_i-x'_i|>\alpha^{1/2}\left(2^{-t}\nu+20\mu\right)\right\},
$$ 
where $\chi'$ is the output of \textsc{EstimateValues} in iteration $t$. We bound $||(x'-\chi^{(t+1)})_{S\setminus S^*}||_1$ by splitting this $\ell_1$ norm into three terms, depending on whether the corresponding elements were updated in iteration $t$ and whether they were well estimated. We have
\begin{equation}\label{eq:org-bound-pihwrve33v}
\begin{split}
&||(x'-\chi^{(t+1)})_{S\setminus S^*}||_1=||(x'-(\chi^{(t)}+\chi'))_{S\setminus S^*}||_1\\
&\leq ||(x'-(\chi^{(t)}+\chi'))_{S\setminus (S^*\cup L)}||_1+||(x'-(\chi^{(t)}+\chi'))_{(S\cap L)\setminus L'\setminus L''}||_1+||(x'-(\chi^{(t)}+\chi'))_{(S\cap L')\setminus L''}||_1\\
&+||(x'-(\chi^{(t)}+\chi'))_{L''}||_1\\
&=||(x'-\chi^{(t)})_{S\setminus (S^*\cup L)}||_1+||(x'-(\chi^{(t)}+\chi'))_{(S\cap L)\setminus L'\setminus L''}||_1+||(x'-(\chi^{(t)}+\chi'))_{(S\cap L')\setminus L''}||_1\\
&+||(x'-(\chi^{(t)}+\chi'))_{(L\cap S)\cap L''}||_1\\
&=:S_1+S_2+S_3+S_4,
\end{split}
\end{equation}
where we used the fact that $\chi'_{S\setminus L}\equiv 0$ to go from the second line to the third. We now bound the four terms.

The {\bf second term} (i.e. $S_2$) captures elements of $S$ that were estimated precisely (and hence they are not in $L''$), but were not included into $L'$ as they did not pass the threshold test (being estimated as larger than $\frac1{1000}2^{-t}\nu+4\mu$)  in \textsc{EstimateValues}. One thus has
\begin{equation}\label{eq:l2p}
\begin{split}
||(x-(\chi^{(t)}+\chi'))_{(S\cap L)\setminus L'\setminus L''}||_1&\leq \alpha^{1/2}(2^{-t}\nu+20\mu)\cdot |(S\cap L')\setminus L''|+(\frac1{1000} 2^{-t}\nu+4\mu)\cdot |(S\cap L')\setminus L''|\\
& \leq ((\frac{1}{1000}+\alpha^{1/2})2^{-t}\nu+(4+20\alpha^{1/2})\mu)2k
\end{split}
\end{equation}
since $|S|\leq 2k$ by assumption of the lemma.

The {\bf third term} (i.e. $S_3$) captures elements of $S$ that were reported by \textsc{EstimateValues} (hence do not belong to $L'$) and were approximated well (hence belong to $L''$). One has, by definition of the set $L''$,
\begin{equation}\label{eq:23grgerg}
\begin{split}
||(x-(\chi^{(t)}+\chi'))_{(S\cap L')\setminus L''}||_1&=\alpha^{1/2}(2^{-t}\nu+20\mu)\cdot |(S\cap L')\setminus L''|\\
& \leq 2\alpha^{1/2}(2^{-t}\nu+20\mu)k
\end{split}
\end{equation}
since $|S|\leq 2k$ by assumption of the lemma.

For the {\bf forth term} (i.e. $S_4$) we have 
$$
||(x'-(\chi^{(t)}+\chi'))_{L''}||_1\leq \alpha^{1/2}\left(2^{-t}\nu+20\mu\right)\cdot |L''|+\sum_{i\in S} \left||\chi'_i-x'_i|-\alpha^{1/2}\left(2^{-t}\nu+20\mu\right)\right|_+.
$$
By  Lemma~\ref{lm:estimate-l1l2}, {\bf (1)} (invoked on the set $S\cup \supp(\chi+\chi^{(t)}+\chi')$) we have $\expect[|L''|]\leq B\cdot 2^{-\Omega(r_{max})}$ and by  Lemma~\ref{lm:estimate-l1l2}, {\bf (2)} 
for any $i$ one has 
$$
\expect\left[\left||\chi'_i-x'_i|-\alpha^{1/2}\left(2^{-t}\nu+20\mu\right)\right|_+\right]\leq |L|\cdot \alpha^{1/2}\left(2^{-t}\nu+20\mu\right) 2^{-\Omega(r_{max})}.
$$
Since the parameter $r_{max}$ in \textsc{EstimateValues} is chosen to be at least $C(\log\log N+d^2+\log (B/k))$ for a sufficiently large constant $C$, and $|L|=O(\log N)B$, we have 
\begin{equation*}
\expect\left[||(x'-(\chi^{(t)}+\chi'))_{L''}||_1\right]\leq \alpha^{1/2}\left(2^{-t}\nu+20\mu\right) |L| 2^{-\Omega(r_{max})}\leq \frac1{\log^{25} N}\left(2^{-t}\nu+20\mu\right)k
\end{equation*}
By Markov's inequality we thus have 
\begin{equation}\label{eq:l2p-11}
||(x'-(\chi^{(t)}+\chi'))_{L''}||_1\leq \alpha^{1/2}\left(2^{-t}\nu+20\mu\right) |L| 2^{-\Omega(r_{max})}\leq \frac1{\log^{22} N}\left(2^{-t}\nu+20\mu\right)k
\end{equation}
 with probability at least $1-1/\log^3 N$. Denote the success event by $\E_t^0$.

Finally, in order to bound the {\bf first term} (i.e. $S_1$), we invoke Theorem~\ref{thm:l1-res-loc} to analyze the call to \textsc{LocateSignal} in the $t$-th iteration. 
We note that since $r_{max}, c_{max}\geq (C_1/\sqrt{\alpha})\log\log N$ (where $C_1$ is the constant from Theorem~\ref{thm:l1-res-loc}) by assumption of the lemma, the preconditions of Theorem~\ref{thm:l1-res-loc} are satisfied. By Theorem~\ref{thm:l1-res-loc} together with {\bf (1)} and {\bf (3)} of the inductive hypothesis we have 
\begin{equation}\label{eq:l111}
\begin{split}
||(x'-\chi^{(t)})_{S\setminus (S^*\cup L)}||_1&\leq (4C_2\alpha)^{d/2} ||(x'-\chi^{(t)})_{S\setminus S^*}||_1+(4C)^{d^2}( ||(\chi+\chi^{(t)})_{\nsq\setminus S}||_1+||(x'-\chi^{(t)})_{S^*}||_1)+4\mu |S|\\
&\leq O((4C_2\alpha)^{d/2}) (2^{-t}\nu+20\mu)k+(4C)^{d^2}(\frac{2}{\log^{20} N} \nu k)+8\mu k\\
&\leq \frac1{1000} (2^{-t}\nu+20\mu)k+8\mu k\\
\end{split}
\end{equation}
if $\alpha$ is smaller than an absolute constant and $N$ is sufficiently large.

Now substituting bounds on $S_1, S_2, S_3, S_4$ provided by ~\eqref{eq:l111},~\eqref{eq:l2p},~\eqref{eq:23grgerg} and~\eqref{eq:l2p-11} into~\eqref{eq:org-bound-pihwrve33v} we get
\begin{equation*}
\begin{split}
||(x'-\chi^{(t+1)})_{S\setminus S^*}||_1&\leq (\frac2{1000}+O(\alpha^{1/2})) 2^{-t}\nu+(16+O(\alpha^{1/2}))\mu k\\
&\leq 2^{-t}\nu+20\mu k\\
\end{split}
\end{equation*}
when $\alpha$ is a sufficiently small constant, as required. This proves the inductive step for {\bf (1)} and completes the proof of the induction.

Let $\E_t=\E^0_t\wedge \E^1_t \wedge \E^2_t$ denote the success event for step $t$. We have by a union bound $\prob[\E_t]\geq 1-3t/\log^2 N$ as required.

\paragraph{Sample complexity and runtime} It remains to bound the sampling complexity and runtime. First note that \textsc{ReduceL1Norm} only takes fresh samples in the calls to \textsc{EstimateValues} that it issues.
By Lemma~\ref{lm:estimate-l1l2} each such call uses $2^{O(d^2)} k(\log\log N)$ samples, amounting to $2^{O(d^2)} k(\log\log N)^2$ samples over $O(\log\log N)$ iterations. 

By Lemma~\ref{lm:loc} each call to \textsc{LocateSignal} takes $O(B(\log N)^{3/2})$ time.  Updating the measurements $m(\wh{x}, H_r, a\star (\one, \h)), \h\in \H$ takes 
$$
|\H|c_{max} r_{max}\cdot \fc^{O(d)}\cdot B \log^{d+1} N \log\log N=2^{O(d^2)}\cdot k\log^{d+2} N
$$ time overall.
The runtime complexity of the calls to \textsc{EstimateValues} is $2^{O(d^2)}\cdot k \log^{d+1} N(\log\log N)^2$ time overall.
Thus, the runtime is bounded by $2^{O(d^2)}k \log^{d+2} N$.
\end{proofof}

\subsection{Analysis of SNR reduction loop in \textsc{SparseFFT}}\label{sec:snr-loop}
In this section we prove

\noindent{\em {\bf Theorem~\ref{thm:l1snr}}
For any $x\in \C^N$, any integer $k\geq 1$, if $\mu^2\geq \err_k^2(x)/k$ and $R^*\geq ||x||_\infty/\mu=N^{O(1)}$, the following conditions hold for the set $S:=\{i\in \nsq: |x_i|>\mu\}\subseteq \nsq$.

Then the SNR reduction loop of Algorithm~\ref{alg:main-sublinear} (lines~19-25) returns $\chi^{(T)}$ such that 
\begin{equation*}
\begin{split}
&||(x-\chi^{(T)})_S||_1\lesssim \mu\text{~~~~~~~~~~~~~~~~~($\ell_1$-SNR on head elements is constant)}\\
&||\chi^{(T)}_{\nsq\setminus S}||_1\lesssim \mu \text{~~~~~~~~~~~~~~~~~~~~~~~~~~(spurious elements contribute little in $\ell_1$ norm)}\\
&||\chi^{(T)}_{\nsq\setminus S}||_0\lesssim \frac1{\log^{19} N} k\text{~~~~~~~~~~~~~(small number of spurious elements have been introduced)}
\end{split}
\end{equation*}

with probability at least $1-1/\log N$ over the internal randomness used by Algorithm~\ref{alg:main-sublinear}. The sample complexity is $2^{O(d^2)}k\log N(\log\log N)$. The runtime is bounded by $2^{O(d^2)} k \log^{d+3} N$.
}

\begin{proof}
We start with correctness. We prove by induction that after the $t$-th iteration one has 
\begin{description}
\item[(1)] $||(x-\chi^{(t)})_S||_1\leq 4(\log^4 N)^{T-t} \mu k+20\mu k$;
\item[(2)] $||x-\chi^{(t)}||_\infty=O((\log^4 N)^{T-(t-1)} \mu)$;
\item[(3)] $||\chi^{(t)}_{\nsq\setminus S}||_0\leq \frac{t}{\log^{20} N} k$.
\end{description}
The base is provided by $t=0$, where all claims are trivially true by definition of $R^*$.
We now prove the inductive step. The main tool here is Lemma~\ref{lm:reduce-l1-norm}, so we start by verifying that its preconditions are satisfied. First note that 

First, since $|S^*|\leq 2^{-\Omega(r_{max})}|S|\leq 2^{-\Omega(r_{max})}k\leq \frac1{\log^{19} N}k$ with probability at least $1-2^{-\Omega(r_{max})}\geq 1-1/\log N$ by Lemma~\ref{lm:good-prob} and choice of $r_{max}\geq (C/\sqrt{\alpha})\log\log N$ for a sufficiently large constant $C>0$. Also, by Claim~\ref{cl:balanced} we have that  with probability at least $1-1/\log^{2} N$ for every $s\in [1:d]$ the sets $\A_r\star (\zero, \mathbf{e}_s)$ are balanced (as per Definition~\ref{def:balance} with $\Delta=2^{\lfloor \frac1{2}\log_2 \log_2 n\rfloor}$, as needed for Algorithm~\ref{alg:location}). Also note that by {\bf (2)} of the inductive hypothesis one has $||x-\chi^{(t)}||_\infty/\mu=R^* \cdot O(\log N)=N^{O(1)}$.

First, assuming the inductive hypothesis {\bf (1)-(3)}, we verify that the preconditions of Lemma~\ref{lm:reduce-l1-norm} are satisfied with $\nu=4(\log^4 N)^{T-t} \mu k$.  First, for {\bf (A)} one has  $||(x-\chi^{(t)})_S||_1\leq 4(\log^4 N)^{T-t} \mu k$. This satisfies precondition {\bf A} of Lemma~\ref{lm:reduce-l1-norm}.  We have
\begin{equation}\label{eq:123rergregregvvv}
\begin{split}
||(x-\chi^{(t)})_{S^*}||_1+||\chi^{(t)}_{\nsq\setminus S}||_1&\leq ||x-\chi^{(t)}||_\infty\cdot (||(x-\chi^{(t)})_{S^*}||_0+||\chi^{(t)}_{\nsq\setminus S}||_0)\\
&\leq O(\log^4 N)\cdot \nu \cdot \left(\frac1{\log^{19} N}k+\frac{t}{\log^{20} N}k \right)\leq \frac{16}{\log^{14} N}\nu k
\end{split}
\end{equation}
for sufficiently large $N$. Since the rhs is less than $\frac1{\log^4 N}\nu k$, precondition {\bf (C)} of Lemma~\ref{lm:reduce-l1-norm} is also satisfied. Precondition {\bf (B)} of Lemma~\ref{lm:reduce-l1-norm} is satisfied by inductive hypothesis, {\bf (3)} together with the fact that $T=o(\log R^*)=o(\log N)$.

Thus, all preconditions of Lemma~\ref{lm:reduce-l1-norm} are satisfied. Then by Lemma~\ref{lm:reduce-l1-norm} with $\nu=4(\log^4 N)^{T-t} \mu$
one has with probability at least $1-1/\log^2 N$

\begin{enumerate}
\item $||(x'-\chi^{(t)}-\chi')_S||_1\leq \frac1{\log^4 N}\nu k+20\mu k$;
\item $||(\chi^{(t)}+\chi')_{\nsq\setminus S}||_0-||\chi^{(t)}_{\nsq\setminus S}||_0\leq \frac1{\log^{20} N}k$;
\item $||(x'-(\chi^{(t)}+\chi'))_{S^*}||_1+||(\chi^{(t)}+\chi')_{\nsq\setminus S}||_1\leq ||(x'-\chi^{(t)})_{S^*}||_1+||\chi^{(t)}_{\nsq\setminus S}||_1+\frac1{\log^{20} N}\nu k$.
\end{enumerate}
Combining 1 above with~\eqref{eq:123rergregregvvv} proves {\bf (1)} of the inductive step:
\begin{equation*}
\begin{split}
||(x-\chi^{(t+1)})_S||_1=||(x-\chi^{(t)}-\chi')_S||_1&\leq \frac1{\log^4 N}\nu k+20\mu k=\frac1{\log^4 N}4(\log^4 N)^{T-t} \mu k+20\mu k\\
&=4(\log^4 N)^{T-(t+1)} \mu k+20\mu k.
\end{split}
\end{equation*}
 Also, combining 2 above with the fact that $||\chi^{(t)}_{\nsq\setminus S}||_0\leq \frac{t}{\log^{20} N} k$ yields $||\chi^{(t+1)}_{\nsq\setminus S}||_0\leq \frac{t+1}{\log^{20} N} k$ as required.

In order to prove the inductive step is remains to analyze the call to \textsc{ReduceInfNorm}, for which we use Lemma~\ref{lm:linf} with parameter $\tilde k=4k/\log^4 N$. We first verify that preconditions of the lemma are satisfied. Let $y:=x-(\chi+\chi^{(t)}+\chi')$ to simplify notation. For that we need to verify that 
\begin{equation}\label{eq:head-c-ojg23g2g}
||y_{[\tilde k]}||_1/\tilde k\leq 4(\log^4 N)^{T-(t+1)}\mu=(\log^4 N)\cdot (\frac1{\log^4 N}\nu+20 \mu)
\end{equation}
and 
\begin{equation}\label{eq:pih2gg3rg}
||y_{\nsq\setminus [\tilde k]}||_2/\sqrt{\tilde k}\leq (\log^4 N)\cdot (\frac1{\log^4 N}\nu +20\mu ),
\end{equation}
where we denote $\tilde k:=4k/\log^4 N$ for convenience. The first condition is easy to verify, as we now show. Indeed, we have
\begin{equation*}
\begin{split}
&||y_{\tilde k}||_1\leq ||y_S||_1+||y_{\text{supp}(\chi^{(t)}+\chi')\setminus S}||_1+||x_{\nsq\setminus S}||_\infty \cdot \tilde k\\
&\leq ||y_S||_1 + ||(\chi^{(t)}+\chi')_{\nsq\setminus S}||_1+||x_{\text{supp}(\chi^{(t)}+\chi')\setminus S}||_\infty\cdot \tilde k+||x_{\nsq\setminus S}||_\infty\cdot \tilde k\\
&\leq \frac1{\log^4 N}\nu k+20\mu k + \frac1{\log^4 N}\nu k+2\mu \tilde k\leq \frac2{\log^4 N}\nu k+40 \mu k,
\end{split}
\end{equation*}
where we used the triangle inequality to upper bound $||y_{\text{supp}(\chi^{(t)}+\chi')\setminus S}||_1$ by $||(\chi^{(t)}+\chi')_{\nsq\setminus S}||_1+||x_{\text{supp}(\chi^{(t)}+\chi')\setminus S}||_\infty\cdot \tilde k$ to go from the first line to the second. We thus have 
$$
||y_{[\tilde k]}||_1/\tilde k\leq (\frac2{\log^4 N}\nu k+40 \mu k)/(4k/\log^4 N)\leq (\log^4 N)\cdot (\frac1{\log^4 N}\nu+20 \mu)
$$ 
as required. This establishes~\eqref{eq:head-c-ojg23g2g}.

To verify the second condition, we first let $\tilde S:=S\cup \supp(\chi+\chi^{(t)}+\chi')$ to simplify notation. We have
\begin{equation}\label{eq:11111}
\begin{split}
||y_{\nsq\setminus [\tilde k]}||_2^2&=||y_{\tilde S\setminus [\tilde k]}||_2^2+||y_{(\nsq\setminus \tilde S)\setminus [\tilde k]}||_2^2\leq ||y_{\tilde S\setminus [\tilde k]}||_2^2+\mu^2 k,
\end{split}
\end{equation}
where we used the fact that $y_{\nsq\setminus \tilde S}=x_{\nsq\setminus \tilde S}$ and hence $||y_{(\nsq\setminus \tilde S)\setminus [\tilde k]}||_2^2\leq \mu^2 k$.
We now note that $||y_{\tilde S\setminus [\tilde k]}||_1\leq ||y_{\tilde S}||_1\leq 2(\frac1{\log^4 N}\nu k+20\mu k)$, and so it must be that $||y_{S\setminus [\tilde k]}||_\infty\leq 2(\frac1{\log^4 N}\nu k+20\mu k) (k/\tilde k)$, as otherwise the top $\tilde k$ elements of $y_{[\tilde k]}$ would contribute more than $2(\frac1{\log^4 N}\nu k+20\mu k)$ to $||y_{\tilde S}||_1$, a contradiction.  With these constraints $||y_{\tilde S\setminus [\tilde k]}||_2^2$ is maximized when there are $\tilde k$ elements in $y_{\tilde S\setminus [\tilde k]}$, all equal to the maximum possible value, i.e. $||y_{\tilde S\setminus [\tilde k]}||_2^2\leq 4(\frac1{\log^4 N}\nu k+20\mu k)^2 (k/\tilde k)^2 \tilde k$. Plugging this into~\eqref{eq:11111}, we get $||y_{\nsq\setminus [\tilde k]}||_2^2\leq ||y_{\tilde S\setminus [\tilde k]}||_2^2+\mu^2k\leq 4(\frac1{\log^4 N}\nu k+20\mu k)^2 (k/\tilde k)^2 \tilde k+\mu^2 k$. This implies that 
\begin{equation*}
\begin{split}
||y_{\nsq\setminus [\tilde k]}||_2/\sqrt{\tilde k}&\leq \sqrt{4(\frac1{\log^4 N}\nu k+20\mu k)^2 (k/\tilde k)^2+\mu^2 (k/\tilde k)}\leq 2 (k/\tilde k)\sqrt{(\frac1{\log^4 N}\nu k+20\mu k)^2  +\mu^2}\\
&\leq 2((\frac1{\log^4 N}\nu k+20\mu k)+\mu) (k/\tilde k)\leq (\log^4 N) (\frac1{\log^4 N}\nu k+20\mu k),
\end{split}
\end{equation*}
establishing ~\eqref{eq:pih2gg3rg}. 

Finally,  also recall that $||y_{S\setminus [\tilde k]}||_\infty\leq 2(\frac1{\log^4 N}\nu k+20\mu k) (k/\tilde k)\leq (\log^4 N)\cdot (\frac1{\log^4 N}\nu k+20\mu k)$ and $||y_{\nsq\setminus \tilde S}||_\infty =||x_{\nsq\setminus S}||_\infty\leq \mu$. 

We thus have that all preconditions of Lemma~\ref{lm:linf} are satisfied for the set of top $\tilde k$ elements of $y$, and hence its output satisfies 
$$
||x-(\chi^{(t)}-\chi'-\chi'')||_\infty=O(\log^4 N)\cdot(\frac1{\log^4 N}\nu k+20\mu k).
$$
Putting these bounds together establishes {\bf (2)}, and completes the inductive step and the proof of correctness.

Finally, taking a union bound over all failure events (each call to \textsc{EstimateValues} succeeds with probability at least $1-\frac1{\log^2 N}$, and  with probability at least $1-1/\log^{2} N$ for all $s\in [1:d]$ the set $\A_r\star (\zero, \mathbf{e}_s)$ is balanced in coordinate $s$) and using the fact that $\log T=o(\log N)$ and each call to \textsc{LocateSignal} is deterministic, we get that success probability of the SNR reduction look is lower bounded by $1-1/\log N$.
 
\paragraph{Sample complexity and runtime} 
The sample complexity is bounded by the the sample complexity of the calls to \textsc{ReduceL1Norm} and \textsc{ReduceInfNorm} inside the loop times $O(\log N/\log\log N)$ for the number of iterations. The former is 
bounded by $2^{O(d^2)}k (\log\log N)^2$ by Lemma~\ref{lm:reduce-l1-norm}, and the latter is bounded by $2^{O(d^2)}k/\log N$ by Lemma~\ref{lm:linf}, amounting to 
at most $2^{O(d^2)} k \log N (\log\log N)$ samples overall. The runtime complexity is at most $2^{O(d^2)} k\log^{d+3} N$ overall for the calls to \textsc{ReduceL1Norm} and no more than $2^{O(d^2)}k \log^{d+3} N$ overall for the calls to \textsc{ReduceInfNorm}.
\end{proof}

\subsection{Analysis of \textsc{SparseFFT}}\label{sec:sfft}

\noindent {\em {\bf Theorem~\ref{thm:main}}
For any $\e>0$, $x\in \C^\nsq$ and any integer $k\geq 1$, if $R^*\geq ||x||_\infty/\mu=\poly(N)$, $\mu^2\geq ||x_{\nsq\setminus [k]}||_2^2/k$, $\mu^2=O(||x_{\nsq\setminus [k]}||_2^2/k)$  and $\alpha>0$ is smaller than an absolute constant, \textsc{SparseFFT}$(\hat x, k, \e, R^*, \mu)$ solves the $\ell_2/\ell_2$ sparse recovery problem using $2^{O(d^2)} (k\log N \log\log N+\frac1{\e}k\log N)$ samples and 
$2^{O(d^2)} \frac1{\e}k \log^{d+3} N$ time with at least $98/100$ success probability.
}

\begin{proof}
By Theorem~\ref{thm:l1snr} the set $S:=\{i\in \nsq: |x_i|>\mu\}$ satisfies
$$
||(x-\chi^{(T)})_S||_1\lesssim \mu k
$$
and
\begin{equation*}
\begin{split}
||\chi^{(T)}_{\nsq\setminus S}||_1\lesssim \mu k\\
||\chi^{(T)}_{\nsq\setminus S}||_0\lesssim \frac1{\log^{19} N}k\\
\end{split}
\end{equation*}
with probability at least $1-1/\log N$.

We now show that the signal $x':=x-\chi^{(T)}$ satisfies preconditions of Lemma~\ref{lm:const-snr} with parameter $k$. Indeed, letting $Q\subseteq \nsq$ denote the top $2k$ coefficients of $x'$, we have 
\begin{equation*}
\begin{split}
||x'_Q||_1&\leq ||x'_{Q\cap S}||_1+||\chi^{(T)}_{(Q\setminus S)\cap \supp \chi^{(T)}}||_1+|Q|\cdot ||x_{\nsq\setminus S}||_1\leq O(\mu k)\\
\end{split}
\end{equation*}
Furthermore, since $Q$ is the set of top $2k$ elements of $x'$, we have 
\begin{equation*}
\begin{split}
||x'_{\nsq\setminus Q}||_2^2&\leq ||x'_{\nsq\setminus (S\cup \supp \chi^{(T)})}||_2^2\leq ||x_{\nsq\setminus (S\cup \supp \chi^{(T)})}||_2^2\leq ||x_{\nsq\setminus S}||_2^2\\
&\leq \mu^2 |S|+||x_{\nsq\setminus [k]}||_2^2=O(\mu^2 k)
\end{split}
\end{equation*}
as required.

Thus, with at least $99/100$ probability we have by Lemma~\ref{lm:const-snr} that
$$
||x-\chi^{(T)}-\chi'||_2\leq (1+O(\e))\err_{k}(x).
$$
By a union bound over the $1/\log N$ failure probability of the SNR reduction loop we have that \textsc{SparseFFT} is correct with probability at least $98/100$, as required.

It remains to bound the sample and runtime complexity. The number of samples needed to compute 
$$
m(\wh{x}, H_r, a\star(\one, \h))\gets \Call{HashToBins}{\hat x, 0, (H_r, a\star (\one, \h))}
$$
 for all $a\in \A_r$, $\h\in \H$ is bounded by $2^{O(d^2)} k\log N (\log\log N)$ by our choice of $B=2^{O(d^2)}k$, $r_{max}=O(\log\log N)$, $|\H|=O(\log N/\log \log N)$ and $|\A_r|=O(\log\log N)$, together with Lemma~\ref{l:hashtobins}. This is asymptotically the same as the $2^{O(d^2)} k \log N (\log\log N)$ sample complexity of the $\ell_1$ norm reduction loop by Theorem~\ref{thm:l1snr}. The sampling complexity of the call to \textsc{RecoverAtConstantSNR} is at most $2^{O(d^2)} \frac1{\e}k \log N$ by Lemma~\ref{lm:const-snr}, yielding the claimed bound.

The runtime of the SNR reduction loop is bounded by $2^{O(d^2)}k \log^{d+3} N$ by Theorem~\ref{thm:l1snr}, and the runtime of \textsc{RecoverAtConstantSNR} is at most $2^{O(d^2)} \frac1{\e}k \log^{d+2} N$ by Lemma~\ref{lm:const-snr}.
\end{proof}

\section{$\ell_\infty/\ell_2$ guarantees and constant SNR case}
In this section we state and analyze our algorithm for obtaining $\ell_\infty/\ell_2$ guarantees in $\tilde O(k)$ time, as well as a procedure for recovery under the assumption of bounded $\ell_1$ norm of heavy hitters (which is very similar to the \textsc{RecoverAtConstSNR} procedure used in \cite{IKP}).

\subsection{$\ell_\infty/\ell_2$ guarantees}\label{sec:linf}
The algorithm is given as Algorithm~\ref{alg:inf-norm-reduction}.

\begin{algorithm}
\caption{\textsc{ReduceInfNorm}($\hat x, \chi, k, \nu, R^*, \mu$)}\label{alg:inf-norm-reduction} 
\begin{algorithmic}[1] 
\Procedure{ReduceInfNorm}{$\hat x, \chi, k, \nu, R^*, \mu$}
\State $\chi^{(0)} \gets 0$ \Comment{in $\C^n$}
\State $B\gets \GT\cdot k/\alpha^d$ for a small constant $\alpha>0$
\State $T\gets \log_2 R^*$
\State $r_{max}\gets (C/\sqrt{\alpha})\log N$ for sufficiently large constant $C>0$
\State $\H\gets \{\mathbf{0}_d\}$, $\Delta\gets 2^{\lfloor \frac1{2}\log_2 \log_2 n\rfloor}$ \Comment{$\mathbf{0}_d$ is the zero vector in dimension $d$}
\For {$g=1$ to $\lceil \log_\Delta n \rceil$}
\State $\H\gets \H\cup \bigcup_{s=1}^d \{n \Delta^{-g} \cdot \mathbf{e}_s\}$ \Comment{$\mathbf{e}_s$ is the unit vector in direction $s$}
\EndFor
\State $G\gets$ filter with $B$ buckets and sharpness $F$, as per Lemma~\ref{lm:filter-prop}
\For {$r=1$ to $r_{max}$}\Comment{Samples that will be used for location}
\State Choose $\Sigma_r\in \gl$, $q_r\in \nsq$ uniformly at random, let $\pi_r:=(\Sigma_r, q_r)$ and let $H_r:=(\pi_r, B, F)$
\State Let $\A_r\gets $ $C\log\log N$ elements of $\nsq\times \nsq$ sampled uniformly at random with replacement
\For {$\h\in \H$}
\State $m(\wh{x}, H_r, a\star(\one, \h))\gets \Call{HashToBins}{\hat x, 0, (H_r, a\star (\one, \h))}$ for all $a\in \A_r, \h\in \H$
\EndFor
\EndFor
\For{$t=0$ to $T-1$} \Comment{Locating elements of the residual that pass a threshold test}
\For{$r=1$ to $r_{max}$}
\State $L_r \gets \textsc{LocateSignal}\left(\chi^{(t)}, k, \{m(\wh{x}, H_r, a\star (\one, \h))\}_{r=1, a\in \A_r, \h\in \H}^{r_{max}}\right)$
\EndFor
\State $L\gets \bigcup_{r=1}^{r_{max}} L_r$
\State $\chi' \gets \Call{EstimateValues}{\hat x,  \chi^{(t)}, L, k, 1, O(\log n), 5(\nu 2^{T-(t+1)}+\mu), \infty}$
\State $\chi^{(t+1)} \gets \chi^{(t)} + \chi'$
\EndFor
\State \textbf{return} $\chi^{(T)}$
\EndProcedure 
\end{algorithmic}
\end{algorithm}

\begin{lemma}\label{lm:linf}
For any $x, \chi\in \C^n$, $x'=x-\chi$, any integer $k\geq 1$, if parameters $\nu$ and $\mu$ satisfy $\nu\geq ||x'_{[k]}||_1/k$, $\mu^2\geq ||x'_{\nsq\setminus [k]}||_2^2/k$, then the following conditions hold. If $S\subseteq \nsq$ is the set of top $k$ elements of $x'$ in terms of absolute value, and $||x'_{\nsq\setminus S}||_\infty\leq \nu$, then the output $\chi\in \C^\nsq$ of a call to \Call{ReduceInfNorm}{$\wh{x}, \chi, k, \nu, R^*, \mu$}  with probability at least $1-N^{-10}$ over the randomness used in the call satisfies
$$
||x'-\chi||_\infty\leq 8(\nu+\mu)+O(1/N^c), \text{~~~~~(all elements in $S$ have been reduced to about $\nu+\mu$)},
$$
where the $O(||x'||_\infty/N^c)$ term corresponds to polynomially small error in our computation of the semiequispaced Fourier transform. Furthermore, we have $\chi_{\nsq\setminus S}\equiv 0$.
The number of samples used is bounded by $2^{O(d^2)} k\log^3 N$. The runtime is bounded by $2^{O(d^2)} k \log^{d+3} N$.
\end{lemma}
\begin{proof}
 We prove by induction on $t$ that with probability at least $1-N^{-10}$ one has for each $t=0,\ldots, T-1$
\begin{description}
\item[{\bf (1)}] $||(x'-\chi^{(t)})_S||_\infty\leq 8(\nu 2^{T-t}+\mu)$
\item[{\bf (2)}] $\chi^{(t)}_{\nsq\setminus S}\equiv 0$
\item[{\bf (3)}] $|(x'_i-\chi^{(t)})_i|\leq |x'_i|$ for all $i\in \nsq$
\end{description}
for all such $t$.

The {\bf base} $t=0$ holds trivially. We now prove the {\bf inductive step}. First, since $r=C\log N$ for a constant $C>0$, we have by Lemma~\ref{lm:good-prob} that each $i\in S$ is isolated under at least a $1-\sqrt{\alpha}$ fraction of hashings $H_1,\ldots, H_{r_{max}}$ with probability at least $1-2^{-\Omega(\sqrt{\alpha}r_{max})}\geq 1- N^{-10}$ as long as $C>0$ is sufficiently large. This lets us invoke Lemma~\ref{lm:eh-s-sstar} with $S^*=\emptyset$. We now use Lemma~\ref{lm:eh-s-sstar} to obtain bounds on functions $e^{head}$ and $e^{tail}$ applied to our hashings $\{H_r\}$ and vector $x'$. Note that $e^{head}$ and $e^{tail}$ are defined in terms of a set $S\subseteq \nsq$ (this dependence is not made explicit to alleviate notation). We use $S=[\tilde k]$, i.e. $S$ is the top $k$ elements of $x'$. The inductive hypothesis together with the second part of Lemma~\ref{lm:eh-s-sstar} gives for each $i\in S$ 
$$
||e^{head}_S(\{H_r\}, x', \chi^{(t)})||_\infty\leq (C\alpha)^{d/2} ||(x'-\chi^{(t)})_S||_\infty.
$$
To bound the effect of tail noise, we invoke the second part of Lemma~\ref{lm:loc-tail-small}, which states that if $r_{max}=C\log N$ for a sufficiently large constant $C>0$ , we have
$||e^{tail}_S(\{H_r, \A_r\}, x')||_\infty=O(\sqrt{\alpha} \mu)$. 

These two facts together imply by the second claim of Corollary~\ref{cor:loc} that each $i\in S$ such that 
$$
|(x'-\chi^{(t)})_i|\geq 20\sqrt{\alpha} ||(x'-\chi^{(t)})_S||_\infty+20\sqrt{\alpha} \mu
$$
is located. In particular, by the inductive hypothesis this means that every $i\in S$ such that 
$$
|(x'-\chi^{(t)})_i|\geq 20\sqrt{\alpha} (\nu 2^{T-t}+2\mu)+(4\mu)
$$
is located and reported in the list $L$ . This means that 
$$
||(x'-\chi^{(t)})_{\nsq\setminus L}||_\infty\leq 20\sqrt{\alpha} (\nu 2^{T-t}+2\mu)+(4\mu),
$$
and hence it remains to show that each such element in $L$ is properly estimated in the call to \textsc{EstimateValues}, and that no elements outside of $S$ are updated. 

We first bound estimation quality. First note that by part {\bf (3)} of the inductive hypothesis together with Lemma~\ref{lm:estimate-l1l2}, {\bf (1)} one has for each $i\in L$
$$
\prob[|\chi'- (x'-\chi^{(t)})_i|>\sqrt{\alpha}\cdot (\nu+\mu)]<2^{-\Omega(r_{max})}<N^{-10},
$$
as long as $r_{max}\geq C\log N$ for a sufficiently large constant $C>0$.
This means that all elements in the list $L$ are estimated up to an additive $(\nu+\mu)/10\leq (\nu 2^{T-t}+\mu)/10$ term as long as $\alpha$ is smaller than an absolute constant. Putting the bounds above together proves part {\bf (1)} of the inductive step. 

To prove parts {\bf (2)} and {\bf (3)} of the inductive step, we recall that the only elements $i\in \nsq$ that are updated are the ones that satisfy $|\chi'|\geq 5(\nu 2^{T-(t+1)}+\mu)$. By the triangle inequality and the bound on additive estimation error above that 
$$
|(x'-\chi^{(t)})_i|\geq 5(\nu 2^{T-(t+1)}+\mu)-(\nu +\mu)/10> 4(\nu 2^{T-(t+1)}+\mu)\geq 4(\nu+\mu).
$$
Since $|(x'-\chi^{(t)})_i|\leq |x_i|$ by part {\bf (2)} of the inductive hypothesis, we have that only elements $i\in \nsq$ with  $|x'_i|\geq 4(\nu+\mu)$ are updated, but those belong to $S$ since $||x'_{\nsq\setminus S}||_\infty\leq \nu$ by assumption of the lemma. This proves part {\bf (3)} of the inductive step. Part {\bf (2)} of the inductive step follows since
$|(x'-\chi^{(t)}-\chi')_i|\leq (\nu +\mu)/10$ by the additive error bounds above, and the fact that $|(x'-\chi^{(t)})_i|>4(\nu+\mu)$. This completes the proof of the inductive step and the proof of correctness. 

\paragraph{Sample complexity and runtime} Since \textsc{HashToBins} uses $B\cdot \fc^d$  samples by Lemma~\ref{l:hashtobins}, the sample complexity of location is bounded by
$$
B\cdot \fc^d\cdot r_{max}\cdot c_{max}\cdot |\H|=2^{O(d^2)} k\log^3 N.
$$
Each call to \textsc{EstimateValues} uses $B\cdot \fc^d\cdot k \cdot r_{max}$ samples, and there are $O(\log N)$ such calls overall, resulting in sample complexity of 
$$
B\cdot \fc^d \cdot r_{max}\cdot \log N=2^{O(d^2)} k\log^2 N.
$$
Thus, the sample complexity is bounded by $2^{O(d^2)} k\log^3 N$. The runtime bound follows analogously.
\end{proof}

\subsection{Recovery at constant SNR}\label{sec:const-snr}
The algorithm is given by

\begin{algorithm}[H]
\caption{\textsc{RecoverAtConstantSNR}($\hat x, \chi, k, \e$)}\label{alg:const-snr} 
\begin{algorithmic}[1] 
\Procedure{RecoverAtConstantSNR}{$\hat x, \chi, k, \e$}
\State $B\gets  \GT\cdot k/(\e\alpha^d)$
\State Choose $\Sigma\in \gl$, $q\in \nsq$ uniformly at random, let $\pi:=(\Sigma, q)$ and let $H_r:=(\pi_r, B, F)$
\State Let $\A\gets $ $C\log\log N$ elements of $\nsq\times \nsq$ sampled uniformly at random with replacement
\State $\H\gets \{\mathbf{0}_d\}$, $\Delta\gets 2^{\lfloor \frac1{2}\log_2 \log_2 n\rfloor}$ \Comment{$\mathbf{0}_d$ is the zero vector in dimension $d$}
\For {$g=1$ to $\lceil \log_\Delta n \rceil$}
\State $\H\gets \H\cup \bigcup_{s=1}^d n \Delta^{-g} \cdot \mathbf{e}_s$ \Comment{$\mathbf{e}_s$ is the unit vector in direction $s$}
\EndFor
\For {$\h\in \H$}
\State $m(\wh{x}, H, a\star(\one, \h))\gets \Call{HashToBins}{\hat x, 0, (H, a\star (\one, \h))}$ for all $a\in \A, \h\in \H$
\EndFor
\State $L \gets \textsc{LocateSignal}\left(\chi^{(t)}, k, \{m(\wh{x}, H, a\star (\one, \h))\}_{a\in \A, \h\in \H}\right)$
\State $\chi' \gets \Call{EstimateValues}{\hat x,  \chi, 2k, \e, O(\log N), 0}$
\State $L'\gets $top $4k$ elements of $\chi'$
\State \textbf{return} $\chi+\chi'_{L'}$
\EndProcedure 
\end{algorithmic}
\end{algorithm}

Our analysis will use 
\begin{lemma}[Lemma~9.1 from~\cite{IKP}]\label{lm:comparison}
Let $x, z\in \C^n$ and $k\leq n$. Let $S$ contain the largest $k$ terms of $x$, and $T$ contain the largest $2k$ terms of $z$. Then $||x-z_T||_2^2\leq ||x-x_S||_2^2+4||(x-z)_{S\cup T}||_2^2$.
\end{lemma}

\noindent{\em {\bf Lemma~\ref{lm:const-snr}}
For any $\e>0$, $\hat x, \chi\in \C^N$, $x'=x-\chi$ and any integer $k\geq 1$ if $||x'_{[2k]}||_1\leq O(||x_{\nsq\setminus [k]}||_2\sqrt{k})$ and $||x'_{\nsq\setminus [2k]}||_2^2\leq ||x_{\nsq\setminus [k]}||_2^2$, the following conditions hold. If $||x||_\infty/\mu=N^{O(1)}$, then the output $\chi'$ of 
\Call{RecoverAtConstantSNR}{$\hat x, \chi, 2k, \e$} satisfies
$$
||x'-\chi'||^2_2\leq (1+O(\e))||x_{\nsq\setminus [k]}||_2^2
$$
with at least $99/100$ probability over its internal randomness. The sample complexity is $2^{O(d^2)}\frac1{\e}  k\log N$, and the runtime complexity is at most $2^{O(d^2)}\frac1{\e}  k \log^{d+1} N.$
}

\begin{remark}
We note that the error bound is in terms of the $k$-term approximation error of $x$ as opposed to the $2k$-term approximation error of $x'=x-\chi$.
\end{remark}
\begin{proof}
Let $S$ denote the top $2k$ coefficients of $x'$. We first derive bounds on the probability that an element $i\in S$ is not located.
Recall that by Lemma~\ref{lm:loc} for any $i\in S$ if
\begin{enumerate}
\item  $e^{head}_{i}(H, x', 0)<|x'_i|/20$;
\item  $e^{tail}_{i}(H, \A\star (\one, \h), x')< |x'_i|/20$ for all $\h\in \H$;
\item for every $s\in [1:d]$ the set $\A\star(\zero, \mathbf{e}_s)$ is balanced (as per Definition~\ref{def:balance}),
\end{enumerate}
then $i\in L$, i.e. $i$ is successfully located in \textsc{LocateSignal}.

We now upper bound the probability that an element $i\in S$ is not located.  We let $\mu^2:=||x_{\nsq\setminus k}||_2^2/k$ to simplify notation.

\paragraph{Contribution from head elements.} We need to bound, for $i\in S$, the quantity
\begin{equation*}
\begin{split}
e^{head}_i(H, x', 0)&=G_{o_i(i)}^{-1}\cdot \sum_{j\in S\setminus \{i\}} G_{o_i(j)}|x'_j|.
\end{split}
\end{equation*}
Recall that $m(\wh{x}, H, a\star(\one, \h))=\Call{HashToBins}{\hat x, 0, (H, a\star (\one, \h))}$, and let $m:=m(\wh{x}, H, a\star(\one, \h))$ to simplify notation. By Lemma~\ref{lm:hashing}, {\bf (1)} one has 
\begin{equation}\label{eq:const-h}
\expect_H[\max_{a\in \nsq} \abs{G_{o_i(i)}^{-1}\omega^{-a^T\Sigma i}m_{h(i)} - (x'_{S})_i}]\leq \GS\cdot C^d ||x'_{S}||_1/B+\mu/N^2
\end{equation}
for a constant $C>0$. This yields
$$
\expect_H[e^{head}_i(H, x', 0)]\leq \GS \cdot C^d ||x'_{S}||_1/B\lesssim  \GS \cdot C^d \mu k/B\lesssim \alpha^d C^d \e \mu.
$$
by the choice of $B$ in \textsc{RecoverAtConstantSNR}. Now by Markov's inequality we have for each $i\in \nsq$ 
\begin{equation}\label{eq:const-eh-prob}
\prob_H[e^{head}_{i}(H, x', 0)>|x'_i|/20]\lesssim \alpha^d C^d \e \mu/|x'_i|\lesssim \alpha \e \mu/|x'_i|
\end{equation}
as long as $\alpha$ is smaller than a constant.

\paragraph{Contribution of tail elements}
We restate the definitions of $e^{tail}$ variables here for convenience of the reader (see~\eqref{eq:et-pi}, \eqref{eq:et-pi-a-h}, \eqref{eq:et-pi-a} and \eqref{eq:et}). 

We have
\begin{equation*}
e^{tail}_i(H, z, x):=\left|G_{o_i(i)}^{-1}\cdot \sum_{j\in \nsq\setminus S} G_{o_i(j)}x_j \omega^{z^T \Sigma (j-i)}\right|.
\end{equation*}

For any $\mathcal{Z}\subseteq \nsq$ we have
\begin{equation*}
e^{tail}_i(H, \mathcal{Z}, x):=\quant^{1/5}_{z\in \mathcal{Z}} \left|G_{o_i(i)}^{-1}\cdot \sum_{j\in \nsq\setminus S} G_{o_i(j)}x_j \omega^{z^T \Sigma (j-i)}\right|.
\end{equation*}
Note that the algorithm first selects sets $\A_r\subseteq \nsq\times \nsq$, and then accesses the signal at locations given by $\A_r\star (\one, \h), \h\in \H$ (after permuting input space).

The definition of $e^{tail}_i(H, \A_r, x')$ for permutation $\pi=(\Sigma, q)$ allows us to capture the amount of noise that our measurements for locating a specific set of bits of $\Sigma i$ suffer from. Since the algorithm requires all $\h\in \H$ to be not too noisy in order to succeed (see preconditions 2 and 3 of Lemma~\ref{lm:loc}), we have
\begin{equation*}
e^{tail}_i(H, \A, x')=40\mu_{H, i}(x)+\sum_{\h\in \H} \left|e^{tail}_i(H, \A\star (\one, \h), x')-40\mu_{H, i}(x')\right|_+
\end{equation*}
where for any $\eta\in \mathbb{R}$ one has $|\eta|_+=\eta$ if $\eta>0$ and $|\eta|_+=0$ otherwise.

For each $i\in S$ we now define an error event $\E^*_i$ whose non-occurrence implies location of element $i$, and then show that for each $i\in S$ one has
\begin{equation}\label{eq:g}
\prob_{H, \A}[\E^*_i]\lesssim \frac{\alpha \e \mu^2}{|x'_i|^2}.
\end{equation}
Once we have ~\eqref{eq:g}, together with~\eqref{eq:const-eh-prob} it allows us to prove the main result of the lemma. In what follows we concentrate on proving 
~\eqref{eq:g}. Specifically, for each $i\in S$ define
\begin{equation*}
\E^*_i=\{(H, \A): \exists \h\in \H \text{~s.t.~} e^{tail}_i(H,\A\star (\one, \h), x')>|x'_i|/20\}.
\end{equation*}

Recall that $e^{tail}_i(H, z, x')=\Call{HashToBins}{\widehat{x_{\nsq\setminus S}}, \chi_{\nsq\setminus S}, (H, z)}$ by definition of the measurements $m$.  By Lemma~\ref{lm:hashing}, {\bf (3)} one has, for a uniformly random $z\in \nsq$, that 
$\expect_z[|e^{tail}_i(H, z, x')|^2|]=\mu_{H, i}^2(x')$. By Lemma~\ref{lm:hashing}, {\bf (2)}, we have that 
\begin{equation}\label{eq:expect-sigmaq}
\expect_H[\mu^2_{H, i}(x')] \leq  \GSS\cdot C^d \norm{2}{(x-\chi)_{\nsq\setminus S}}^2/B+\mu^2/N^2\leq  \alpha \e \mu^2.
\end{equation} 
Thus by Markov's inequality 
$$
\prob_{z}[e^{tail}_{i}(H, z, x')^2> (|x'_i|/20)^2]\leq \alpha\e (\mu_{H, i}(x'))^2/(|x'_i|/20)^2.
$$
Combining this with Lemma~\ref{lm:quant-exp}, we get for all $\tau\leq (1/20)(|x'_i|/20)$ and all $\h\in \H$
\begin{equation}\label{eq:xyz}
\prob_{\A}[\quant^{1/5}_{z\in \A\star (\one, \h)} e^{tail}_{i}(H, z, x')> |x'_i|/20| \mu_{H, i}^2(x')=\tau]<(4\tau/(|x'_i|/20))^{\Omega(|\A|)}.
\end{equation}

Equipped with the bounds above, we now bound $\prob[\E^*_i]$. To that effect, for each $\tau>0$ let the event $\E_i(\tau)$ be defined as 
$\E_i(\tau)=\{\mu_{H, i}(x')=\tau\}$. Note that since we assume that we operate on $O(\log n)$ bit integers, $\mu_{H, i}(x')$ takes on a finite number of values, and hence $\E_i(\tau)$ is well-defined. It is convenient to bound $\prob[\E_i^*]$ as a sum of three terms:
\begin{equation*}
\begin{split}
\prob_{H, \A}[\E^*_i]&\leq \prob_{H, \A}\left[\left.e^{tail}_{i}(H, \A, x')> |x'_i|/20\right|\bigcup_{\tau\leq \sqrt{\alpha \e}\mu}\E_i(\tau)\right] \\
&+\int_{\sqrt{\alpha \e}\mu}^{(1/8)(|x'_i|/20)} \prob_{H, \A}\left[\left.e^{tail}_{i}(H, \A, x')> |x'_i|/20\right.|\E_i(\tau)\right] \prob[\E_i(\tau)]d\tau\\
&+\int_{(1/8)(|x'_i|/20)}^\infty \prob[\E_i(\tau)]d\tau\\
\end{split}
\end{equation*}

We now bound each of the three terms separately for $i$ such that $|x'_i|/20\geq 2\sqrt{\alpha\e} \mu_{H, i}(x')$. This is sufficient for our purposes, as other elements only contribute a small amount of $\ell_2^2$ mass.
\begin{description}
\item[1.] By \eqref{eq:xyz} and a union bound over $\H$ the first term is bounded by 
\begin{equation}\label{eq:r-ub-1}
|\H|\cdot (\sqrt{\alpha \e}\mu/(|x'_i|/20))^{\Omega(|A|)}\leq \alpha \e \mu^2/|x'_i|^2\cdot |\H|\cdot 2^{-\Omega(|A|)}\leq \alpha\e \mu^2/|x'_i|^2.
\end{equation}
since $|\A|\geq C\log \log N$ for a sufficiently large constant $C>0$ in \textsc{RecoverAtConstantSNR}.

\item[2.] The second term, again by a union bound over $\H$ and using \eqref{eq:xyz}, is bounded by 
\begin{equation}\label{eq:const-snr-abc}
\begin{split}
&\int_{\sqrt{\alpha \e}\mu}^{(1/8)(|x'_i|/20)} |\H|\cdot (4\tau /(|x'_i|/20))^{\Omega(|A|)} \prob[\E_i(\tau)]d\tau\\
\leq &\int_{\sqrt{\alpha \e}\mu}^{(1/8)(|x'_i|/20)} |\H|\cdot (4\tau /(|x'_i|/20))^{\Omega(|A|)} (4\tau /(|x'_i|/20))^{2}\prob[\E_i(\tau)] d\tau\\
\end{split}
\end{equation}
Since $|\A|\geq C\log \log N$ for a sufficiently large constant $C>0$ and $(4\tau /(|x'_i|/20))\leq 1/2$ over the whole range of $\tau$ by our assumption that $|x'_i|/20\geq 2\sqrt{\alpha\e} \mu_{H, i}(x')$, we have
$$
|\H|\cdot (4\tau /(|x'_i|/20))^{\Omega(|A|)}\leq |\H|\cdot (1/2)^{\Omega(|A|)}=o(1)
$$ 
for each $\tau\in [\sqrt{\alpha \e}\mu, (1/8)(|x'_i|/20)]$. Thus, \eqref{eq:const-snr-abc} is upper bounded by
\begin{equation*}
\begin{split}
\int_{\sqrt{\alpha \e}\mu}^{(1/8)(|x'_i|/20)} (4\tau /(|x'_i|/20))^{2}\prob[\E_i(\tau)] d\tau\\
\lesssim \frac1{(|x'_i|/20)^{2}}\int_{\sqrt{\alpha \e}\mu}^{(1/8)(|x'_i|/20)} \tau^2\prob[\E_i(\tau)] d\tau\\
\leq \frac{\alpha \e \mu^2}{(|x'_i|/20)^2}\\
\end{split}
\end{equation*}
since 
$$
\int_{\sqrt{\alpha \e}\mu}^{(1/8)(|x'_i|/20)} \tau^2\prob[\E_i(\tau)] d\tau\leq \int_{0}^{\infty} \tau^2\prob[\E_i(\tau)] d\tau=\expect_H[\mu^2_{H, i}(x')]=O(\alpha)\e \mu^2
$$
by \eqref{eq:expect-sigmaq}.

\item [3.] For the third term we have
$$
\int_{(1/8)(|x'_i|/20)}^\infty \prob[\E_i(\tau)]d\tau=\prob[\mu_{H, i}(x')>(1/8)(|x'_i|/20)]\lesssim \frac{\alpha \e \mu^2}{|x'_i|^2}
$$
by Markov's inequality applied to \eqref{eq:expect-sigmaq}.
\end{description}
Putting the three estimates together, we get $\prob[\E_i^*]=\frac{O(\alpha)\e \mu^2}{|x'_i|^2}$. Together with \eqref{eq:const-eh-prob} this yields for $i\in S$
$$
\prob[i\not \in L]\lesssim \frac{\alpha \e \mu^2}{|x'_i|^2}+\frac{\alpha \e \mu}{|x'_i|}.
$$

In particular, 
\begin{equation*}
\begin{split}
\expect\left[\sum_{i\in S} |x'_i|^2\cdot \mathbf{1}_{i\in S\setminus L}\right]&\leq \sum_{i\in S} |x'_i|^2 \prob[i\not \in L]\\
&\lesssim \sum_{i\in S} |x'_i|^2 \left(\frac{\alpha \e \mu}{|x'_i|}+\frac{\alpha \e \mu^2}{|x'_i|^2}\right)\lesssim \alpha \e\mu^2 k,
\end{split}
\end{equation*}
where we used the assumption of the lemma that $||x'_{[2k]}||_1\leq O(||x_{\nsq\setminus [k]}||_2\sqrt{k})$ and $||x'_{\nsq\setminus [2k]}||_2^2\leq ||x_{\nsq\setminus [k]}||_2^2$ in the last line.
By Markov's inequality we thus have $\prob[||x'_{S\setminus L}||_2^2>\e \mu^2 k]<1/10$ as long as $\alpha$ is smaller than a constant. 

We now upper bound $||x'-\chi'||_2^2$. We apply Lemma~\ref{lm:comparison} to vectors $x'$ and $\chi'$ with sets $S$ and $L'$ respectively, getting 
\begin{equation*}
\begin{split}
||x'-\chi'_{L'}||_2^2&\leq ||x'-x'_S||_2^2+4||(x'-\chi')_{S\cup L'}||_2^2\\
&\leq ||x_{\nsq\setminus [k]}||_2^2+4||(x'-\chi')_{S\setminus L}||_2^2+4||(x'-\chi')_{S\cap L}||_2^2\\
&\leq ||x_{\nsq\setminus [k]}||_2^2+4\e \mu^2 k+4\e \mu^2 |S|\\
&\leq ||x_{\nsq\setminus [k]}||_2^2+O(\e \mu^2 k),\\
\end{split}
\end{equation*}
where we used the fact that $||(x'-\chi')_{S\cap L}||_\infty\leq \sqrt{\e} \mu$ with probability at least $1-1/N$ over the randomness used in \textsc{EstimateValues} by Lemma~\ref{lm:estimate-l1l2}, {\bf (3)}.
This completes the proof of correctness.

\paragraph{Sample complexity and runtime}The number of samples taken is bounded by $2^{O(d^2)}\frac1{\e} k\log N$ by Lemma~\ref{l:hashtobins}, the choice of $B$. The sampling complexity of the call to \textsc{EstimateValues} is at most $2^{O(d^2)}\frac1{\e} k\log N$. The runtime is bounded by $2^{O(d^2)}\frac1{\e}k \log^{d+1} N \log\log N$ for computing the measurements $m(\wh{x}, H, a\star (\one, \h))$ and $2^{O(d^2)}\frac1{\e} k \log^{d+1} N$ for estimation.
\end{proof}
\section{Utilities}\label{sec:utils}

\subsection{Properties of \textsc{EstimateValues}}
In this section we describe the procedure \textsc{EstimateValues} (see Algorithm~\ref{alg:est}), which, given access to a signal $x$ in frequency domain (i.e. given $\wh{x}$), a partially recovered signal $\chi$ and a target list of locations $L\subseteq \nsq$,
estimates values of the elements in $L$, and outputs the elements that are above a threshold $\nu$ in absolute value. The SNR reduction loop uses the thresholding function of \textsc{EstimateValues} and passes a nonzero threshold, while \textsc{RecoverAtConstantSNR} uses $\nu=0$.

\begin{algorithm}
\caption{\textsc{EstimateValues}($x, \chi, L, k, \e, \nu, r_{max}$)}\label{alg:est} 
\begin{algorithmic}[1] 
\Procedure{EstimateValues}{$x, \chi, L, k, \e, \nu, r_{max}$}\Comment{$r_{max}$ controls estimate confidence}
\State $B\gets \GT\cdot k/(\e \alpha^{2d})$
\For {$r=0$ to $r_{max}$}
\State Choose $\Sigma_r\in \gl, q_r, z_r\in \nsq$ uniformly at random
\State Let $\pi_r:=(\Sigma_r, q_r)$, $H_r:=(\pi_r, B, F), F= 2d$
\State $u_r \gets \Call{HashToBins}{\hat x, \chi, \chi, (H_r, z_r)}$
\State \Comment{Using semi-equispaced Fourier transform  (Corollary~\ref{c:semiequi1})}
\EndFor
\State $L'\gets \emptyset$\Comment{Initialize output list to empty}
\For{$f\in L$}
\For{$r=0$ to $r_{max}$}
\State $j\gets h_r(f)$ 
\State $w^r_f\gets v_{r, j} G^{-1}_{o_f(f)}\omega^{-z_r^T \Sigma_r f}$ \Comment{Estimate $x'_f$ from each measurement}
\EndFor
\State $w_f\gets \text{median}\{w^r_f\}_{r=1}^ {r_{max}}$ \Comment{Median is taken coordinatewise}
\State {\bf If~} $|w_f|>\nu$ {\bf then~} $L'\gets L'\cup \{f\}$
\EndFor
\State \textbf{return} $w_{L'}$
\EndProcedure 
\end{algorithmic}
\end{algorithm}

\begin{lemma}[$\ell_1/\ell_2$ bounds on estimation quality]\label{lm:estimate-l1l2}
For any $\e\in (0, 1]$, any $x, \chi\in \C^n, x'=x-\chi$, any $L\subseteq \nsq$, any integer $k$ and any set $S\subseteq \nsq, |S|\leq 2k$ the following conditions hold. If $\nu\geq ||(x-\chi)_S||_1/k$ and $\mu^2\geq ||(x-\chi)_{\nsq\setminus S}||_2^2/k$, then the output $w$ of \Call{EstimateValues}{$\hat x, \chi, L, k, \e, \nu, r_{max}$} satisfies the following bounds if $r_{max}$ is larger than an absolute constant.

For each $i\in L$ 
\begin{description}
\item[(1)] $\prob[|w_i- x'_i|>\sqrt{\e \alpha}(\nu+\mu)]<2^{-\Omega(r_{max})}$;
\item[(2)] $\expect\left[\left||w_i- x'_i|-\sqrt{\e \alpha}(\nu+\mu)\right|_+\right]\leq \sqrt{\e\alpha}(\nu+\mu)2^{-\Omega(r_{max})}$;
\item[(3)] $\expect\left[\left||w_i-x'_i|^2-\e \alpha(\nu+\mu)^2\right|_+\right]\leq 2^{-\Omega(r_{max})}\e(\nu^2+\mu^2)$.
\end{description}
The sample complexity is bounded by $\frac1{\e}2^{O(d^2)} k r_{max}$. The runtime is bounded by $2^{O(d^2)}\frac1{\e}k \log^{d+1} N r_{max}$.
\end{lemma}
\begin{proof}
We analyze  the vector $u_{r} \gets \Call{HashToBins}{\hat x, \chi, (H_r, z_r)}$ using the approximate linearity of \textsc{HashToBins} given by Lemma~\ref{lm:linearity} (see Appendix~\ref{app:A}).  Writing $x'=x'_S+x'_{\nsq\setminus S}$, we let 
$$
u^{head}_r:=\Call{HashToBins}{\widehat{x_S}, \chi_S, (H_r, z_r)}\text{~~~and~~~}u^{tail}_r:=\Call{HashToBins}{\widehat{x_{\nsq\setminus S}}, \chi_{\nsq\setminus S}, (H_r, z_r)}
$$
we apply Lemma~\ref{lm:hashing}, {\bf (1)} to the first vector, obtaining 
\begin{equation}\label{eq:ig34gwgt43tUf}
\expect_{H_r, z_r}[\abs{G_{o_i(i)}^{-1}\omega^{-z_r^T\Sigma i}u^{head}_{h(i)} - (x'_S)_i}]\leq \GS\cdot C^d ||x_S||_1/B+\mu/N^2
\end{equation} 
Similarly applying Lemma~\ref{lm:hashing}, {\bf (2)} and {\bf (3)} to the $u^{tail}$, we get
\begin{equation*}
\expect_{H_r, z_r}[\abs{G_{o_i(i)}^{-1}\omega^{-z_r^T\Sigma i}u^{tail}_{h_r(i)} - (x'_{\nsq\setminus S})_i}^2]\leq \GSS\cdot C^d \norm{2}{x'_{\nsq\setminus S}}^2/B,
\end{equation*} 
which by Jensen's inequality implies 
\begin{equation}\label{eq:02ht42jggdsgfgxcc}
\begin{split}
\expect_{H_r, z_r}[\abs{G_{o_i(i)}^{-1}\omega^{-a_r^T\Sigma i}u^{tail}_{h(i)} - ((x-\chi)_{\nsq\setminus S})_i}]&\leq \GS\cdot C^d \sqrt{\norm{2}{x_{\nsq\setminus S}}^2/B}\\
&\leq \GS\cdot C^d \mu\cdot \sqrt{k/B}.
\end{split}
\end{equation}
 Putting~\eqref{eq:ig34gwgt43tUf} and ~\eqref{eq:02ht42jggdsgfgxcc} together and using Lemma~\ref{lm:linearity}, we get 
\begin{equation}\label{eq:02ht42jggdsggggrg344c}
\expect_{H_r, z_r}[\abs{G_{o_i(i)}^{-1}\omega^{-z_r^T\Sigma i}u_{h(i)} - (x-\chi)_i}]\leq \GS\cdot C^d (||x_S||_1/B+\mu \cdot \sqrt{k/B}).
\end{equation}
We hence get by Markov's inequality together with the choice $B=\GT\cdot k/(\e \alpha^{2d})$ in \textsc{EstimateValues} (see Algorithm~\ref{alg:est})
\begin{equation}\label{eq:geropghj3g34t}
\prob_{H_r, z_r}[\abs{G_{o_i(i)}^{-1}\omega^{-z_r^T\Sigma i}u_{h(i)} - (x-\chi)_i}>\frac1{2}\sqrt{\e \alpha}(\nu+\mu)]\leq (C\alpha)^{d/2}.
\end{equation}
The rhs is smaller than $1/10$ as long as $\alpha$ is smaller than an absolute constant. 

Since $w_i$ is obtained by taking the median in real and imaginary components, we get by Lemma~\ref{lm:median-est}
$$
|w_i-x_i'|\leq 2\text{median}(|w^1_i-x'_i|, \ldots, |w^{r_{max}}_i-x'_i|).
$$

By ~\eqref{eq:geropghj3g34t} combined with Lemma~\ref{lm:quant-exp} with $\gamma=1/10$ we thus have
$$
\prob_{\{H_r, z_r\}}[|w_i-x'_i|>\sqrt{\e \alpha}(\nu+\mu)]<2^{-\Omega(r_{max})}.
$$
This establishes {\bf (1)}. {\bf (2)} follows similarly by applying the first bound from Lemma~\ref{lm:quant-exp} with $\gamma=1/2$ to random variables $X_r=|w^r_i-x_i|, r=1,\ldots, r_{max}$ and $Y=|w_i-x_i|$.
The third claim of the lemma follows analogously.

The sample and runtime bounds follow by Lemma~\ref{l:hashtobins} and Lemma~\ref{l:semiequi1} by the choice of parameters.
\end{proof}

\subsection{Properties of \textsc{HashToBins}}\label{sec:hash2bins}
\begin{algorithm}
\caption{Hashing using Fourier samples (analyzed in Lemma~\ref{l:hashtobins})}\label{alg:hash2bins} 
\begin{algorithmic}[1] 
\Procedure{HashToBins}{$\wh{x}, \chi, (H, a)$}
\State $G \gets$ filter with $B$ buckets, $\fc=2d$\Comment{$H=(\pi, B, F), \pi=(\Sigma q)$}
\State Compute $y'=\hat G\cdot P_{\Sigma, a, q}(\hat x-\hat \chi')$,
for some $\chi'$ with $\norm{\infty}{\wh{\chi}-\wh{\chi}'}<N^{-\Omega(c)}$\Comment{$c$ is a large constant}
\State Compute $u_j = \sqrt{N}\F^{-1}(y')_{(n/b)\cdot j}$ for $j \in [b]^d$
\State {\bf return} $u$
\EndProcedure 
\end{algorithmic}
\end{algorithm}

\begin{lemma}\label{l:hashtobins}
  \Call{HashToBins}{$\wh{x}, \chi, (H, a)$} computes
  $u$ such that for any $i \in [n]$,
  \[
  u_{h(i)} = \Delta_{h(i)} + \sum_j G_{o_i(j)}(x - \chi)_j \omega^{a^T\Sigma j}
  \]
  where $G$ is the filter defined in section~\ref{sec:prelim}, and $\Delta_{h(i)}^2 \leq \norm{2}{\chi}^2
  / ((R^*)^2N^{11})$ is a negligible error term.  It takes $O(B \fc^d)$
  samples, and if $\norm{0}{\chi} \lesssim B$, it takes $O(2^{O(d)}\cdot B\log^d N)$ time.
\end{lemma}
\begin{proof}
  Let $S = \supp(\wh{G})$, so $\abs{S} \lesssim (2\fc)^d\cdot B$ and in fact
  $S \subset \B^\infty_{\fc \cdot B^{1/d}}(0)$.

  First, \textsc{HashToBins} computes
 \[
  y'=\wh{G} \cdot P_{\Sigma, a, q}\wh{x - \chi'}=\wh{G} \cdot P_{\Sigma, a, q}\wh{x - \chi}+\wh{G} \cdot P_{\Sigma, a, q}\wh{\chi - \chi'},
  \]
  for an approximation $\wh{\chi}'$ to $\wh{\chi}$.  This is efficient
  because one can compute $(P_{\Sigma, a, q}\wh{x})_S$ with
  $O(\abs{S})$ time and samples, and $P_{\Sigma, a, q}\wh{\chi}'_S$ is
  easily computed from $\wh{\chi}'_T$ for $T = \{\Sigma(j-b) : j \in
  S\}$.  Since $T$ is an image of an $\ell_\infty$ ball under a linear transformation and $\chi$ is $B$-sparse, by
Corollary~\ref{c:semiequi1}, an approximation $\wh{\chi}'$ to
  $\wh{\chi}$ can be computed in $O(2^{O(d)}\cdot B\log^d N)$ time such
  that $|\wh{\chi}_i-\wh{\chi}'_i|<N^{-\Omega(c)}$
  for all $i\in T$. Since $\norm{1}{\wh{G}} \leq
  \sqrt{N}\norm{2}{\wh{G}} = \sqrt{N}\norm{2}{G} \leq N
  \norm{\infty}{G} \leq N$ and $\wh{G}$ is $0$ outside $S$, this
  implies that
  \begin{equation}\label{eq:gh-norm}
    \norm{2}{\wh{G}\cdot P_{\Sigma, a, q}(\wh{\chi-\chi'})}\leq \norm{1}{\wh{G}}\max_{i \in S} \abs{(P_{\Sigma, a, q}(\wh{\chi-\chi'}))_i} = \norm{1}{\wh{G}}\max_{i \in T} \abs{(\wh{\chi-\chi'})_i} \leq N^{-\Omega(c)}
  \end{equation}
  as long as $c$ is larger than an absolute constant. 
  Define $\Delta$ by $\wh{\Delta}=\sqrt{N}\wh{G}\cdot P_{\Sigma, a, q}(\wh{\chi-\chi'})$.  Then \textsc{HashToBins} computes $u \in
  \C^B$ such that for all $i$,
  \[
  u_{h(i)} =
  \sqrt{N}\F^{-1}(y')_{(n/b)\cdot h(i)}=\sqrt{N}\F^{-1}(y)_{(n/b)\cdot h(i)}+\Delta_{(n/b)\cdot h(i)},
  \]
  for $y = \wh{G} \cdot P_{\Sigma, a, q}\wh{x - \chi}$.  This
  computation takes $O(\norm{0}{y'} + B \log B) \lesssim B \log (N)$
  time.  We have by the convolution theorem that
  \begin{align*}
    u_{h(i)} &= \sqrt{N}\F^{-1}(\wh{G} \cdot P_{\Sigma, a, q}\wh{(x - \chi)})_{(n/b)\cdot h(i)} + \Delta_{(n/b)\cdot h(i)}\\
    &= (G * \F(P_{\Sigma, a, q}\wh{(x - \chi)}))_{(n/b)\cdot h(i)}+\Delta_{(n/b)\cdot h(i)}\\
    &= \sum_{\pi(j) \in [N]} G_{(n/b)\cdot h(i)-\pi(j)} \F(P_{\Sigma, a, q}\wh{(x - \chi)})_{\pi(j)}+\Delta_{(n/b)\cdot h(i)}\\
    &= \sum_{i \in [N]} G_{o_i(j)} (x - \chi)_{j} \omega^{a^T\Sigma j}+\Delta_{(n/b)\cdot h(i)}
  \end{align*}
  where the last step is the definition of $o_i(j)$ and Lemma~\ref{lm:perm}.
  
Finally, we note that 
\[
|\Delta_{(n/b)\cdot h(i)}|\leq \norm{2}{\Delta} =\norm{2}{\wh{\Delta}}=\sqrt{N}\norm{2}{\wh{G}\cdot P_{\Sigma, a, q}(\wh{\chi-\chi'})}\leq N^{-\Omega(c)},
\]
 where we used \eqref{eq:gh-norm} in the last step. This completes the proof.
\end{proof}

\subsection{Lemmas on quantiles and the median estimator}
In this section we prove several lemmas   useful for analyzing the concentration properties of the median estimate. We will use 

\begin{theorem}[Chernoff bound]\label{thm:chernoff}
Let $X_1,\ldots, X_n$ be independent $0/1$ Bernoulli random variables with $\sum_{i=1}^n \expect[X_i]=\mu$. Then for any $\delta>1$ one has $\prob[\sum_{i=1}^n X_i>(1+\delta)\mu]<e^{-\delta \mu/3}$.
\end{theorem}

\begin{lemma}[Error bounds for the median estimator]\label{lm:median-est}
Let $X_1,\ldots, X_n\in \C$ be independent random variables. Let $Y:=\text{median} (X_1,\ldots, X_n)$, where the median is applied coordinatewise. Then for any $a\in \C$ one has
\begin{equation*}
\begin{split}
|Y-a|\leq &2\text{median}(|X_1-a|,\ldots, |X_n-a|)\\
= &2\sqrt{\text{median}(|X_1-a|^2,\ldots, |X_n-a|^2)}.\\
\end{split}
\end{equation*}

\end{lemma}
\begin{proof}
Let $i, j\in [n]$ be such that $Y=\text{re}(X_i)+\mathbf{i}\cdot \text{im}(X_j)$. 
Suppose that $\text{re}(X_i)\geq \text{re}(a)$ (the other case is analogous). Then since $\text{re}(X_i)$  is the median in the list $(\text{re}(X_1),\ldots, \text{re}(X_n))$ by definition of $Y$, we have that at least half of $X_s, s=1,\ldots, n$ satisfy $|\text{re}(X_s)-\text{re}(a)|>|\text{re}(X_i)-\text{re}(a)|$, and hence 
\begin{equation}\label{eq:25855ughgf-1}
|\text{re}(X_i)-\text{re}(a)|\leq \text{median}(|\text{re}(X_1)-\text{re}(a)|,\ldots, |\text{re}(X_n)-\text{re}(a)|).
\end{equation}
Since squaring a list of numbers preserves the order, we also have
\begin{equation}\label{eq:25855ughgf-2}
(\text{re}(X_i)-\text{re}(a))^2\leq \text{median}((\text{re}(X_1)-\text{re}(a))^2,\ldots, (\text{re}(X_n)-\text{re}(a))^2).
\end{equation}

 A similar argument holds for the imaginary part.
Combining
$$
|Y-a|^2=(\text{re}(a)-\text{re}(X_i))^2+(\text{im}(a)-\text{im}(X_i))^2
$$
with~\eqref{eq:25855ughgf-1} gives 
\begin{equation*}
\begin{split}
|Y-a|^2\leq &\text{median} ((\text{re}(X_1)-\text{re}(a))^2,\ldots, (\text{re}(X_n)-\text{re}(a))^2)\\
&+\text{median} ((\text{im}(X_1)-\text{im}(a))^2,\ldots, (\text{im}(X_n)-\text{im}(a))^2)\\
\end{split}
\end{equation*}
Noting that 
$$
|Y-a|=((\text{re}(a)-\text{re}(X_i))^2+(\text{im}(a)-\text{im}(X_i))^2)^{1/2}\leq |\text{re}(a)-\text{re}(X_i)|+|\text{im}(a)-\text{im}(X_i)|
$$
and using ~\eqref{eq:25855ughgf-2}, we also get 
\begin{equation*}
\begin{split}
|Y-a|\leq &\text{median}(|\text{re}(X_1)-\text{re}(a)|,\ldots, |\text{re}(X_n)-\text{re}(a)|)\\
&+\text{median} (|\text{im}(X_1)-\text{im}(a)|,\ldots, |\text{im}(X_n)-\text{im}(a)|).
\end{split}
\end{equation*}
The results of the lemma follow by noting that $|\text{re}(X)-\text{re}(a)|\leq |X-a|$ and $|\text{im}(X)-\text{im}(a)|\leq |X-a|$.
\end{proof}

\begin{lemma}\label{lm:quant-exp}
Let $X_1,\ldots, X_n\geq 0$ be independent random variables with $\expect[X_i]\leq \mu$ for each $i=1,\ldots, n$. Then for any $\gamma\in (0, 1)$ if $Y\leq\quant^{\gamma} (X_1,\ldots, X_n)$,
then 
$$
\expect[\left|Y-4\mu/\gamma\right|_+]\leq (\mu/\gamma) \cdot 2^{-\Omega(n)}
$$
and 
$$
\prob[Y\geq 4\mu/\gamma]\leq 2^{-\Omega(n)}.
$$
\end{lemma}
\begin{proof}

For any $t\geq 1$ by Markov's inequality $\prob[X_i>t \mu/\gamma]\leq \gamma/t$. Define indicator random variables $Z_i$ by letting $Z_i=1$ if $X_i>t \mu/\gamma$ and $Z_i=0$ otherwise. Note that $\expect[Z_i]\leq \gamma/t$ for each $i$.
Then since $Y$ is bounded above by the $\gamma n$-th largest of $\{X_i\}_{i=1}^n$, we have $\prob[Y>t \mu/\gamma]\leq \prob[\sum_{i=1}^n Z_i\geq \gamma n]$. As $\expect[Z_i]\leq \gamma/t$, this can only happen if the sum $\sum_{i=1}^n Z_i$ exceeds expectation by a factor of at least $t$. We now apply Theorem~\ref{thm:chernoff} to the sequence $Z_i,i=1,\ldots, n$. We have
\begin{equation}\label{eq:vi34g43qqg}
\prob\left[\sum_{i=1}^n Z_i\geq \gamma n\right]\leq e^{-(t-1)\gamma n/3}
\end{equation}
by Theorem~\ref{thm:chernoff} invoked with $\delta=t-1$. The assumptions of Theorem~\ref{thm:chernoff} are satisfied as long as $t> 2$.  This proves the second claim we have $t=4$ in that case.

For the first claim we have
\begin{equation*}
\begin{split}
\expect[Y\cdot \mathbf{1}_{Y\geq 4\cdot \mu/\gamma}]&\leq \int_{4}^\infty t \mu \cdot \prob[Y\geq t\cdot \mu/\gamma]dt\\
&\leq \int_{4}^\infty t \mu e^{-(t-1)n/3}dt\text{~~~~~~~~~~ (by~\eqref{eq:vi34g43qqg})}\\
&\leq e^{-n/3}\int_{4}^\infty t \mu e^{-(t-2)n/3}dt\\
&=O(\mu\cdot e^{-n/3})
\end{split}
\end{equation*}
as required.
\end{proof}

\section{Semi-equispaced Fourier Transform}\label{sec:semiequi}
In this section we give an algorithm for computing the semi-equispaced Fourier transform, prove its correctness and give runtime bounds.
\begin{algorithm}[H]
\caption{Semi-equispaced Fourier Transform in $2^{O(d)}k \log^{d} N$ time}\label{alg:semi-fft}
\begin{algorithmic}[1]
\Procedure{SemiEquispacedFFT}{$x$, $c$}\Comment{$x \in \C^{\nsq}$ is $k$-sparse}
\State Let $B\geq 2^d k$, be a power of $2^d$, $b=B^{1/d}$
\State $G, \wh{G'} \gets $ $d$-th tensor powers of the flat window functions of~\cite{HIKP2}, see below
\State $y_i \gets \frac{1}{\sqrt{N}}(x * G)_{i\cdot \frac{n}{2b}}$ for $i \in [2b]^d$.
\State $\wh{y} \gets \textsc{FFT}(y)$ \Comment{FFT on $[2b]^d$}
\State $\wh{x}'_i \gets \wh{y}_i$ for $||i||_\infty \leq b/2$.
\State \Return $\wh{x}'$
\EndProcedure
\end{algorithmic}
\end{algorithm}

  We define filters $G, \wh{G'}$ as $d$-th tensor powers of the flat window functions of~\cite{HIKP2}, so
  that $G_i = 0$ for all $||i||_\infty \gtrsim c (n/b) \log N$,
  $\norm{2}{G -G'} \leq N^{-c}$,
  \[
  \wh{G'}_i = \left\{
    \begin{array}{cl}
      1 & \text{ if } ||i||_\infty \leq b/2\\
      0 & \text{ if } ||i||_\infty > b\\
    \end{array}
  \right.,
  \]
  and $\wh{G'}_i \in [0, 1]$ everywhere.

The following is similar to results of~\cite{DR93, IKP}.

\begin{lemma}\label{l:semiequi1}
  Let $n$ be a power of two, $N=n^d$, $c\geq 2$ a constant. Let integer $B\geq 1$, be a power of $2^d$, $b=B^{1/d}$. For any $x\in \C^\nsq$ Algorithm~\ref{alg:semi-fft} computes $\wh{x}'_i$
  for all $||i||_\infty \leq b/2$ such that
  \[
  \abs{\wh{x}'_i - \wh{x}_i} \leq \norm{2}{x} / N^c
  \]
in $c^{O(d)}||x||_0 \log^{d} N+2^{O(d)}B\log B$   time.
\end{lemma}
\begin{proof}
  Define
  \[
  z = \frac{1}{\sqrt{N}}x * G.
  \]
  We have that $\wh{z}_i = \wh{x}_i \wh{G}_i$ for all $i\in \nsq$.
  Furthermore, because subsampling and aliasing are dual under the
  Fourier transform, since $y_i = z_{i\cdot (n/2b)}, i\in [2b]^d$ is a subsampling of $z$ we have for
  $i$ such that $||i||_\infty \leq b/2$ that
  \begin{align*}
    \wh{x}'_i = \wh{y}_i &= \sum_{j \in [n/(2b)]^d} \wh{z}_{i + 2b\cdot j} \\
    &= \sum_{j\in [n/(2b)]^d} \wh{x}_{i + 2b \cdot j}\wh{G}_{i+2b\cdot j} \\
    &= \sum_{j\in [n/(2b)]^d } \wh{x}_{i + 2b\cdot j}\wh{G'}_{i+2b\cdot j} + \sum_{j\in [n/(2b)]^d} \wh{x}_{i + 2b\cdot j}(\wh{G}_{i+2b\cdot j} - \wh{G'}_{i+2b\cdot j})\\
    &= \sum_{j\in [n/(2b)]^d} \wh{x}_{i +
      2b\cdot j}\wh{G'}_{i+2b\cdot j}+ \sum_{j\in [n/(2b)]^d} \wh{x}_{i +
      2b\cdot j}(\wh{G}_{i+2b\cdot j} -\wh{G'}_{i+2b\cdot j}).
  \end{align*}
For the second term we have using Cauchy-Schwarz 
$$
\sum_{j\in [n/(2b)]^d} \wh{x}_{i + 2b\cdot j}(\wh{G}_{i+2b\cdot j} -\wh{G'}_{i+2b\cdot j})\leq ||x||_2 ||\wh{G} -\wh{G'}||_2\leq ||x||_2/N^c.
$$
For the first term we have 
$$
\sum_{j\in [n/(2b)]^d} \wh{x}_{i + 2b\cdot j}\wh{G'}_{i+2b\cdot j}=\wh{x}_i\cdot \wh{G'}_{i+2b\cdot 0}=\wh{x}_i
$$
for all $i\in [2b]^d$ such that $||i||_\infty\leq b$, since for any $j\neq 0$ the argument of $\wh{G'}_{i+2b\cdot j}$ is larger than $b$ in $\ell_\infty$ norm, and hence $\wh{G'}_{i+2b\cdot j}=0$ for all $j\neq 0$. 

Putting these bounds together we get that 
  \[
  \abs{\wh{x}'_i - \wh{x}_i} \leq \norm{2}{\wh{x}}
  \norm{2}{\wh{G}-\wh{G'}} \leq \norm{2}{x}N^{-c}
  \]
  as desired.

  The time complexity of computing the FFT of $y$ is $2^{O(d)}B\log B$. 
  The vector $y$ can  be constructed in time $c^{O(d)}||x||_0\log^d N$. This is because the support of $G_i$ is localized so that 
  each nonzero coordinate $i$ of $x$ only contributes to $c^{O(d)}\log^d N$
  entries of $y$. 
\end{proof}

We will need the following simple generalization:
\begin{corollary}\label{c:semiequi1}
  Let $n$ be a power of two, $N=n^d$, $c\geq 2$ a constant, and $\Sigma\in \gl, q\in \nsq$.  Let integer $B\geq 1$ be a power of $2^d$, $b=B^{1/d}$.
   Let $S = \{\Sigma(i - q) : i  \in \Z, ||i||_\infty \leq b/2\}$.  Then for any $x\in \C^\nsq$ we can compute $\wh{x}'_i$ for all
  $i \in S$ time such that
  \[
  \abs{\wh{x}'_i - \wh{x}_i} \leq \norm{2}{x} / N^c
  \]
in $c^{O(d)}||x||_0 \log^d N+2^{O(d)}B\log B$ time.  
\end{corollary}
\begin{proof}
Define
  $x^*_j = \omega^{qj}x_{\Sigma^{-T}j}$.  Then for all $i \in [n]$,
  \begin{align*}
    \wh{x}_{\Sigma(i-q)}
    &= \frac{1}{\sqrt{N}}\sum_{j \in \nsq} \omega^{-j^T\Sigma(i-q)}x_j\\
    &= \frac{1}{\sqrt{N}}\sum_{j \in \nsq} \omega^{-j^T\Sigma i}\omega^{j^T\Sigma q}x_j\\    
    &= \frac{1}{\sqrt{N}}\sum_{j'=\Sigma^{T}j \in \nsq} \omega^{-(j')^Ti}\omega^{(j')^Tq}x_{\Sigma^{-T}j'}\\    
    &= \frac{1}{\sqrt{N}}\sum_{j'=\Sigma^{T}j \in \nsq} \omega^{-(j')^Ti}x^*_{j'}\\        
    &= \wh{x}^*_i.
  \end{align*}
  We can access $\wh{x}^*_i$ with $O(d^2)$ overhead, so by
  Lemma~\ref{l:semiequi1} we can approximate $\wh{x}_{\Sigma(i-q)} =
  \wh{x}^*_i$ for $||i||_\infty \leq k$ in $c^{O(d)}k\log^d N$ time.
\end{proof}

\section{Acknowledgements}
The author would like to thank Piotr Indyk for many useful discussions at various stages of this work.


\newcommand{\etalchar}[1]{$^{#1}$}

\begin{appendix}
\section{Omitted proofs}\label{app:A}

\begin{proofof}{Lemma~\ref{lm:isolated-pi}}
We start with
\begin{equation}\label{eq:lixchbsp}
\begin{split}
\expect_{\Sigma, q}[|\pi(S\setminus \{i\})\cap \B^\infty_{(n/b)\cdot h(i)}((n/b)\cdot 2^t)|]&=\sum_{j\in S\setminus \{i\}} \prob_{\Sigma, q}[\pi(j)\in  \B^\infty_{(n/b)\cdot h(i)}((n/b)\cdot 2^t)]\\
\end{split}
\end{equation}
Recall that by definition of $h(i)$ one has $||(n/b)\cdot h(i)-\pi(i)||_\infty\leq (n/b)$, so by triangle inequality 
$$
||\pi(j)-\pi(i)||_\infty\leq ||\pi(j)-(n/b)h(i)||_\infty+||\pi(i)-(n/b)h(i)||_\infty,
$$
so
\begin{equation}\label{eq:lixchbs}
\begin{split}
\expect_{\Sigma, q}[|\pi(S\setminus \{i\})\cap \B^\infty_{(n/b)\cdot h(i)}((n/b)\cdot 2^t)|]&\leq \sum_{j\in S\setminus \{i\}} \prob_{\Sigma, q}[\pi(j)\in  \B^\infty_{\pi(i)}((n/b)\cdot (2^t+1))]\\
&\leq \sum_{j\in S\setminus \{i\}} \prob_{\Sigma, q}[\pi(j)\in  \B^\infty_{\pi(i)}((n/b)\cdot 2^{t+1})]\\
\end{split}
\end{equation}

Since $\pi_{\Sigma, q}(i)=\Sigma(i-q)$ for all $i\in \nsq$, we have
\begin{equation*}
\begin{split}
\prob_{\Sigma, q}[\pi(j)\in  \B^\infty_{\pi(i)}((n/b)\cdot 2^{t+1})]&=\prob_{\Sigma, q}[||\Sigma(j-i)||_\infty\leq (n/b)\cdot 2^{t+1}]\leq 2(2^{t+2}/b)^d,
\end{split}
\end{equation*}
where we used the fact that by Lemma~\ref{lemma:limitedindependence}, for any fixed $i$, $j\neq i$ and any radius $r\geq 0$,
\begin{equation}\label{eq:li}
 \prob_{\Sigma}[\norm{\infty}{\Sigma(i-j)} \leq r] \leq 2(2r/n)^d
\end{equation}
with $r=(n/b)\cdot 2^{t+1}$.

Putting this together with~\eqref{eq:lixchbs}, we get
\begin{equation*}
\begin{split}
\expect_{\Sigma, q}[|\pi(S\setminus \{i\})\cap \B^\infty_{(n/b)\cdot h(i)}((n/b)\cdot 2^t)|]\leq |S|\cdot 2(2^{t+2}/b)^{d}&\leq (|S|/B)\cdot 2^{(t+2)d+1}\\
&\leq  \frac1{4}\GI\cdot 64^{-(d+F)}\alpha^d 2^{(t+2)d+1}.
\end{split}
\end{equation*}
Now by Markov's inequality we have that $i$ fails to be isolated at scale $t$ with probability at most
$$
\prob_{\Sigma, q}\left[|\pi(S\setminus \{i\})\cap \B^\infty_{\pi(i)}((n/b)\cdot 2^t)|>\GI\cdot 64^{-(d+F)}\alpha^{d/2} 2^{(t+2)d+t+1} \right]\leq \frac1{4}2^{-t} \alpha^{d/2}.
$$
Taking the union bound over all $t\geq 0$, we get
$$
\prob_{\Sigma, q}[i~\text{is not isolated}]\leq \sum_{t\geq 0}\frac1{4}2^{-t} \alpha^{d/2} \leq \frac1{2}\alpha^{d/2}\leq \frac1{2}\alpha^{1/2}
$$
as required.

\end{proofof}

Before giving a proof of Lemma~\ref{lm:hashing}, we state the following lemma, which is immediate from Lemma~\ref{l:hashtobins}:
\begin{lemma}\label{lm:linearity}
Let $x, x^1, x^2, \chi, \chi^1, \chi^2\in \C^N$,  $x=x^1+x^2$, $\chi=\chi^1+\chi^2$.
Let $\Sigma\in \gl, q, a\in \nsq$, $B=b^d, b\geq 2$ an integer. Let
\begin{equation*}
\begin{split}
u&=\Call{HashToBins}{\hat x, \chi, (H, a)} \\
u^1&=\Call{HashToBins}{\widehat{x^1}, \chi^1, (H, a)} \\
u^2&=\Call{HashToBins}{\widehat{x^2}, \chi^2, (H, a)}.
\end{split}
\end{equation*}

 Then for each  $j\in [b]^d$ one has 
\begin{equation*}
\begin{split}
|G^{-1}_{o_i(i)} u_j\omega^{-a^T\Sigma i}-(x-\chi)_i|^p&\lesssim |G^{-1}_{o_i(i)} u^1_j\omega^{-a^T\Sigma i}-(x^1-\chi^1)_i|^p+|G^{-1}_{o_i(i)} u^2_j\omega^{-a^T\Sigma i}-(x^2-\chi^2)_i|^p\\
&+N^{-\Omega(c)}
\end{split}
\end{equation*}
for $p\in \{1, 2\}$, where $O(c)$  is the word precision of our semi-equispaced Fourier transform computations.
\end{lemma}

\begin{proofof}{Lemma~\ref{lm:hashing}}
  By Lemma~\ref{lemma:limitedindependence}, for any fixed $i$ and $j$ and any $t\geq 0$,
  \[
 \prob_{\Sigma}[\norm{\infty}{\Sigma(i-j)} \leq t] \leq 2(2t/n)^d.
  \]

 Per Lemma~\ref{l:hashtobins},
  \textsc{HashToBins} computes the vector $u \in \C^B$ given by
\begin{equation}\label{eq:u-delta}
  u_{h(i)} - \Delta_{h(i)} = \sum_{j\in \nsq} G_{o_i(j)}x'_j \omega^{a^T \Sigma j}
\end{equation}
  for some $\Delta$ with $\norm{\infty}{\Delta}^2 \leq  N^{-\Omega(c)}$.
  We define the vector $v \in \C^n$ by $v_{\Sigma j} = x'_j G_{o_i(j)}$, so that
  \[
  u_{h(i)} - \Delta_{h(i)} = \sum_{j\in \nsq} \omega^{a^Tj} v_j = \sqrt{N}\wh{v}_a
  \]
  so
  \[
  u_{h(i)} - \omega^{a^T\Sigma i}G_{o_i(i)}x'_i - \Delta_{h(i)} = \sqrt{N}(\wh{v_{\overline{\{\Sigma i\}}}})_a.
  \]
  We have by  \eqref{eq:u-delta} and the fact that $(X+Y)^2\leq 2X^2+2Y^2$ 
  \begin{equation*}
  \begin{split}
   \abs{G_{o_i(i)}^{-1}\omega^{-a^T\Sigma i}u_{h(i)} - x'_i}^2 =G_{o_i(i)}^{-2} \abs{u_{h(i)} - \omega^{a^T\Sigma i}G_{o_i(i)}x'_i}^2\\
    \leq 2G_{o_i(i)}^{-2} \abs{u_{h(i)} - \omega^{a^T\Sigma i}G_{o_i(i)}x'_i - \Delta_{h(i)}}^2 + 2G_{o_i(i)}^{-2}\Delta_{h(i)}^2\\    
    =2G_{o_i(i)}^{-2} \abs{\sum_{j\in \nsq} G_{o_i(j)}x'_j \omega^{a^T \Sigma j}}^2 +2G_{o_i(i)}^{-2} \Delta_{h(i)}^2\\        
  \end{split}
  \end{equation*}
  
  By Parseval's theorem, therefore, we have
  \begin{equation}\label{eq:a-est}
  \begin{split}
    \expect_a[\abs{G_{o_i(i)}^{-1}\omega^{-a^T\Sigma i}u_{h(i)} - x'_i}^2]
    &\leq 2G_{o_i(i)}^{-2} \expect_a[\abs{\sum_{j\in \nsq} G_{o_i(j)}x'_j \omega^{a^T \Sigma j}}^2] + 2\expect_a[\Delta_{h(i)}^2]\\
    &= 2G_{o_i(i)}^{-2} (\norm{2}{v_{\overline{\{\Sigma i\}}}}^2 + \Delta_{h(i)}^2)\\
    &\lesssim N^{-\Omega(c)} + \sum_{j\in \nsq \setminus \{i\}} \abs{x'_j G_{o_i(j)}}^2\\
    &\lesssim N^{-\Omega(c)} + \sum_{j \in \nsq \setminus \{i\}} \abs{x'_j G_{o_i(j)}}^2\\
    &\lesssim N^{-\Omega(c)} + \mu_{\Sigma, q}^2(i).\\
    \end{split}
  \end{equation}

We now prove {\bf (2)}.  Recall that the filter $G$ approximates an ideal filter, which would be $1$ inside $\B^\infty_{0}(n/b)$ and $0$ everywhere else. We use the bound on $G_{o_i(j)}=G_{\pi(i)-\pi(j)}$ in terms of $||\pi(i)-\pi(j)||_\infty$ from Lemma~\ref{lm:filter-prop}, (2). In order to leverage the bound, we partition $\nsq=\B^\infty_{(n/b)\cdot h(i)}(n/2)$ as 
$$
\B^\infty_{(n/b)\cdot h(i)}(n/2)=\B^\infty_{(n/b)\cdot h(i)}(n/b)\cup \bigcup_{t=1}^{\log_2 (b/2)}  \left(\B^\infty_{(n/b)\cdot h(i)}((n/b)2^{t})\setminus \B^\infty_{(n/b)\cdot h(i)}((n/b)2^{t-1})\right).
$$
For simplicity of notation, let $X_0=\B^\infty_{(n/b)\cdot h(i)}(n/b)$ and $X_t=\B^\infty_{(n/b)\cdot h(i)}((n/b)\cdot 2^{t})\setminus \B^\infty_{(n/b)\cdot h(i)}((n/b)\cdot 2^{t-1})$ for $t\geq 1$.
For each $t\geq 1$ we have by Lemma~\ref{lm:filter-prop}, (2) 
$$
\max_{\pi(l)\in X_t} |G_{o_i(l)}|\leq \max_{\pi(l)\not \in  \B^\infty_{(n/b)\cdot h(i)}((n/b)2^{t-1})} |G_{o_i(l)}| \leq \left(\frac{2}{1+2^{t-1}}\right)^{\fc}.
$$
Since the rhs is greater than $1$ for $t\leq 0$, we can use this bound for all $t\leq \log_2 (b/2)$.
Further, by Lemma~\ref{lemma:limitedindependence} we have for each $j\neq i$  and $t\geq 0$
$$
\prob_{\Sigma, q}[\pi(j)\in X_t]\leq \prob_{\Sigma, q}[\pi(j)\in \B^\infty_{(n/b)\cdot h(i)}((n/b)\cdot 2^{t})]\leq 2(2^{t+1}/b)^{d}.
$$

Putting these bounds together, we get
\begin{equation*}
\begin{split}
\expect_{\Sigma, q}[\mu^2_{\Sigma, q}(i)]&=\expect_{\Sigma, q}[\sum_{j \in \nsq\setminus \{i\}} \abs{x'_j G_{o_i(j)}}^2] \\
&\leq \sum_{j\in \nsq\setminus \{i\}}  \abs{x'_j}^2 \cdot \sum_{t=0}^{\log_2 (b/2)} \prob_{\Sigma, q}[\pi(j)\in X_t]\cdot \max_{\pi(l)\in X_t} |G_{o_i(l)}|\\
&\leq \sum_{j\in \nsq\setminus \{i\}}  \abs{x'_j}^2 \cdot \sum_{t=0}^{\log_2 (b/2)} (2^{t+1}/b)^{d}\cdot \left(\frac{2}{1+2^{t-1}}\right)^{\fc}\\
&\leq \frac{2^\fc}{B}\sum_{j\in \nsq\setminus \{i\}} \abs{x'_j}^2 \sum_{t=0}^{+\infty}  2^{(t+1)d-\fc (t-1)}\\
& \leq 2^{O(d)} \frac{\norm{2}{x'}^2}{B}
\end{split}
\end{equation*}
as long as $\fc\geq 2d$ and $\fc=\Theta(d)$. Recalling that $G_{o_i(i)}^{-1}\leq \GS$ completes the proof of {\bf (2)}.

The proof of {\bf (1)} is similar. We  have
\begin{equation}
\begin{split}
\expect_{\Sigma, q}[\max_{a\in \nsq}  |\sum_{j \in \nsq\setminus \{i\}} x'_j G_{o_i(j)} \omega^{a^T\Sigma j}|] &\leq \expect_{\Sigma, q}[\sum_{j \in \nsq\setminus \{i\}} |x'_j G_{o_i(j)}|]+|\Delta_{h(i)}|\\
&\leq |\Delta_{h(i)}|+\sum_{j\in \nsq\setminus \{i\}}  \abs{x'_j} \cdot \sum_{t=0}^{\log_2 (b/2)} \prob_{\Sigma, q}[\pi(j)\in X_t]\cdot \max_{\pi(l)\in X_t} |G_{o_i(l)}|\\
&\leq |\Delta_{h(i)}|+\sum_{j\in \nsq\setminus \{i\}}  \abs{x'_j} \cdot \sum_{t=0}^{\log_2 (b/2)} (2^{t+1}/b)^{d}\cdot \left(\frac{2}{1+2^{t-1}}\right)^{\fc}\\
&\leq |\Delta_{h(i)}|+\frac{2^\fc}{B}\sum_{j\in \nsq\setminus \{i\}} \abs{x'_j} \sum_{t=0}^{+\infty}  2^{(t+1)d-\fc (t-1)}\\
& \leq |\Delta_{h(i)}|+2^{O(d)} \frac{\norm{1}{x'}}{B},\\
\end{split}
\end{equation}
where 
\begin{equation*}
\Delta_{h(i)} \lesssim N^{-\Omega(c)}.
\end{equation*}
Recalling that $G_{o_i(i)}^{-1}\leq \GS$ and  $R^*\leq ||x||_\infty/\mu$ 
completes the proof of {\bf (1)}.

\end{proofof}

\end{appendix}

\end{document}